\documentclass[a4paper,onecolumn,11pt,accepted=2025-12-04]{quantumarticle}
\pdfoutput=1

\usepackage[utf8]{inputenc}
\usepackage[english]{babel}
\usepackage{bbold}
\usepackage{amsmath}
\usepackage{amsthm}
\usepackage{csquotes}
\usepackage{enumitem}
\usepackage{fullpage}
\usepackage[caption=false,subrefformat=parens,labelformat=parens]{subfig}
\usepackage{graphicx}
\usepackage[dvipsnames]{xcolor}
\usepackage{tabularx}
\usepackage{tikz}
\usetikzlibrary{shapes.geometric,plotmarks,backgrounds,fit,positioning,matrix,calc,circuits.logic.US}
\usepackage{hyperref}
\hypersetup{
    colorlinks = true,
    citecolor=blue,
    filecolor=blue,
    linkcolor=blue,
    urlcolor=blue
}
\usepackage{mathtools}
\usepackage{framed} 
\definecolor{shadecolor}{gray}{0.9}

\usepackage{scrextend} 

\usepackage{float} 
\usepackage{physics}
\usepackage{cleveref}

\usepackage{amsmath}
\usepackage{amsfonts}
\usepackage{amssymb}
\usepackage{IEEEtrantools}
\usepackage{amsthm}
\usepackage{cases}
\usepackage{stmaryrd}
\usepackage{dsfont}
\usepackage{braket}
\usepackage{algorithm}
\usepackage{algpseudocode}
\usepackage{bm}
\usepackage{longtable}

\usepackage{booktabs}
\usepackage{comment}
\usepackage{placeins}
\usepackage{bm}

\usepackage[affil-it]{authblk}

\usepackage[numbers]{natbib}

\newtheorem{definition}{Definition}
\newtheorem{prop}{Proposition}
\newtheorem{theorem}{Theorem}
\newtheorem{assumption}{Assumption}

\newtheorem{lemma}{Lemma}

\theoremstyle{definition}

\newtheorem{remark}{Remark}

  \usepackage{longtable}

\newcommand{\tA}{\widetilde A}
\newcommand{\tB}{\widetilde B} 

\newcommand{\vect}[1]{\bm{#1}}
\DeclareMathOperator*{\maximize}{maximize}

\DeclareMathOperator{\aff}{aff}
\DeclareMathOperator{\relint}{relint}
\DeclareMathOperator{\pos}{Pos}

\newcommand{\q}{\vect{q}} 
\DeclareMathOperator{\Q}{\mathcal{Q}} 
\DeclareMathOperator{\M}{ \mathcal{M}}
\DeclareMathOperator{\id}{ \mathcal{I} }
\DeclareMathOperator{\idm}{ \mathds{1} }
\DeclareMathOperator{\DM}{D}
\DeclareMathOperator{\dom}{dom}
\newcommand{\leak}{\text{\rm leak}}
\newcommand{\hash}{\text{\rm hash}}
\newcommand{\cA}{\mathcal{A}}
\newcommand{\cB}{\mathcal{B}}
\newcommand{\cE}{\mathcal{E}}
\newcommand{\cF}{\mathcal{F}}
\newcommand{\primal}{\text{primal}}
\newcommand{\dual}{\text{dual}}
\newcommand{\opt}{\text{\rm opt}}
\newcommand{\lin}{\text{\rm lin}}
\newcommand{\gen}{\text{\rm gen}}

\newcommand*{\freq}[1]{\mathsf{freq}(#1)}
\newcommand*{\cM}{\mathcal{M}}
\newcommand*{\Max}[1]{\mathsf{Max}(#1)}
\newcommand*{\MinSigma}[1]{\mathsf{Min}_{\Sigma}(#1)}
\newcommand*{\Min}[1]{\mathsf{Min}(#1)}
\newcommand*{\Var}[1]{\mathsf{Var}(#1)}
\newcommand*{\cX}{\mathcal{X}}

\newcommand{\mbP}{\mathbb{P}}
\newcommand{\mbN}{\mathbb{N}}
\newcommand{\mbR}{\mathbb{R}}
\newcommand*{\cI}{\mathcal{I}}
\newcommand*{\cT}{\mathcal{T}}

\newcommand*{\cZ}{\mathcal{Z}}
\newcommand*{\cK}{\mathcal{K}}
\newcommand{\cL}{\mathcal{L}}
\newcommand{\cQ}{\mathcal{Q}}
\newcommand{\ve}{\varepsilon}
\newcommand{\cS}{\mathcal{S}}
\newcommand{\cP}{\mathcal{P}}

\newcommand{\PP}{{{\mathbb P}}}
\newcommand{\<}{\langle}
\renewcommand{\>}{\rangle}
\usepackage{accents}

\newcommand{\ox}{\otimes}

\newcommand{\cD}{\mathcal{D}}

\newcommand{\testround}{\text{\rm test}}
\newcommand{\genround}{\text{\rm gen}}

\newtheorem{corollary}[theorem]{Corollary}
\newtheorem{proposition}[theorem]{Proposition}
\newtheoremstyle{named}%
    {}{}{\itshape}{}{\bfseries}{:}{.5em}{\thmnote{#3}}
\theoremstyle{named}
\newtheorem*{assumption*}{Structure Assumption}

\newcommand{\EC}{\text{EC}}
\newcommand{\vPE}{\varepsilon_\text{\rm PE}}
\newcommand{\vPA}{\varepsilon_\text{\rm PA}}
\newcommand{\vEC}{\varepsilon_\text{\rm EC}}
\newcommand{\vSec}{\varepsilon_{\text{\rm sec}}}
\newcommand{\vBar}{\bar{\varepsilon}}
\newcommand{\vEA}{\varepsilon_{\text{\rm EA}}}
\newcommand{\vAcc}{\varepsilon_{\text{\rm{acc}}}}
\newcommand{\vECRob}{\varepsilon_{\text{\rm{EC,rob}}}}

\newcommand{\ol}{\overline}
\newcommand{\wt}{\widetilde}
\newcommand{\Density}{\mathrm{D}}

\DeclareMathOperator{\st}{subject\ to\ }
\DeclareMathOperator{\fa}{,\ for\ all\ }
\newcommand{\minimize}{ \mathop{\mathrm{minimize}}\limits }

\makeatletter
\newenvironment{breakablealgorithm2}[1][htb]
{
	\begin{flushleft}
		\refstepcounter{protocol}
		\hrule height.8pt depth0pt \kern2pt
		\renewcommand{\caption}[2][\relax]{
			{\raggedright\textbf{\fname@algorithm~\thealgorithm} ##2\par}%
			\ifx\relax##1\relax 
			\addcontentsline{loa}{algorithm}{\protect\numberline{\thealgorithm}##2}%
			\else 
			\addcontentsline{loa}{algorithm}{\protect\numberline{\thealgorithm}##1}%
			\fi
			\kern2pt\hrule\kern2pt
		}
	}{
		\kern2pt\hrule\relax
	\end{flushleft}
}
\makeatother

\newcounter{protocol}
\makeatletter
\newenvironment{protocol}{%
	\renewcommand{\ALG@name}{Protocol}
		\let\c@algorithm\c@protocol
	\begin{breakablealgorithm2}%
	}{\end{breakablealgorithm2}
}

\makeatother

\newcommand{\cH}{\mathcal{H}}
\newenvironment{aeq}{\begin{equation}
		\begin{aligned}
		}{
		\end{aligned}
\end{equation}}

\crefname{thm}{Theorem}{Theorems}
\Crefname{thm}{Theorem}{Theorems}
\crefname{assump}{Assumption}{Assumptions}
\Crefname{assump}{Assumption}{Assumptions}
\crefname{problem}{Problem}{Problems}
\Crefname{problem}{Problem}{Problems}
\crefname{conjecture}{Conjecture}{Conjectures}
\Crefname{conjecture}{Conjecture}{Conjectures}
\crefname{proposition}{Proposition}{Propositions}
\Crefname{proposition}{Proposition}{Propositions}
\crefname{prop}{Proposition}{Propositions}
\Crefname{prop}{Proposition}{Propositions}
\crefname{corr}{Corollary}{Corollaries}
\Crefname{corr}{Corollary}{Corollaries}
\crefname{lem}{Lemma}{Lemmas}
\Crefname{lem}{Lemma}{Lemmas}
\theoremstyle{definition}
\crefname{definition}{definition}{definitions}
\Crefname{definition}{Definition}{Definitions}
\crefname{defn}{definition}{definitions}
\Crefname{defn}{Definition}{Definitions}
\crefname{remark}{Remark}{Remarks}
\Crefname{remark}{Remark}{Remarks}
\crefname{rmk}{Remark}{Remarks}
\Crefname{rmk}{Remark}{Remarks}
\crefname{example}{Example}{Examples}
\Crefname{example}{Example}{Examples}
\crefname{align}{}{}
\Crefname{align}{}{}
\crefname{equation}{Eq.}{Eqs.}
\Crefname{equation}{Eq.}{Eqs.}
\crefname{ineq}{Ineq.}{Ineqs.}
\Crefname{ineq}{Ineq.}{Ineqs.}
\crefname{protocol}{protocol}{protocols}
\Crefname{protocol}{Protocol}{Protocols}
\crefname{algorithm}{algorithm}{algorithms}
\Crefname{algorithm}{Algorithm}{Algorithms}
\crefname{step}{Step}{Steps}
\Crefname{step}{Step}{Steps}
\crefname{item}{}{}
\Crefname{item}{}{}
\crefname{condition}{condition}{conditions}
\Crefname{condition}{Condition}{Conditions}
\creflabelformat{ineq}{~\upshape(#2#1#3)}

\makeatletter
\newenvironment{breakablealgorithm}
{
	\begin{flushleft}
		\refstepcounter{algorithm}
		\hrule height.8pt depth0pt \kern2pt
		\renewcommand{\caption}[2][\relax]{
			{\raggedright\textbf{\fname@algorithm~\thealgorithm} ##2\par}%
			\ifx\relax##1\relax 
			\addcontentsline{loa}{algorithm}{\protect\numberline{\thealgorithm}##2}%
			\else 
			\addcontentsline{loa}{algorithm}{\protect\numberline{\thealgorithm}##1}%
			\fi
			\kern2pt\hrule\kern2pt
		}
	}{
		\kern2pt\hrule\relax
	\end{flushleft}
}
\makeatother

\begin{document}
	
	\title{Finite-Key Analysis of Quantum Key Distribution with Characterized Devices Using Entropy Accumulation}

	\author{Ian George$^{*}$}
		\affiliation{Institute for Quantum Computing and Department of Physics and Astronomy, University of Waterloo, Waterloo, Ontario N2L 3G1, Canada}
	\affiliation{Department of Electrical \& Computer Engineering, University of Illinois, Urbana, Illinois 61801, USA}
	\orcid{0000-0002-2803-2421}
	\author{Jie Lin$^{*}$}
	\affiliation{Institute for Quantum Computing and Department of Physics and Astronomy, University of Waterloo, Waterloo, Ontario N2L 3G1, Canada}
	\orcid{0000-0003-1750-2659}
		\author{Thomas van Himbeeck$^{*}$}
	\affiliation{Institute for Quantum Computing and Department of Physics and Astronomy, University of Waterloo, Waterloo, Ontario N2L 3G1, Canada}
	\affiliation{Department of Electrical \& Computer Engineering, University of Toronto, Toronto, Ontario M5S 3G4, Canada}
\author{Kun Fang}
		\affiliation{Institute for Quantum Computing and Department of Physics and Astronomy, University of Waterloo, Waterloo, Ontario N2L 3G1, Canada}
 \affiliation{School of Data Science, The Chinese University of Hong Kong, Shenzhen, Guangdong, 518172, China}
	
	\author{Norbert L{\"u}tkenhaus}
	\affiliation{Institute for Quantum Computing and Department of Physics and Astronomy, University of Waterloo, Waterloo, Ontario N2L 3G1, Canada}
	\orcid{0000-0002-4897-3376}
	
	\def\thefootnote{*}\footnotetext{These authors contributed equally to this work}\def\thefootnote{\arabic{footnote}}

	\maketitle

	\begin{abstract}
	  The Entropy Accumulation Theorem (EAT) was introduced to significantly improve the finite-size rates for device-independent quantum information processing tasks such as device-independent quantum key distribution (QKD). A natural question would be whether it also improves the rates for device-dependent QKD. In this work, we provide an affirmative answer to this question. We present new tools for applying the EAT in the device-dependent setting. We present sufficient conditions for the Markov chain conditions to hold as well as general algorithms for constructing the needed min-tradeoff function. Utilizing Dupuis' recent privacy amplification without smoothing result, we improve the key rate by optimizing the sandwiched R\'{e}nyi entropy directly rather than considering the traditional smooth min-entropy. We exemplify these new tools by considering several examples including the BB84 protocol with the qubit-based version and with a realistic parametric downconversion source, the six-state four-state protocol and a high-dimensional analog of the BB84 protocol.
	\end{abstract}
	\tableofcontents
\section{Introduction}
Quantum key distribution (QKD) protocols \cite{Bennett1984, Ekert1991}  enable two legitimate parties (Alice and Bob) to establish shared secret keys in the presence of any eavesdropper (Eve) who has control over the channel connecting the two parties and is only limited by the laws of quantum mechanics. That is, QKD protocols establish an information-theoretically secure shared secret key. The generated secret keys can be used in many cryptographic applications that our modern cybersecurity relies on, such as encryption and authentication. The field of QKD has grown rapidly in the last couple decades. Novel protocols have been proposed and analyzed (see \cite{Scarani2009, Xu2020,Pirandola2020} for reviews). Many QKD experiments (e.g. \cite{Boaron2018, Fang2020}) have been demonstrated to reach increasingly longer distances and/or to achieve higher key bits per second. The maturity of the field can be witnessed by development of commercial prototypes in several companies, successful launch of a QKD satellite in China \cite{Liao2017}, many ongoing efforts to launch QKD satellites from all over the world \cite{Bedington2017}, and developments of chip-based QKD systems \cite{Sibson2017,Zhang2019,Wei2020}.

While currently the field of implementation security (as opposed to protocol security), requiring certification procedures and security proofs with more refined models, captures an increasing amount of attention, there are still challenges to be addressed on the level of protocol security itself.  One major technical difficulty is guaranteeing security of the protocol when using finite resources. The security of a novel protocol is often first proved in the asymptotic key limit where Alice and Bob exchange infinitely many signals, since in that regime typically reduction to the identically and independently distributed (i.i.d.) assumption about the individual exchanged signals can be made and statistical fluctuations can be neglected. However, any realistic system can only exchange a finite number of signals. For example, the time interval of low orbiting satellites having line of sight to ground stations is limited to minutes, thus naturally limiting the total number of signals that can be exchanged. The study of security with finite resources is known as finite key analysis. In recent years, finite key analysis against general attacks has been performed for several QKD protocols. One aims to provide security proofs that allow one to generate the most key using the fewest signals and converge to the fundamental asymptotic limit quickly.

Over the past decade, several techniques have been proposed to bridge the gap between asymptotic security and finite key analysis against general attacks. The Quantum de Finetti theorem \cite{Renner2005} was first used to achieve this task. While it makes no restrictions on the QKD protocol implemented,\footnote{The Quantum de Finetti theorem requires permutation invariance, which in principle entails applying a coordinated random permutation on the state by Alice and Bob. However, it was shown in \cite{Scarani2008b} that the min-entropy after symmetrization lower bounds the original, so the \textit{implemented} protocol does not need the symmetry enforced.} the key rate becomes very pessimistic at practical signal block size. The postselection technique \cite{Christandl2009} was later presented to improve the key rate if the implemented protocol is permutation invariant with respect to any input. Although it gives better performance, the key rate is known to in general remain pessimistic as it scales exponentially in the dimension of Alice and Bob's signals. More recently, the entropy accumulation theorem \cite{Dupuis2016, Dupuis2019} was developed to offer tighter finite key analysis. It has been successfully applied to prove the security of device-independent (DI) QKD protocols \cite{Arnon2018,Arnon2019}.

From the perspective of implementation security, DIQKD is favorable since it requires minimal assumptions on Alice and Bob's devices and thus is immune to most side-channel attacks. Thanks to the entropy accumulation theorem \cite{Dupuis2016} and its application to DIQKD \cite{Arnon2018}, the experimental requirements to implement DIQKD have become within reach of state-of-the-art technologies. Very recently, proof-of-principle experimental demonstrations of DIQKD with positive key rates are reported for extremely short distances \cite{Nadlinger2021,Zhang2021,Liu2021}. Even if the transmission distance is extended to practical regimes in the near future, a DIQKD system is expected to be more costly since the experimental requirements are still demanding. Beyond the experimental requirements, unlike device-dependent QKD, DIQKD is currently not known to be universally composably secure \cite{Barrett13,Arnon2019}. This means currently this property must either be sacrificed or a new set of devices must be used each time one establishes a key. As such, while DIQKD has advantages for implementation security and can be implemented, at least over very short distances, device-dependent QKD (i.e. QKD with characterized devices) remains a practical option with different theoretical security guarantees at the moment.\footnote{We note that measurement-device-independent (MDI) QKD also falls in the category of device-dependent QKD since one still makes assumptions on both Alice's and Bob's devices.}

Beyond these limitations, as less assumptions are put on trusted parties' devices, DIQKD key rates are unavoidably lower than a standard device-dependent QKD protocol. On the other hand, with the best practice to countermeasure known side-channel attacks, suitable risk management and use cases, and efforts from standardization, device-dependent QKD remains a vital player for a wide adoption of QKD in the future. Therefore, it is still of great interest to provide as tight as possible key rates for device-dependent QKD. Our goal here is to apply the entropy accumulation theorem to provide better finite key rates of device-dependent QKD. 

While concurrent work used the generalised entropy accumulation theorem \cite{Metgers-2022a} to establish the security of device-dependent QKD protocols, with a focus on prepare-and-measure protocols when only one signal is in the physical channel at a time \cite{Metgers-2022b}, the original entropy accumulation theorem has not been applied to determine the key rate for any device-dependent QKD protocol.\footnote{In \cite{Dupuis2016} the authors presented how the asymptotic rate was recovered for BB84, but did not present finite-size key lengths.} In this work we focus on applying the original EAT to the analysis to the most general class of entanglement-based protocols we can.\footnote{For entanglement-based protocols, there is no real advantage to using the generalised EAT and so this is largely without loss of generality.} There are two major obstacles to overcome in applying the original EAT to device-dependent QKD. The first is to guarantee the Markov condition is satisfied by the protocol. The second challenge is to have a general method to construct the required min-tradeoff function. In this work we address both challenges.

The Markov chain condition in effect guarantees the protocol is well-behaved in what side-information is leaked to Eve during the protocol. This is an important aspect as device-dependent protocols often make use of announcing more information about each signal than DI protocols. In this work, we provide sufficient conditions for satisfying the Markov conditions of the EAT. Specifically, we prove that when Alice and Bob's announcements are determined by POVMs with a simple block-diagonal structure, then the protocol satisfies the Markov chain conditions (see \Cref{subsec:restrictionsOnAnnouncements}). This block-diagonal structure is satisfied for many practical discrete-variable QKD protocols. Moreover, we note that many practical non-sequential protocols will satisfy the Markov chain conditions necessary, but as they aren't sequential, don't satisfy the conditions of \cite{Metgers-2022b}.

The challenge of developing a general method for constructing a min-tradeoff function has been addressed for the device-independent scenario in a recent work \cite{Brown2020}. In concurrent work \cite{Metgers-2022b}, the construction of a min-tradeoff function has been reduced to asymptotic key rate solvers in the device-dependent setting under the assumption only a single signal is ever in the physical channel at a time. This is analagous to the first algorithm of our numerical method for constructing tight min-tradeoff functions for device-dependent QKD protocols although we only recover the Markov chain conditoins be satisfied. The basic idea of this first algorithm is similar to our existing numerical method for asymptotic key rates \cite{Coles2016, Winick2018}. The second algorithm improves on the first one by considering the second-order term in the original EAT. It is based on Fenchel duality. However, unlike the formulation of the objective function in Refs. \cite{Coles2016, Winick2018}, we present here a different formulation of the objective function for the optimization which has advantages in terms of dimension needed for the numerical optimization as well as a simpler representation of procedures to handle classical postprocessing.

Using the guarantee of our numerical method constructing a min-tradeoff function, we present a general key length bound. It is largely similar to that of Ref. \cite{Arnon2018}, but, building on the recent result that privacy amplification may be characterized using sandwiched R\'{e}nyi entropies rather than the smooth min-entropy \cite{Dupuis2021}, we are able to show we can improve the key rate by optimizing over sandwiched R\'{e}nyi entropies rather than using the smooth min-entropy bounds. To the best of our knowledge, the only work to use R\'{e}nyi entropies in establishing the security of QKD is \cite{Abruzzo_2011}. However, we note that \cite{Abruzzo_2011} does so by considering the smoothed $2$-R\'{e}nyi entropy. This differs from the ability to analyze the key rate of R\'{e}nyi entropies directly by establishing a security proof that makes no use of smoothed entropic quantities as is done in this work.\footnote{We remark that the smoothed $2$-R\'{e}nyi entropy was also considered for analyzing the large deviation regime of privacy amplification by Hayashi \cite{Hayashi-2014a}.}

Taking all of these tools together, we apply our method to several examples. The first example is a simple qubit-based BB84 protocol and the second example is the entanglement-based six-state four-state protocol \cite{Tannous2019}. We use these two examples to demonstrate behaviors of our algorithms. The third one is a high-dimensional analog of the BB84 protocols. We compare these results with the key rates obtained by the postselection technique \cite{Christandl2009} to show that the EAT can outperform the postselection technique for high-dimensional signals. In the last example, we demonstrate that our method can also work for optical implementations of QKD and in particular, we study the entanglement-based BB84 protocol with a spontaneous parametric downconversion source in a lossy and noisy channel.

Our work here provides a first step along the direction of applying EAT to general device-dependent QKD protocols. Currently, we limit ourselves to entanglement-based protocols in exchange for the generality of protocol which we can consider. To handle prepare-and-measure protocols, one needs to be able to incorporate the source-replacement scheme \cite{Ferenczi2012}, but it remains unclear how to do this while satisfying the Markov chain conditions of the original EAT. In \cite{Metgers-2022b}, source-replacement is used by circumventing the Markov chain conditions by using the generalised EAT and assuming truly sequential implementation of the rounds in the protocol. As noted previously, the demand of a truly sequential implementation is not practical in general. It remains a technical challenge how to relax this condition and prove security for prepare and measure protocols.

The rest of the paper is organized as follows. In \Cref{sec:notation}, we summarize our notational conventions and in \Cref{sec:securitydefinition}, we review the security definition of QKD. In \Cref{sec:EATBackground}, we review the entropy accumulation theorem and adapt the theorem to the device-dependent setting. We do this first in the case where side information is all seeded randomness, as was the case in previous works. We then provide sufficient conditions for satisfying the Markov chain conditions necessary to apply the EAT without relying on all side-information being randomly seeded. This is important as many device-dependent QKD protocols make announcements which depend on the outcome of the measurement. In \Cref{sec:QKD_protocol}, we describe the generic protocol to which our method is applicable, protocols where all quantum signals are exchanged prior to any classical communication and all communication acts on a single `round' of quantum communication. We also provide the assumptions on public announcements. In \Cref{sec:SecurityofDDQKDfromEAT}, we then apply the EAT to this family of protocols and obtain a key length formula. In \Cref{sec:construction_min_tradeoff}, we discuss how to construct min-tradeoff functions numerically and present two numerical algorithms for obtaining a nearly tight min-tradeoff function. In \Cref{sec:example}, we apply our method to several examples. We then conclude in \Cref{sec:conclusions}. We leave technical details to appendixes. 

\section{Preliminaries}\label{sec:preliminaries}
We discuss the notational convention and some useful definitions in \Cref{sec:notation}. We then review the security definition of QKD in \Cref{sec:securitydefinition}. 
\subsection{Notation}\label{sec:notation}
We briefly summarize the notation which we use throughout this work.

\begin{longtable}{l|l} 
\toprule[2pt]
General notations & Descriptions \\
\hline
$A,B,C,\cdots$ & Quantum systems and their associated Hilbert spaces\\
$|A|,|B|,|C|,\cdots$ & Dimension of the Hilbert spaces\\
$A^n_1$ & Shorthand for $A_1A_2 \cdots A_n$\\
$[n]$ & The set of natural numbers from 1 to $n$\\
$\id$ & Identity map\\
$f$ & A min-tradeoff function\\
$\vect{f}$ & A vector that is used to construct the min-tradeoff function $f$\\
$\vect{1}$ & A vector of all ones\\
\midrule[2pt]
EAT statement & Descriptions \\
\hline
$P_i$ & Public information in the EAT statement\\
$S_i$ & Secret information in the EAT statement\\
$T_i$ & Test flag variable\\
$X_i$ & Test result register in the EAT statement\\
$E$ & Eve's system\\
$\Omega$ & An event\\
$\rho[\Omega]$ & Probability of $\Omega \subset \cX$ for $\rho_{XA}$ (See main text in \Cref{sec:notation}) \\
$\freq$ & Frequency of a given event\\
$\q$ & A probability vector\\
$\PP(\cX)$ & The probability vector space with events from $\cX$\\
$\mathrm{D}(A)$ & The set of quantum states on Hilbert space $A$\\
 \bottomrule[2pt]  
\caption{\small Overview of notational conventions.}
\label{tab: notations}
\end{longtable}

For consistency, we at times use the notation for quantum states conditioned on an event from previous works \cite{Dupuis2016,Dupuis2019}. Let $\cX$ be a finite alphabet, $X \cong \mathbb{C}^{|\cX|}$. Then a classical-quantum (CQ) state may be written as $\rho_{XA} = \sum_{x \in \cX} \dyad{x} \otimes \rho_{A,x}$ where the quantum part of this decomposition is generally called the conditional state. Given an event $\Omega \subset \cX$, the probability of the event is $\rho[\Omega]$, which may be calculated by $\rho[\Omega]:= \sum_{x \in \Omega} \Tr[\rho_{A,x}]$. Lastly, the state conditioned on the event is given by $\rho_{XA|\Omega} = \frac{1}{\rho[\Omega]} \sum_{x \in \Omega} \rho_{A,x}$. Note that a conditioned state is therefore re-normalized.

We will also require various entropies which we also introduce here. We use the notation of \cite{Dupuis2016,Dupuis2019}. We refer to \cite{Tomamichel2016} for further background, but note that the notation differs in that work.

For $\alpha \in (0,1) \cup (1, \infty)$, define the minimal divergence by 
\begin{align*}
	D_{\alpha}(\rho||\sigma) := 
	\begin{cases}
		\frac{1}{\alpha-1} \log \frac{\|\sigma^{\frac{1-\alpha}{2\alpha}} \rho \sigma^{\frac{1-\alpha}{2\alpha}} \|^{\alpha}_{\alpha}}{\Tr(\rho)} & \text{ if }\alpha <1 \text{ or if } \alpha > 1\text{ and } \mathrm{supp}(\rho) \subseteq \mathrm{supp}(\sigma) \\ 
		+\infty & \text{ otherwise}
	\end{cases}\ , 
\end{align*}where $\norm{\cdot}_p$ is the Schatten $p$-norm for $p \geq 1$ and the definition is extended to all $p >0$. 
For any $\rho \in \mathrm{D}(A\otimes B)$, $\sigma \in \mathrm{D}(B)$, define  $$H_{\alpha}(\rho_{AB}||\sigma_{B}) := -D_{\alpha}(\rho_{AB}||\mathbb{1}_{A} \otimes \sigma_{B}) \, . $$
We then define
$H_{\alpha}(A|B)_{\rho_{AB}} := H_{\alpha}(\rho_{AB}||\rho_{B})$ and $H^{\uparrow}_{\alpha}(A|B)_{\rho} := \sup_{\sigma_{B}} -D_{\alpha}(\rho_{AB}|| \mathbb{1}_{A} \otimes \sigma_{B})$. We refer to both these classes of entropies as sandwiched R\'{e}nyi entropies. We may also recall the smooth min-entropy defined by 
$$ H^{\varepsilon}_{\min}(A|B)_{\rho_{AB}} := \underset{\tilde{\rho} \in \mathcal{B}^{\epsilon}(\rho_{AB})}{\max} H_{\min}(A|B)_{\tilde{\rho}} \ , $$
where $H_{\min}(A|B)_{\rho_{AB}} := -\log(\min \{\Tr(Y): \rho_{AB} \leq \mathbb{1}_{A} \otimes Y_{B} \ \})$, $\mathcal{B}^{\epsilon}(\rho_{AB}) := \{\tilde{\rho}_{AB} \geq 0 : \Tr(\tilde{\rho}_{AB}) \leq 1 \, \& \, P(\tilde{\rho}_{AB},\rho_{AB}) \leq \varepsilon \}$, and $P$ is the purified distance \cite{Tomamichel2012thesis,Tomamichel2016}. Throughout this work, $\log$ is referred to $\log_2$.

\subsection{QKD security definition}\label{sec:securitydefinition}
We provide a short review of the $\ve$-security framework of QKD \cite{Renner2005,Portmann2014}. A QKD protocol is $\ve$-secure if for any input state, the output state $\rho_{S_AS_BE}$ conditioned on that the protocol does not abort (and thus subnormalized) satisfies
\begin{aeq}\label{eq:security_def_trace_distance}
\frac{1}{2}\norm{\rho_{S_AS_BE}-\pi_{S_AS_B}\ox \rho_E}_1 \leq \ve \, ,
\end{aeq}where $\pi_{S_AS_B} = \sum_{s \in \cS}\frac{1}{\abs{\cS}}\dyad{s}\ox\dyad{s}$, and $\cS$ is the space of keys that could be generated from the protocol. The security parameter $\ve$ quantifies the amount of deviation of the real protocol from an ideal protocol. In an ideal QKD protocol, Alice and Bob are supposed to obtain an identical key, which is the correctness requirement; The key is supposed to be distributed from a uniform distribution among all possible keys and that Eve knows no information about the key, which is the secrecy requirement. This security definition in terms of trace distance has an operational interpretation: If a distinguisher is given either the real or the ideal protocol as a black box with an equal \textit{a priori} probability and the goal is to verify which protocol the black box implements, then the probability that this distinguisher can guess correctly by looking at output states of the black box is at most $\frac{1}{2}(1+\ve)$. We note that in the case of aborting, both protocols output a trivial key symbol.    

In \Cref{eq:security_def_trace_distance}, we use subnormalized states. Equivalently, if we define in terms of normalized output states $\tilde{\rho}_{S_AS_BE} = \rho_{S_AS_BE}/\Tr(\rho_{S_AS_BE})$ and $\tilde{\rho}_E = \rho_E/\Tr(\rho_E)$ where $\Tr(\rho_{S_AS_BE}) = \Tr(\rho_E):=\Pr(\text{accept})$ denotes the probability that the protocol does not abort, then \Cref{eq:security_def_trace_distance} can be written as
\begin{equation}\label{eq:security-scaled-by-accept}
\frac{1}{2}\Pr(\text{accept})\norm{\tilde{\rho}_{S_AS_BE}-\pi_{S_AS_B}\ox \tilde{\rho}_E}_1 \leq \ve \ .
\end{equation}

It is often convenient to discuss the secrecy requirement and correctness requirement separately since the correctness requirement is usually guaranteed by the error correction and error verification steps of a QKD protocol. A key is $\vSec$-secret if 
\begin{aeq}\label{eq:secrecy}
\frac{1}{2}\norm{\rho_{S_AE}-\pi_{S_A}\ox \rho_E}_1 \leq \vSec \, ,
\end{aeq}where $\pi_{S_A} = \sum_{s \in \cS}\frac{1}{\abs{\cS}}\dyad{s}$. A QKD protocol is $\ve_{\text{cor}}$-correct if the joint probability that the protocol does not abort and that Bob's key is different from Alice's key is at most $\ve_{\text{cor}}$. By the triangle inequality of the trace norm, it is easy to see that if the QKD protocol is $\ve_{\text{cor}}$-correct and it generates $\vSec$-secret keys, then the protocol is $\ve$-secure with $\ve=\vSec +\ve_{\text{cor}}$.

\section{The Entropy Accumulation Theorem}\label{sec:EATBackground}
	\FloatBarrier
	
	\begin{figure}
		\begin{center}
			\begin{tikzpicture}[thick]
				\tikzstyle{porte} = [draw=black!50, fill=black!20]
				\draw
				++(1, 0) node[porte] (m1) {$\cM_1$}
				++(3, 0) node[porte] (m2) {$\cM_2$}
				++(2, 0) node (dotdotdot) {$\cdots$}
				++(2, 0) node[porte] (mn) {$\cM_n$}
				(m1) ++(-1, -1.2) node (a1) {$S_1$}
				(m1) ++(1, -1.2) node (b1) {$P_1$}
				(m1) ++(0, -2.3) node (x1) {$X_1$}
				(m2) ++(-1, -1.2) node (a2) {$S_2$}
				(m2) ++(1, -1.2) node (b2) {$P_2$}
				(m2) ++(0, -2.3) node (x2) {$X_2$}
				(mn) ++(-1, -1.2) node (an) {$S_n$}
				(mn) ++(1, -1.2) node (bn) {$P_n$}
				(mn) ++(0, -2.3) node (xn) {$X_n$}
				(x2) ++(0,-1) node {(a)  EAT Process.}
				;
				\draw 
				(m1) edge[->] (a1)
				(m1) edge[->] (b1)
				(m2) edge[->] (a2)
				(m2) edge[->] (b2)
				(mn) edge[->] (an)
				(mn) edge[->] (bn)
				(a1) edge[->, dotted] (x1)
				(b1) edge[->, dotted] (x1)
				(a2) edge[->, dotted] (x2)
				(b2) edge[->, dotted] (x2)
				(an) edge[->, dotted] (xn)
				(bn) edge[->, dotted] (xn)
				;
				\draw
				(m1) ++(-1.5, 0) edge[->] node[midway, above] {$R_0$} (m1)
				(m1) edge[->] node[midway, above] {$R_1$} (m2)
				(m2) edge[->] node[midway, above] {$R_2$} (dotdotdot)
				(dotdotdot) edge[->] node[midway, above] {$R_{n-1}$} (mn)
				(-0.5,0)  -- (-1,0.5) -- (-0.5,1) node[above,shift={(0.5,0)}]  {$E$} (-0.5,1) edge[->] (8.9,1);
				\begin{scope}[shift={(12,0)}]
					\draw[dashed,very thick] (-2.2,2.4) -- (-2.2,-3);
					\draw
					(0.5, 0) node[porte] (mi) {$\cM_i$}
					(mi) ++(-1, -1.2) node (ai) {$S_i$}
					(mi) ++(1, -1.2) node (bi) {$P_i $}
					(mi) ++(0, -2.3) node (xi) {$X_i$} 
					(xi) ++(0,-1) node {(b) Single round.}
					;
					\draw 
					(mi) ++(-1.5, 0) edge[->] node[midway, above] {$R_{i-1}$} (mi)
					(mi) edge[->] (ai)
					(mi) edge[->] (bi)
					(ai) edge[->, dotted] (xi)
					(bi) edge[->, dotted] (xi)
					(mi) edge[->] node[midway,above] {$R_{i}$} (2,0)
					;
					\draw
					(-1,0)  -- (-1.5,0.5) -- (-1,1) node[above,shift={(0.5,0)}]  {$R$} (-1,1) edge[->] (2,1)
					;
				\end{scope}
			\end{tikzpicture}
		\end{center}
		\caption{Diagrammatic depiction of EAT process and theorem. (a) Captures the overall process and (b) is the process captured by the min-tradeoff function (See \Cref{defn:mintradeofffunction}). They are related by the Entropy Accumulation Theorem (\Cref{thm:EATv2}). Note that $\cM_{n}$ may be viewed as outputting a trivial register, $R_{n} \cong \mathbb{C}$, which we have suppressed. The output of the EAT process is $\rho_{S^{n}_{1}P^{n}_{1}X^{n}_{1}E}$, the output registers along with the purification.}
		\label{fig:EATProcess}
	\end{figure}

	In this section we review the Entropy Accumulation Theorem (EAT) \cite{Dupuis2016,Dupuis2019} which is the unifying technical tool of this work. The EAT is motivated as follows. Fundamentally, the formal property of secrecy using a process with a finite number of rounds is characterized by the smooth min-entropy \cite{Renner2005}. Unfortunately, the smooth min-entropy is a functional that is difficult to calculate directly for large systems. It therefore follows that one wishes to determine tight bounds on the smooth min-entropy of a given process via a reduction that is computationally feasible. The EAT accomplishes this task for well-behaved sequential processes. The idea of the EAT is that, under some reasonable behavior, one can bound the smooth min-entropy of the finite-length process by the worst-case accepted asymptotic behavior along with some correction terms. 
	
	However, in this work we find empirically that a recent work that characterizes secrecy using the sandwiched R\'{e}nyi entropies $H^{\uparrow}_{\alpha}$ \cite{Dupuis2021} can lead to improved key rates. As such, rather than presenting results in terms of smoothed entropies in the main text, we present them in terms of $H^{\uparrow}_{\alpha}$ entropies. For completeness and perhaps intuition, in Appendix \ref{app:EAT-Sec-with-Smoothing}, we present all results in terms of smooth entropies. We now present the formal statement of the EAT along with relevant definitions after which we elaborate on the intuition and relation to the rest of this work.
	
	\subsection{Formal statement}
	
	To formally state the EAT, we review some definitions first \cite{Dupuis2016,Dupuis2019} (See \Cref{fig:EATProcess} for a visualization).
	
	\begin{definition}[EAT Channels]\label{defn:EATchannels}
		EAT channels are Completely Positive Trace-Preserving (CPTP) maps 
		\begin{IEEEeqnarray}{rL}
			\cM_i: R_{i-1} \to S_i P_i X_i R_{i}
		\end{IEEEeqnarray}
		where for all $i \in \{1,\cdots n\}$, $R_i$ are quantum systems (with $R_{n} \cong \mathbb{C}$ the trivial register) and where $S_i$, $P_i$, and $ X_i$ for $i\in \{1,\cdots n\}$ are classical systems taking values in $\cS, \cP$, and $\cX$ respectively. 
		
		Furthermore, we assume that $X_i$ is a deterministic function of $S_i$ and $P_i$. In other words, the EAT channels can be decomposed as 
		\begin{IEEEeqnarray}{rL}
			\cM_{i} = \cT_{i} \circ \cM_{i} '
		\end{IEEEeqnarray} 
		where $\cM_{i}': R_{i-1} \to S_{i} P_{i} R_{i}$ is some CPTP map and where $\cT_{i}:S_{i} P_{i} \to S_{i} P_{i}X_{i}$ is a classical operation assigning a value for $X_i$ as a function of $S_i$ and $P_i$ of the form  
		\begin{aeq}\label{eq:eat_testingMap}
	\cT_{i}(W_{S_i P_i}) = \sum_{s \in \cS , \; p \in \cP} (\Pi_{S_i,s} \ox     \Pi_{P_i,p}) W_{S_i P_i} (\Pi_{S_i,s} \ox \Pi_{P_i,p}) \ox   \ket{t(s,p)}\bra{t(s,p)}_{X_i},
\end{aeq}where $\{\Pi_{S_i, s}: s \in \cS\}$ and $\{\Pi_{P_i, p}: p \in \cP\}$ are families of mutually orthogonal projectors on $S_i$ and $P_i$, and the function $t : \cS \times \cP  \to \cX$ is a deterministic function.

	\end{definition}
	
	In the above definition, the registers $S_{i}$ and $P_{i}$ stand for `secret' key registers and `public' announcement registers respectively as this will be the natural interpretations of these registers in the context of cryptography. The $X_{i}$ register is then a `testing' register which is a function of the secret and public registers. Note that we consider the particular case of classical systems $S_{i}$ and $P_{i}$, compared to \cite{Dupuis2016,Dupuis2019,Dupuis2021} where they can be quantum, as this is sufficient for this work. The main idea of the EAT channels is that they can be composed such as in \Cref{fig:EATProcess}, starting from an initial state $\rho_{R_{0}E}^0 \in \mathrm{D}(R_0\otimes E)$ and applying the EAT maps sequentially to produce the output state 
	\begin{align}\label{eq:outputStateForm}
		\rho_{S^{n}_{1}P^{n}_{1}X^{n}_{1}E} = (\mathcal{M}_{n} \circ \cdots \circ \mathcal{M}_{1} \ox \id_{E})(\rho_{R_0E}^0)\,.
	\end{align}
	
	The most important restriction on this state is that we require it to satisfy the Markov chain conditions: 
	\begin{align} \label{eq:Markov_cond}
		S_1^{i-1} \leftrightarrow P_1^{i-1}  E \leftrightarrow P_{i} \qquad \forall i \in \{1, \ldots, n\} \ ,
	\end{align}where one may recall a quantum Markov chain, denoted $A \leftrightarrow B \leftrightarrow C$, is a quantum state $\rho_{ABC}$ such that the conditional mutual information is zero, i.e.\ $I(A:C|B) = 0$.\footnote{There exist other characterizations of quantum Markov chains, though the characterization presented here is sufficient for our exposition. See \cite{Sutter2018} for further information on quantum Markov chains.} It follows the Markov chain conditions are claims about the mutual information between the previous secret registers and the previous public registers (along with Eve's information) when conditioned on the current public announcement. The reason this is important to the proof of the theorem is that the Markovian behavior restriction guarantees that the process (defined by the sequential application of $\cM_{i}$) does not \textit{a priori} destroy entropy being accumulated in the $S^{n}_{1}$ registers. It does this by guaranteeing that, for each $i \in [n]$, the public announcement in round $i$, $P_{i}$, is not correlated to the generated secret information of previous rounds, $S^{i-1}_{1}$, when we condition on the side information $E$ and all public announcements up to then $P_1^{i-1}$.
 
 It is also perhaps useful to preemptively stress that the EAT channels do not have to be a model of the actual implementation, rather they only need to capture the same relationships between the output random variables $S^{n}_{1},P^{n}_{1}$ and $X^{n}_{1}$ as the actual implementation because the entropy only depends on the random variables. This insight is what allowed the EAT to be used for device-independent information processing where dimensions cannot be bounded. 
	
\begin{definition}[Min-tradeoff functions]\label{defn:mintradeofffunction}
	An affine function $f$ on the space of probability distributions $\PP(\cX)=\{\vect{p} \in \mathbb{R}^{\abs{\cX}}:\vect{p}(x)\geq 0, \sum_x \vect{p}(x) = 1\}$ is called a min-tradeoff function for the EAT channel $\cM_i$ if it satisfies 
	\begin{align}\label{eqn:mintradeofffunction}
		f(\vect{q}) \leq \min_{\nu \in \Sigma_i(\vect{q})} H(S_i|P_i R)_{\nu} \quad \forall \vect{q} \in \PP(\cX),
	\end{align}
	where the minimization is over the set of quantum states compatible with the statistics $\vect{q}$ that is defined as below:
	\begin{align}\label{eq:minTF_seet}
		\Sigma_i(\vect{q}):= \left\{\nu_{S_iP_iX_iR_iR} = (\cM_i \ox \id_{R})(\omega_{R_{i-1}R}): \omega \in \DM(R_{i-1}\ox R) \text{ and } \nu_{X_i} =\vect{q}\right\} \, ,
	\end{align}
	where $R$ is isomorphic to $R_{i-1}$.
\end{definition}
	\begin{remark}
	    Note that Ref. \cite{Dupuis2016} considers affine functions. However, because of $\sum_x \vect{q}(x) = 1$, any affine function on $\PP(\cX)$ can be represented as a linear function. In particular, $f$ can be specified by a vector $\vect{f} \in \mbR^{\abs{\cX}}$, that is, $f(\vect{q})= \vect{f}\cdot \vect{q} := \sum_x \vect{f}(x) \vect{q}(x)$. We note that this is equivalent to an affine function since for an affine function $f'_0 + \vect{f'}\cdot \vect{q}$, one can define $\vect{f}(x) = \vect{f'}(x) + f'_0 \ \forall x \in \cX$ such that $f'_0 + \vect{f'}\cdot \vect{q} =\vect{f}\cdot \vect{q}$.
	\end{remark}

	By \Cref{defn:mintradeofffunction}, the min-tradeoff function characterizes the minimum amount of (von Neumann) entropy of the secret information conditioned on the public information and quantum side-information (encoded as the register $R$) for all probability distributions on the testing register. This encapsulates the notion of the minimum von Neumann entropy being accumulated in a given round. For QKD security proofs, the min-tradeoff function can be seen as the term that bounds Eve's ignorance about the key in the asymptotic regime given by the Devetak-Winter formula \cite{Devetak2005}, and so the min-tradeoff function may be seen as determining the first-order term in a finite-size key distillation protocol. 
	
	With these considerations, we can state the EAT with its improved second order term \cite{Dupuis2019}.
	\begin{theorem}[Special Case of Proposition V.3 of \cite{Dupuis2019}]\label{thm:EATv2}
		Consider EAT channels $\cM_{1},...,\cM_{n}$ and their output $\rho_{S^{n}_{1}P^{n}_{1}X^{n}_{1}E}$ such that it satisfies the Markov conditions and $S_{i}$ is a classical register for all $i \in [n]$. Let $h \in \mathbb{R}$, $\alpha \in (1,2)$, and $f$ be a min-tradeoff function for $\cM_1,\dots,\cM_n$. Then, for any event $\Omega \subseteq \cX^n$ that implies $f(\freq{X_1^n}) \geqslant h$,\footnote{Following \cite{Dupuis2016}, we say an event $\Omega \subseteq \cX^{n}$ implies $f(\freq{X_{1}^n}) \geq h$ if for every $x^{n}_{1} \in \Omega$, $f(\freq{X_{1}^n}) \geq h$.}
		\begin{align}
			\label{eqn:eat-min} H_{\alpha}^{\uparrow}(S_1^n | P_1^n E)_{\rho_{|\Omega}} & > nh - n \frac{(\alpha-1)\ln 2}{2}V^{2} - \frac{\alpha}{\alpha-1} \log \frac{1}{\rho[\Omega]} - n(\alpha-1)^{2} K_{\alpha}
		\end{align}
		holds for 
		\begin{align}
			V &= \sqrt{\Var{f}+2} + \log(2 d_{S}^{2}+1) \; \label{eq: EAT constant V} \\
			K_\alpha &= \frac{1}{6(2-\alpha)^{3} \ln 2} 2^{(\alpha-1)(\log d_{S} + (\Max{f} - \MinSigma{{f}))}} \ln^{3} \left( 2^{\log d_{S} + (\Max{f} - \MinSigma{f})} + e^{2} \right) \label{eq: EAT constant Kalpha}
		\end{align}
		where 
		\begin{aeq}\label{eq:definition_minmaxvar}
		\Max{f} &:= \max_{\vect{q} \in \PP(\cX)} f(\vect{q}),\\
		\Min{f} &:= \min_{\vect{q} \in \PP(\cX)} f(\vect{q}),\\
		\mathrm{Min}_{\Sigma}(f) &:= \min_{\vect{q}:\Sigma_{i}(\vect{q}) \neq \emptyset} f(\vect{q}),\\
		\mathrm{Var}(f) &:= \max_{\vect{q}:\Sigma_{i}(\vect{q}) \neq \emptyset} \sum_{x} \vect{q}(x)\vect{f}(x)^{2} - \left[ \sum_{x} \vect{q}(x)\vect{f}(x) \right]^{2}
		\end{aeq}and $d_S = \max_{i\in [n]} |S_i|$ is the maximum dimension of the systems~$S_i$.
	\end{theorem}
	\begin{remark}
	    Note that, beyond the Markov chain conditions, the EAT applying to a state depends on the event $\Omega \subseteq \cX^{n}$ implying $f(\freq{X^{n}_{1}})$. Letting $\cZ$ be a finite alphabet, if one is interested in an event $\Omega' \subseteq \cZ^{n} \times \cX^{n}$ that guarantees such an $\Omega$ (i.e.\ $\Pr[\Omega|\Omega'] = 1$ which denotes the probability of event $\Omega$ conditioned on event $\Omega'$) then, assuming the conditions for EAT held on generating $X^{n}_{1},S^{n}_{1},P^{n}_{1}$, the EAT will hold for $\rho_{X^{n}_{1}S^{n}_{1}P^{n}_{1}E|_{|\Omega'}}$. This has been used implicitly in previous works such as \cite{Arnon2018}, but we stress it here for completeness as we also use it.
	\end{remark}
	
	It is worth noting the above theorem from \cite{Dupuis2019} is actually an improved version of the original EAT \cite{Dupuis2016}. The improvement is in the the second-order term where the dependency on the gradient of the min-tradeoff function $f$ is eliminated, which for certain applications caused the second-order term to dominate. We note that for a specific choice of $\alpha$, one can write \Cref{eqn:eat-min} in the form $H_{\alpha}^{\uparrow} >nh - O(\sqrt{n})$ \cite{Dupuis2019} to clearly separate the first-order and second-order terms. However, we do not state it in this form because, while asymptotically optimal, to get the best finite size bounds we will optimize over the parameter $\alpha \in (1,2)$ as suggested in \cite{Dupuis2019}. It is also noted that the exact from of the expression of $K_{\alpha}$ here uses the fact that $S_i$ is classical \cite[Dicussion after Eq. (22)]{Dupuis2019}.

	\subsection{Applying EAT to device-dependent QKD}\label{subsec:ApplyingEATtoDDQKD}
	
	\subsubsection{Tensor product structure}
	
	As noted earlier, the EAT maps in general do not have to be the same maps as the actual process as long as they capture the same relationship between output random variables. In the case of device-independent information processing, this is necessary since one cannot describe the device itself. In contrast, for device-dependent QKD as we consider here, without loss of generality we can let the EAT maps model the guaranteed behavior of the device in each round. It then follows that the EAT maps act on separate quantum systems. Formally, if we let $(Q_i)_{i\in [n]}$ be $n$ quantum systems then we can consider the $n$ rounds of the QKD protocol as $n$ CPTP maps of the form 
	\begin{aeq}\label{eq:tildeMaps}
		\widetilde{\cM}_i : Q_i \rightarrow S_i P_i X_i \, ,
	\end{aeq}each acting independently on its own $Q_i$ space. It follows that these maps can be expressed in the notation of the EAT theorem by defining the EAT channels $\cM_i: R_{i-1} \rightarrow S_i P_i X_i R_i$ as follows:
	\begin{aeq}
		\cM_i &= \widetilde {\cM}_i \otimes \id_{Q_{i+1}^n}
	\end{aeq}where $R_i = \bigotimes_{j=i+1}^n Q_j$ and $\id_{Q_{i+1}^n}$ is the identity map on the registers $Q_{i+1}^n$. In other words, at round $i$ the EAT channel $\cM_i$ effectively only acts on the system $Q_i$ to produce the outputs $S_i,P_i,$ and $X_i$, but not on the next systems $Q_{i+1}^n$. 
	
By the fact that the outputs $S_i$ and $P_i$ are classical, we can make another simplification. In this case, the first requirement in \Cref{defn:EATchannels} of the EAT channels boils down to the requirement that $X_i$ is obtained by applying a deterministic function on $S_i$ and $P_i$: $x_i = t(s_i,p_i)$. In typical QKD protocols, we can further assume this function is identical for all rounds. Under this assumption, the EAT channels are thus entirely defined by specifying the function $t$ and the POVM elements $\{M_{sp}\}_{s,p}$ such that 
	\begin{aeq}\label{eq:param_tensor_map}
		\widetilde\cM_i (\rho) = \sum_{s,p} \Tr(\rho M_{sp}) \ketbra{s,p}_{S_iP_i} \otimes \ketbra{t(s,p)}_{X_i}
	\end{aeq}for all $\rho\in \mathrm D(Q_i)$. These POVM elements $\{M_{sp}\}_{s,p}$ are uniquely associated to $\widetilde\cM_i$ and satisfy $M_{sp}\geq 0$ and $\sum_{s,p} M_{sp}= \idm{}$.
	
	Now that we have reduced the scope of the EAT theorem, we still have two challenges to solve. The first challenge is how  we can generate the best possible min-tradeoff function. The second is how we guarantee the Markov chain conditions in \Cref{eq:Markov_cond}.
	
	\subsubsection{Challenge 1: constructing optimal min-tradeoff functions}
	
	It is desirable to have a general procedure (applicable to many protocols) to construct min-tradeoff functions according to \Cref{defn:mintradeofffunction}. In addition to having valid min-tradeoff functions, we would like to find the best possible min-tradeoff function that can produce as tight key rates as possible for each signal block size when it is used in the EAT. In general, the construction of tight min-tradeoff functions is difficult. The difficulty arises from the non-trivial behavior of the conditional von Neumann entropy of the output state of a map (in this case the EAT channel) as a function of the resulting observations. We note that while for certain small-dimensional and theoretically simple QKD protocols, it may be possible to determine the optimal min-tradeoff function from uncertainty relations \cite[Section 5.1]{Dupuis2016}, the analytic construction of min-tradeoff functions for generic device-dependent protocols is less straightforward, and thus it warrants a numerical method. This issue has also been recognized in the device-independent scenario and been addressed with its own numerical method \cite{Brown2020}. In this work, we address this issue in the device-dependent setting and utilize additional structures of device-dependent QKD protocols. In \Cref{sec:construction_min_tradeoff}, we present two algorithms for numerically constructing (almost) tight min-tradeoff functions in the case where one knows the structure of the EAT channels.
	
	\subsubsection{Challenge 2: guaranteeing Markov chain conditions}\label{sec: MC-condition}
	
	The Markov chain conditions [\Cref{eq:Markov_cond}] put strong restrictions on the maps $\cM_i$'s that can be used with the EAT theorem. Roughly speaking, it states that, from the point of view of the adversary $E$, the process at round $j$ does not leak information about the secret register(s) $S_{i}$ of previous rounds $i < j$. In typical device-independent protocols, this restriction on the output state is in a sense trivially satisfied as all public announcements $P_{i}$, such as measurement settings, are independently seeded with random numbers. In other words, the probability distribution of the announcement $P_i$ does not depend on the state sent by Eve. Formally, if $\{p\}$ is the set of public announcements in round $i$ and there are $n$ rounds, then the Markov chain conditions trivially hold if $\Pr[p|\rho_i] = \Pr[p]$ $\forall \rho_i \in \mathrm{D}(Q_{i})$, $\forall i \in [n]$, i.e., the distribution over announcements is independent of the state being measured in each round.
	
	However, one advantage of device-dependent QKD over device-independent QKD is its ability to have more complicated public announcement structures. One example is postselection on detection events in device-dependent QKD. Postselection implicitly requires an announcement. However, since the probability of a detection event depends on the state sent by Eve, the public announcement is not based on independently seeded randomness. A simple argument shows that this is potentially problematic, as it can lead to a violation of the Markov chain conditions, even when the EAT channels act on independent systems. This happens when Eve prepares a pure state $\ket{\phi}_{Q_1^nE} = \ket{\psi}_{Q_1^n}\otimes \ket{\psi'}_E$, entangled between different rounds but not with her quantum memory. In that case, there can be correlations between announcements in one round and the private key register in another round. This could potentially prevent us from applying the EAT even if Eve could only learn from the public announcement. We therefore prove in \Cref{app:AppendixMarkov} that the following condition (\Cref{defn:weakly-dependent}) still guarantees the Markov conditions hold for Eve's optimal attack, which is sufficient for applying the EAT (\Cref{thm:EATv2}). We state this result below as \Cref{thm:Markov_block_diag}. 
	
	\begin{definition}\label{defn:weakly-dependent}
		Given some quantum-to-classical CPTP map $\mathcal W: Q \rightarrow SP$ which can be fully specified by a POVM $\{M_{sp}: s \in \cS, \ p \in \cP\}$, we say that the variable $P$ is \emph{weakly dependent} (on the input state) when there exists a decomposition of the input space $Q$ into a direct sum of orthogonal subspaces $V^\lambda$, i.e.,\ $Q = \bigoplus_\lambda V^\lambda$, such that 
		\begin{enumerate}[label=(\alph*)]
			\item the channel $\mathcal W$ is block diagonal along $V^\lambda$: i.e.,\ its POVM elements are of the form $M_{sp} = {\bigoplus}_\lambda M_{sp}^{(\lambda)}$ with $M_{sp}^{(\lambda)}\in \mathcal L(V^{\lambda})$ acting on $V^\lambda$;
			
			\item  \label{item:weakly_dependent_b} the probability of an announcement $P$ is the same for all states in a given subspace $V^\lambda$: i.e., there exist constants $c_{p,\lambda}$'s such that
			\begin{IEEEeqnarray}{rL}
				\mathrm{Pr}_P(p|\rho_\lambda) = c_{p,\lambda} \, , \quad \text{ for all }\rho_{\lambda} \in D(V^\lambda)
			\end{IEEEeqnarray}
			where $\mathrm{Pr}_P(p|\rho_\lambda) := \Tr[M_p \rho_\lambda]$ with $M_p = \sum_{s} M_{sp}$.
		\end{enumerate}
	\end{definition}
	
	Note that \labelcref{item:weakly_dependent_b} in the definition is equivalent to saying that for each $\lambda$ and $p$, $M_p^{(\lambda)} := \sum_s M_{sp}^{(\lambda)} = c_{p,\lambda}\idm_{V^\lambda}$ is proportional to the identity. This means that, equivalently, for any $\sigma \in Q$, $\Pr_{P}(p|\sigma) = c_{p,\lambda} \Tr[\Pi_{\lambda}\sigma\Pi_{\lambda}]$, where $\Pi_{\lambda}$ is the projector onto the subspace, i.e. the announcement only depends on a constant and the weight of the state on the subspace, which makes it `weakly' dependent on the state. For intuition, this is distinct from independently randomly seeded announcements where each announcement has POVM elements of the form $M_{p} = c_{p} \idm_{Q}$, which results in the announcement being independent on the state all together. Note the independently randomly seeded case trivially satisfies \Cref{defn:weakly-dependent}, and so it is a (strictly) special case of weakly dependent announcements. 
	
	Here we state that the advantage of weakly dependent announcements is that they guarantee the EAT may be applied.
	
	\begin{shaded}
	\begin{theorem}\label{thm:Markov_block_diag}
		If the announcements $P_i$ of the CPTP maps $\widetilde \cM_i: Q_i \rightarrow S_iP_iX_i$ are \emph{weakly dependent}, then the result of \Cref{thm:EATv2} may be applied to prove security.
	\end{theorem}
	\end{shaded}
	
	We give a physical intuition why this theorem would hold and defer the proof to \Cref{app:AppendixMarkov}. The POVM being of the block-diagonal form in \Cref{defn:weakly-dependent} means the subspace information is knowable to Eve. This is because without loss of generality, Eve should send such block-diagonal states and so she knows the information by implementing a quantum nondemolition (QND) measurement that determines the subspace which she then stores in a secondary register. Indeed, this idea of Eve having an extra register with the subspace information is how this is proven in \Cref{app:AppendixMarkov}.  We note equivalently that Eve's purification of such a block-diagonal state includes the subspace information and so it is knowable to Eve from the perspective of purification as well.

	\section{The Quantum Key Distribution Protocol}\label{sec:QKD_protocol}
	In this section we present the class of entanglement-based (EB) QKD protocols for which we prove security. At a high level, there are three related protocols: the  QKD protocol for physical implementations (\Cref{prot:physicalDDQKD}), its equivalent virtual QKD protocol (\Cref{prot:DDQKD}) for security proof purposes, and the (also virtual) Entropy Accumulation subprotocol (\Cref{prot:DDEA}), referred to as the EAT subprotocol, to which we apply the EAT to derive a desired entropic bound. The virtual QKD protocol requires certain properties so that the EAT subprotocol can be applied to its analysis. Of particular importance is that the virtual QKD protocol can be thought of as acting sequentially, i.e.\ the signals are processed round by round, and that the announcements satisfy the Markov chain conditions in \Cref{eq:Markov_cond}. The necessary properties of the virtual QKD protocol impose necessary structures on the physical implementation of the protocol so that it is equivalent to the virtual QKD protocol, where equivalence is in terms of the (quantum) random variables generated as this is what both the EAT and the security proof more generally consider. 
	
	\subsection{Physical and virtual protocol description}\label{subsec:QKD_protocol_description}
	In this section, we present the physical QKD protocol followed by the virtual QKD protocol to which it is equivalent.
	\begin{protocol}[H]
		\caption{Physical Device-Dependent Quantum Key Distribution Protocol}\label[protocol]{prot:physicalDDQKD}
		\textbf{Inputs:} \\
		\hspace{0.5cm}
		\begin{tabular}{ l l}
			$\{M^{A}_{a}\}_a, \{M^{B}_{b}\}_b$ & Alice and Bob's measurement devices (POVMs) \\
			$ \mathbf{K} $  & Subset of Alice and Bob's public announcements to be kept during sifting \\
			$n \in \mathbb{N}_{+}$ & Number of rounds  \\
			$\gamma \in (0,1]$ & Probability of testing \\
			$\mathcal{Q}$ & Set of acceptable frequency distributions over $\cX$
		\end{tabular}
		
		\vspace{0.5cm}
		
		\textbf{Protocol:}
		\begin{enumerate}
			{\setlength\itemindent{-1em}\item[] For $i \in [n]$, do \Cref{step:phys_transmission,step:phys_measurements,step:data_partitioning,step:phys_labeling}}:
			    \item\label[step]{step:phys_transmission} \begin{addmargin}[0.5cm]{0cm}
			        \textbf{State Transmission}: A source (that may be under Eve's control) distributes a state $\rho_{Q_{i}}$ between Alice and Bob.
			    \end{addmargin}
			
			    \item \label[step]{step:phys_measurements}
			    \begin{addmargin}[0.5cm]{0cm} 
				    \textbf{Measurements}: 
				    Alice and Bob implement their local POVMs $\{M_{a}^{A}\}_{a}$, $\{M_{b}^{B}\}_{b}$ to measure their respective halves of the state and record their outcomes. 
				\end{addmargin}
			
			    \item \label[step]{step:data_partitioning}
			    \begin{addmargin}[0.5cm]{0cm} 
				\textbf{Data Partition}: 
				    Alice partitions her data into (data that will eventually be) public $\widetilde{A}_{i}$ and (data that will stay) private $\overline{A}_{i}$. Likewise, Bob partitions his data into public $\widetilde{B}_{i}$ and private $\overline{B}_{i}$.
			    \end{addmargin}
			
			    \item \label[step]{step:phys_labeling}\begin{addmargin}[0.5cm]{0cm}
		        \textbf{Testing Designation}:
			    Alice randomly chooses $T_{i} \in \{0,1\}$ according to $\Pr(T_{i}=1) = \gamma$.
			    \end{addmargin}

			\item \label[step]{step:phys_announcements} \textbf{Announcements}: Alice and Bob announce their public data $\wt{A}^{n}_{1}$ and $\wt{B}^{n}_{1}$, respectively. Alice then announces $T^{n}_{1}$. For all $i \in [n]$ such that $T_{i} = 1$, Alice announces $\ol{A}_{i}$.
			
			\item \label[step]{step:phys_PE}
			\textbf{Parameter Estimation}:
			For all $i \in [n]$, Bob computes $X^{n}_{1}$ according to $X_{i} = t(\overline{A}_{i}, \overline{B}_{i}, \widetilde{A}_{i}, \widetilde{B}_{i})$ if $T_{i}=1$ and $X_{i} = \perp$ otherwise, where $t$ is a deterministic function. Alice and Bob abort the protocol if $\freq{X^{n}_{1}} \not \in \mathcal{Q}$.

			\item \label[step]{step:phys_sifting} 
			\textbf{General Sifting}:
			If $(\widetilde{A}_{i},\widetilde{B}_{i},T_{i}) \in \mathbf{K}^{C} \times \{0\}$, Alice sets the $\overline{A}_{i} = \perp$, where $\mathbf{K}^{C}$ is the complement of $\mathbf{K}$.
			
			\item \label[step]{step:phys_keymap}
			\textbf{Key Map}: If $(\widetilde{A}_{i},\widetilde{B}_{i},T_{i}) \in \mathbf{K} \times \{0\}$, Alice updates $\overline{A}_{i} := f_{KM}(\overline{A}_{i},\widetilde{A}_i,\widetilde{B}_i)$ where $f_{KM}$ is the key map. This subset of rounds may be denoted as the register $\mathbf{Z}$, Alice's raw key.
			
			\item \label[step]{step:phys_EC}\textbf{Error Correction \& Detection:} Using an error correction scheme, Alice and Bob communicate for Bob to construct his guess of Alice's raw key, $\mathbf{\widehat{{Z}}}$. He then uses a 2-universal hash function to send a hash of his guess to Alice. This detects if the correction worked. If it did not, they abort. Otherwise, they continue.
			
			\item \label[step]{step:phys_PA}\textbf{Privacy Amplification:} Using a family of 2-universal hash functions, Alice and Bob perform privacy amplification to achieve the desired secrecy.
		\end{enumerate}
	\end{protocol}
	A few remarks are necessary. First, we have implicitly required that the announcement structure of the protocol be round by round. This announcement structure is necessary to move to our virtual protocol as we need a protocol that is sequential for the majority of the steps. Second, we have required no announcements be made until all of the measurement data have been obtained. This is important as it avoids Eve altering her actions depending on announcements. This requirement makes the protocol equivalent to the one where Eve distributes the whole $n$-round state (for which she may hold a purification) at the beginning and then only gains additional information via Alice's and Bob's announcements. The latter will be the necessary structure for the virtual protocol. Third, we have required the function $t$ in \Cref{step:phys_PE} to be deterministic. This is a condition needed to apply the Entropy Accumulation Theorem [see \Cref{eq:eat_testingMap}], but it does not seem limiting for standard protocols.
	
	We note that requiring the testing be done round by round specifically differs from standard practice in device-dependent QKD security proofs \cite{Renner2005} which perform fixed-length parameter estimation. Fixed-length parameter estimation is when, before the protocol is executed, it is decided that $m$ of the $n$ signals will be used for parameter estimation. Then, rather than having \Cref{step:phys_labeling,step:phys_PE} of \Cref{prot:physicalDDQKD}, the protocol would include Alice uniformly choosing a bit string from the set of bit strings of length $n$ and Hamming weight $m$ which determines which rounds to test. This is not necessarily a large gap if one considers testing probability $\gamma$ such that $\gamma n = m$, as a Bernoulli variable converges to its mean quickly. However, for rigor, after proving the security of \Cref{prot:DDQKD}, which is equivalent to the security of \Cref{prot:physicalDDQKD}, in \Cref{sec:FixedLengthPE}, we present how to convert statements of security on this probabilistic round-by-round testing protocol to statements of security on the fixed length testing protocol without introducing any looseness.
	
	Lastly, we note that we need the extra assumption beyond being round by round that the announcements satisfy the Markov conditions. More technically, the announcement structure will have to be such that the virtual QKD protocol would satisfy the Markov chain conditions in \Cref{eq:Markov_cond} in the case that Alice did not announce her fine-grained data $\ol{A}_{i}$ when the round $i$ is a testing round, i.e., $T_{i} = 1$. Since Alice does announce her fine-grained data during testing rounds, we stress why this works preemptively. If Alice's fine-grained data $\ol{A}_{i}$ is announced publicly, it could threaten the Markov chain conditions by leaking too much information if Eve sends states that are correlated across rounds. However, this fine-grained data is needed by Bob to compute $X_{i}$ on testing rounds, so we require Alice to announce this data. We therefore need a way to address this. In the security proof, we start with the conditional entropy of Alice's raw key which, among other registers, is conditioned on Eve knowing the fine-grained data Alice announces during testing rounds. By applying entropic chain rules, we are able to convert to a conditional entropy term corresponding to the case in which Alice kept all fine-grained data private. In this case, the Markov conditions hold by restrictions on the announcement structure we require. The final technical issue is that physically $X_{i}$ needs to be computed using $\ol{A}_{i}$ and $\ol{B}_{i}$. It follows that if neither party has both registers, there is no physical process to compute $X_{i}$. However, this is not an issue as the EAT subprotocol is virtual (that is, there is no need to implement this protocol in practice) and only needs to generate the same (quantum) random variables as the real process. Therefore we construct a sequence of protocols where the security claim on the physical protocol holds by the equivalence to the security of the virtual protocol whose security relies on a virtual EAT subprotocol. We also note for intuition that there is a penalty to the key rate by announcing the $\ol{A}_{i}$ in the aforementioned chain rules, which is not a part of the virtual EAT subprotocol. We discuss the specific assumptions on the announcement structure to guarantee the Markov chain conditions hold in \Cref{subsec:restrictionsOnAnnouncements} after presenting the virtual QKD protocol.
	
	We now present the virtual QKD protocol. Given the discussion above, we note that the difference from the physical QKD protocol is that announcements, sifting, generation of the test round data (but not aborting based on the test data), and the key map are all implemented round by round. Beyond this conversion of many steps to being performed sequentially, the virtual QKD protocol is the same as the physical one. This is what will allow the protocols to be equivalent.
	
	\begin{protocol}[H]
		\caption{Virtual Device-Dependent Quantum Key Distribution Protocol}\label[protocol]{prot:DDQKD}
		\textbf{Inputs:} \\
		\hspace{0.5cm}
		\begin{tabular}{ l l}
			$\{M^{A}_{a}\}_a, \{M^{B}_{b}\}_b$ & Alice and Bob's measurement devices (POVMs) \\
			$ \mathbf{K} $  & Subset of Alice and Bob's public announcements to be kept during sifting \\
			$n \in \mathbb{N}_{+}$ & Number of rounds  \\
			$\gamma \in (0,1]$ & Probability of testing \\
			$\mathcal{Q}$ & Set of acceptable frequency distributions over $\cX$
		\end{tabular}
		
		\vspace{0.5cm}
		
		\textbf{Protocol:}
		\begin{enumerate}
			\addtocounter{enumi}{-1}
			\itemsep0em
			\item \label[step]{step:transmission} \textbf{State Transmission:} Eve distributes the $n$ states, which may be entangled in an arbitrary manner, such that the total state is of the form $\rho_{Q_{1}^{n}E}$.
			{\setlength\itemindent{-1em}\item[] For $i \in [n]$, do \Cref{step:measurements,step:announcement,step:labeling,step:sifting,step:keymap,step:test}:}
			
			\item \label[step]{step:measurements}\begin{addmargin}[0.5cm]{0cm} 
				\textbf{Measurements}: 
				Alice and Bob implement their local POVMs $\{M_{a}^{A}\}_{a}$, $\{M_{b}^{B}\}_{b}$ to measure their respective halves of the state and record their outcomes. \end{addmargin}
			
			\item \label[step]{step:announcement}\begin{addmargin}[0.5cm]{0cm} 
				\textbf{Data Partition and Announcement:} Alice partitions her data into public $\widetilde{A}_{i}$ and private $\overline{A}_{i}$. Likewise, Bob partitions his data into public $\widetilde{B}_{i}$ and private $\overline{B}_{i}$. Then they announce their public data.
			\end{addmargin}
			
			\item \label[step]{step:labeling}\begin{addmargin}[0.5cm]{0cm} 
				\textbf{Testing Designation}: Alice randomly chooses $T_{i} \in \{0,1\}$ according to $\Pr(T_{i}=1) = \gamma$.
			\end{addmargin}
			
			\item \label[step]{step:sifting}\begin{addmargin}[0.5cm]{0cm} \textbf{General Sifting}:
				If $(\widetilde{A}_{i},\widetilde{B}_{i},T_{i}) \in \mathbf{K}^{C} \times \{0\}$, Alice sets the $\overline{A}_{i} = \perp$. Denote $$\mathcal{S} = \{ i \in [n]: (\widetilde{A}_{i},\widetilde{B}_{i},T_{i}) \in \mathbf{K}^{C} \times \{0\} \} $$ so that $\overline{A}_{\mathcal{S}}$ denotes the registers of discarded rounds. \end{addmargin}
			
			\item \label[step]{step:keymap}\begin{addmargin}[0.5cm]{0cm} \textbf{Key Map}: If $(\widetilde{A}_{i},\widetilde{B}_{i},T_{i}) \in \mathbf{K} \times \{0\}$, Alice updates $\overline{A}_{i} := f_{KM}(\overline{A}_{i},\widetilde{A}_i,\widetilde{B}_i)$ where $f_{KM}$ is the key map. This subset of rounds may be denoted as the register $\mathbf{Z}$, Alice's raw key. \end{addmargin}
			
			\item \label[step]{step:test}\begin{addmargin}[0.5cm]{0cm} \textbf{Statistical Tests}:\\[-7mm]
				\begin{itemize}
					\item If $T_{i} = 1$, Alice announces $\overline{A}_{i}$ publicly. Using this, Bob generates $X_i$ using a deterministic function $t$ such that $X_{i} = t(\overline{A}_{i}, \overline{B}_{i}, \widetilde{A}_{i}, \widetilde{B}_{i})$.
					\item If $T_{i} = 0$, Bob sets $X_{i} = \perp$. 
				\end{itemize}
				Denote $\cT = \{i \in [n]: T_{i} = 1\}$. Then the registers Alice announces are $\overline{A}_{\mathcal{\cT}}$. \end{addmargin}

			\item  \label[step]{step:PE}\textbf{Parameter Estimation}: Alice and Bob abort the protocol if $\freq{X^{n}_{1}} \not \in \mathcal{Q}$.
			
			\item \textbf{Error Correction and Parameter Estimation}: Do \Cref{step:phys_EC,step:phys_PA} of \Cref{prot:physicalDDQKD}.

		\end{enumerate}
	\end{protocol}

	We note that both \Cref{prot:physicalDDQKD,prot:DDQKD} assume Alice establishes the key. However, if Bob were to establish the key, this merely switches the roles of Alice and Bob in \Cref{prot:physicalDDQKD} and \Cref{prot:DDQKD}. Thus, this setting is not a restriction. Furthermore, we emphasize we still require that testing be determined in a round-by-round manner as the EAT is an i.i.d. reduction for sequential processes, which was discussed earlier.
	
	The only remaining assumption to make explicit is the ability to guarantee the Markov chain conditions hold from the announcement structure.
	
	\subsection{Assumptions on public announcements}\label{subsec:restrictionsOnAnnouncements}
	
	To verify the applicability of the EAT for a given protocol, we need to verify that the conditions from \Cref{thm:Markov_block_diag} are satisfied. This puts some restrictions on the type of protocols that can be included in our security proof and more specifically on the type of announcements and postselection that we can do. 
	
	\Cref{thm:Markov_block_diag} says that the partitioned announcements can have some dependence on the state, but only in a limited way. It can only depend on the subspace $V^\lambda$ in which the state lies. The underlying reason for this is that for block diagonal measurements, we can assume without loss of generality that Eve sends a state $\rho_{Q_1^nE}$ where each state in register $Q_i$ is block diagonal and so she already holds that information in her purification. We therefore do not give her new information about the state by leaking $\wt{A}_i$ and $\wt{B}_i$. Recall that announcements being independent of the input state is a particular case of this setting as explained in \Cref{sec: MC-condition}.
	
	To summarize, we make the following assumption throughout this work, which guarantees we can apply \Cref{thm:EATv2} by \Cref{thm:Markov_block_diag} so long as we guarantee the conditions stated in Definition \ref{defn:weakly-dependent} are satisfied.
	\begin{shaded}
	    \begin{assumption}\label[assump]{assumption:block_diagonal}
	        The measurements and subsequent announcement structure of Alice and Bob guarantee that the partitioned data $\widetilde{A}_i$'s and $\wt{B}_i$'s are weakly dependent (\Cref{defn:weakly-dependent}).
	    \end{assumption}
	\end{shaded}
	
	\FloatBarrier
	
	\paragraph{Example (optical discrete-variable protocols):}
	
	The generalization from independently seeded announcements to weakly dependent announcements is crucial in the case of discrete-variable protocols. In those protocols, we typically perform postselection in the case of loss to remove the no-detection events from the raw data. However, postselection implies that each party needs to publicly announce if they have a detection or not (in addition to announcing the basis choice for protocols like BB84). This announcement is potentially problematic because the detection probability depends on the state; i.e., states with more photons have a higher probability of being detected. However, we can use the fact that the measurements by single-photon detectors commute with the total photon number operators $N$. In other words, Alice's (or Bob's) optical space can be decomposed in subspaces, each of which has a given total photon number $n_a$ (or $n_b$), as $Q_A = \oplus_{n_a=0}^\infty V_A^{n_a}$ (or $Q_B = \oplus_{n_b=0}^\infty V_B^{n_b}$), and the measurement device's POVM elements are block diagonal along the subspaces $V^{n_a}_A$ (or $V^{n_b}_B$). 
	
	Let's take BB84 as an example. Assuming that the basis choice is independently seeded, we only need to verify that the probability of a detection is the same for all states with the given basis choice $x$ and the same total photon number $N_A=n_a$. That is, for each total photon number $n_a$ and basis choice $x$, there exists some constant $c_{x,n_a}$ such that the probability of detection conditioned on sending a state from the basis $x$ in the subspace $V^{n_a}_A$ is
	\begin{aeq}
		\mathrm{Pr}(\mathrm{detection}\; | x,\rho_{n_a})= c_{x,n_a}\fa \rho_{n_a} \in \DM(V^{n_a}_A) \; .
	\end{aeq}Likewise, we have a similar requirement for Bob's detectors. This requirement needs to hold for both Alice's and Bob's detectors.
	
We remark that this property holds for the BB84 passive-detection setup using identical (imperfect) single-photon detectors (See \Cref{subsec:opticalBB84}).

	\begin{remark}
	    As mentioned in the introduction, \cite{Metgers-2022b} uses the generalised EAT \cite{Metgers-2022a} to establish security for  prepare-and-measure protocols by assuming the protocol only has one signal in the physical channel at at time, i.e. is truly a sequential protocol. The assumption in \cite{Metgers-2022b} is so that they can satisfy the `no-signalling' condition for the generalised EAT. In \cite{Metgers-2022a} it is shown how the original EAT's Markov chain conditions imply the `no-signalling' condition. Thus, the set of protocols we consider here all can also be covered by the generalised EAT too. However, as explained in \cite{Metgers-2022b}, the only advantage of the generalised EAT for device-dependent QKD is to apply to prepare-and-measure protocols, which requires the truly sequential assumption, which then does not include most of the protocols considered in \Cref{prot:physicalDDQKD}. That is to say, our class of entanglement-based protocols can be analysed by the generalised EAT, but this does not seem to provide an advantage over the original EAT, and if one considers prepare-and-measure protocols, our class of protocols is too general to be addressed by the generalised EAT.
	\end{remark}

 \subsubsection{Completeness of a QKD Protocol}\label{subsec:completeness}
 A protocol that generates a secret key but always aborts on the honest implementation is not practically useful. Completeness measures the probability the protocol will not abort on the honest implementation.\footnote{This has also been referred to as robustness \cite{Renner2005}.}
 \begin{definition}\label{def:completeness}
     A QKD protocol is called $\ve_{\text{QKD}}^{\text{C}}$-complete if the probability of aborting the entire protocol is at most $\ve_{\text{QKD}}^{\text{C}}$. Similarly, any subprotocol $P$ is $\ve^{\text{C}}_{\text{P}}$-complete if that subprotocol aborts with probability at most $\ve^{\text{C}}_{\text{P}}$.
 \end{definition}
  Clearly one may bound the QKD protocol's completeness $\ve^{\text{C}}_{\text{QKD}}$ via the union bound of the completeness of each subprotocol as has been done previously, e.g. \cite{Arnon-2020a}. The completeness of error correction is a function of the specific implementation, so we omit a discussion on this. Therefore, the major concern is the completeness of parameter estimation. It is both intuitive and well-known that as the set of accepted statistics increases, the probability of aborting the protocol can only decrease for any input, but the key length is likely to also decrease. We provide simple bounds for completeness that are relevant for any protocol that satisfies Protocol \ref{prot:DDQKD}.
  \begin{proposition}
      Following the notation of Protocol \ref{prot:DDQKD}, let $\sigma = \cE(\rho)$ be the state after the honest preparation, $\rho$, and (possibly noisy) transmission channel $\cE$, which we assume to be memoryless between transmissions. Let $p$ be the probability distribution of joint outcomes of Alice and Bob's measurements on $\sigma$ in testing rounds. Let the set of accepted statistics be $\cQ := \{f \in \PP(\cX): \|p-f\|_{1} \leq \xi_{t} \}$ where $\xi_{t} \in (0,1]$. Then the completeness of the parameter estimation subprotocol when $\ve^{\text{C}}_{\text{PE}}$ satisfies the upper bound
      \begin{align*}
            \ve^{\text{C}}_{\text{PE}} \leq 2\exp(-2m' (\gamma - \gamma')^{2}) + \exp(-m' \left(\frac{t^{2}}{2\ln2}-|\cX|\frac{\log(m'+1)}{m'}\right)) \ , 
      \end{align*}
      where $\gamma' < \gamma$ and $m' := \lfloor \gamma' n \rfloor$. Moreover, this goes to zero as $n$ goes to infinity.
  \end{proposition}
  \begin{proof}
      First, for any fixed number of tests $m$, the probability of aborting, $\ve^{\text{C}}_{\text{PE|m}}$, is bounded as
      \begin{align*}
          \ve^{\text{C}}_{\text{PE|m}} \leq \min\left\{1, \exp(-m\left(\frac{t^{2}}{2\ln2}-|\cX|\frac{\log(m+1)}{m}\right))\right\} \ ,
      \end{align*}
      which is a straightforward implication of Pinsker's inequality \cite[Lemma 11.6.1]{Cover2005} and the law of large numbers via method of types \cite[Theorem 11.2.1]{Cover2005} as noted in \cite{George-2020a}. Then we have
      \begin{align*}
          \ve^{C}_{\text{PE}} =& \sum_{k' \in [n]} \Pr[k' \text{ test rounds}]\ve^{\text{C}}_{\text{PE}|k'} \\
          \leq& \Pr[\# \text{ of test rounds} < m'] + \ve^{\text{C}}_{\text{PE}|m'} \\
          \leq & \Pr[\text{fraction of rounds that are tests}\geq \gamma - \gamma' ] + \ve^{\text{C}}_{\text{PE}|\lfloor \gamma' n \rfloor } \\
          \leq & 2\exp(-2 \lfloor \gamma' n \lfloor (\gamma - \gamma')^{2}) + \ve^{\text{C}}_{\text{PE}|\lfloor \gamma' n \rfloor }\ ,
      \end{align*}
      where the first inequality is choosing to group the probabilities by some $m'$, the second is using $m' := \gamma' n$ for some $\gamma' < \gamma$, and the final equality Hoeffding's inequality for independent Bernoulli random variables.
  \end{proof}

	\section{Security of Device-Dependent QKD from EAT}\label{sec:SecurityofDDQKDfromEAT}
		In this section we present the security of the considered QKD protocols (\Cref{prot:physicalDDQKD}). We stress that our security proof is for coherent attacks. Recall \Cref{prot:physicalDDQKD} does not assume anything about the state distribution but guarantees announcements are after all measurements. It is therefore equivalent to \Cref{prot:DDQKD} where $\rho_{Q_1^{n}E}$ is an arbitrary state and announcements are made round-by-round. Recall that for an i.i.d.\ collective attack,\footnote{Collective attacks are usually defined (e.g. \cite{Scarani2009}) by assuming that Eve interacts with each signal with a new ancillary state using the same unitary operation (which also includes mixed-unitary channels). Under this definition, collective attacks and i.i.d.\ assumption can be treated as synonymous for many protocols (as long as the protocol structure allows for the i.i.d. behavior). Some authors seem to generalize the definition to allow time-dependent unitary operations to include, for example, channels with a slowly rotating reference frame. This generalized definition would be non-i.i.d. collective attacks. For this reason, we use the terminology of i.i.d.\ collective attacks to emphasize the i.i.d. assumption.} one would assume $\rho_{Q_{1}^{n}E} = \sigma_{QR}^{\otimes n}$ so that Eve sends $n$ copies of the state $\sigma_{Q}$ for which she holds a purification $\sigma_{QR}$ to Alice and Bob. However, we do not make this assumption here in the security proof.

	\subsection{Entropy rate of EAT process}\label{sec:EATprocess}
	
As discussed previously, we aim to apply the EAT to analyze the virtual QKD protocol (\Cref{prot:DDQKD}) which is equivalent to the physical QKD protocol (\Cref{prot:physicalDDQKD}) in terms of security. However, we cannot directly use the EAT to analyze the virtual protocol. Specifically, the QKD protocol only accumulates entropy\footnote{In the main text we consider the sandwiched R\'{e}nyi entropy. In \Cref{app:EAT-Sec-with-Smoothing} we consider smooth min-entropy. In both cases the main point is that some entropic quantity accumulates, so we just use entropy without a qualifier to refer to both cases.} until parameter estimation. As such, our interest is in the entropy accumulated given some event passes, namely that parameter estimation passes. Therefore, the security proof is broken into two parts. First, one proves the entropy accumulation rate on a `subprotocol' that is nearly (for technical reasons) equivalent to that of \Cref{step:transmission,step:measurements,step:announcement,step:labeling,step:sifting,step:keymap,step:test,step:PE} of \Cref{prot:DDQKD}. After this, one proves the security of the virtual protocol by relating it back to the subprotocol. In this subsection, we present the EAT subprotocol (\Cref{prot:DDEA}) and its entropy accumulation rate. Then in the next subsection we present the secure key length of the QKD protocol (\Cref{thm:keylengthwithoutsmoothing}). A similar proof for a device-independent QKD protocol was given in \cite{Arnon2018}. However, that proof relied on a particular structure of the protocol which simplifies the analysis but is not as general as the protocol we consider here. In particular, the parameter estimation was done on the error corrected bit string instead of on the raw measurement results, like we do here. This allows us to use all the available measurement information to bound the key rate.
	\begin{protocol}[H]
		\caption{Device-Dependent Entropy Accumulation Subprotocol} \label[protocol]{prot:DDEA}
		\textbf{Inputs:} Same as \Cref{prot:DDQKD} \\
		\textbf{Protocol:} Run \Cref{step:transmission,step:measurements,step:announcement,step:labeling,step:sifting,step:keymap,step:test,step:PE} of \Cref{prot:DDQKD}, except in \Cref{step:sifting,step:test}, Bob assigns $\overline{B}_{i} = \perp$ (i.e. if $T_{i} = 0$, then $\overline{B}_{i} = \perp$), and in \Cref{step:test} Alice \textit{does not} announce $\ol{A}_{i}$ when $T_{i} = 1$. Regardless, $X_{i} = t(\ol{A}_{i},\ol{B}_{i},\wt{A}_{i},\wt{B}_{i})$ can be calculated the same as before.
	\end{protocol}
	
	It is worth noting that in principle the EAT does not rely on knowing all of the steps of the protocol explicitly. It just requires the existence of EAT channels that output (quantum) random variables that are the same as the real process. This is why we are not concerned about $X_{i}$ being computed locally. Here we described the procedure per round largely the same as in \Cref{prot:DDQKD}, because when we use our numerical algorithms to construct the min-tradeoff function (\Cref{sec:construction_min_tradeoff}), we use our knowledge of a way to implement the process to construct the EAT channels explicitly. We also note for this reason, this protocol can handle QKD protocols in which one party's public announcement is conditioned upon the other, so long as one can prove the resulting output random variables still satisfy the required Markov conditions.
	
	With the protocol defined, we may use the EAT to bound the entropy accumulated. We first state the relabeling from the registers used in the EAT statement in \Cref{sec:EATBackground} and ones used in \Cref{prot:DDEA}:
	\begin{align*}
		S_{i} &\leftrightarrow \overline{A}_{i} \overline{B}_{i} \\
		P_{i} &\leftrightarrow \widetilde{A}_{i} \widetilde{B}_{i} T_{i} \\
		X_{i} &\leftrightarrow X_{i} \\
		Q_{i} &\leftrightarrow A_{i}B_i \\
		R_{i} & \leftrightarrow R_i\\
		E &\leftrightarrow E \ .
	\end{align*}
	With these conversions, we state the sandwiched R\'{e}nyi  entropy rate of the entropy accumulation subprotocol (\Cref{prot:DDEA}).
	\begin{theorem}\label{thm:entRate} 
		Consider the entropy accumulation protocol defined in \Cref{prot:DDEA} and assume the \Cref{assumption:block_diagonal} is satisfied. Let $\Omega = \{x^{n}_{1} \in \cX^{n}: \freq{x^{n}_{1}} \in \mathcal{Q}\}$ and $\rho$ be the output of the protocol. Let $h$ such that $f(\vect{q}) \geq h$ for all $\vect{q} \in \mathcal{Q}$ where $f$ is the min-tradeoff function generated by either \Cref{alg:Algorithm1_min_tradeoff_function} or \Cref{alg:Algorithm2_min_tradeoff_function}. Then for any $\vEA \in (0,1)$, either the protocol aborts with probability greater than $1-\vEA$, or
		\begin{align}\label{eq:entropyRate}
			H^{\uparrow}_{\beta}(\ol{A}^{n}_{1} \ol{B}^{n}_{1}| \widetilde{A}^{n}_{1} \widetilde{B}^{n}_{1} T^{n}_{1} E)_{\rho_{|\Omega}} > n h - n \frac{(\beta-1)\ln 2}{2}V^{2} - \frac{\beta}{\beta-1} \log \frac{1}{\vEA} - n(\beta-1)^{2} K_{\beta}
			\end{align}
		where $\beta \in (1,2)$ by \Cref{thm:EATv2}.
	\end{theorem}
	\begin{proof} 
		We simply check that everything we have done satisfies the requirements of EAT.
		\begin{enumerate}
			\item As we assume the protocol satisfies \Cref{assumption:block_diagonal} which implies that \Cref{thm:Markov_block_diag} holds for the protocol, the output state $\rho$ satisfies the Markov conditions in \Cref{eq:Markov_cond}.
			\item By our definition of how we compute the test register $X_{i}$, the testing map $\cT_{i}$ is of the form in \Cref{eq:eat_testingMap} of \Cref{defn:EATchannels}.
			\item By our construction of the min-tradeoff function (\Cref{alg:Algorithm1_min_tradeoff_function} or \Cref{alg:Algorithm2_min_tradeoff_function}), we have a min-tradeoff function and the value $h$ in the statement of the theorem above satisfies the requirements to be $h$ in the statement of \Cref{thm:EATv2}.
		\end{enumerate}
		Thus we have satisfied the requirements of the EAT and can apply it.
	\end{proof}

	\begin{remark} 
		The statement of \Cref{thm:entRate} requires that either the EAT protocol aborts with a probability greater than $1-\vEA$ or else the entropy bound holds. In the language of Renner's PhD thesis \cite{Renner2005}, this theorem says either $\rho$ is $\vEA$-securely filtered by the EAT protocol or the entropy bound holds. This explains the replacement of the failure probability of parameter estimation, which appeared in Renner's original coherent-attack security proofs \cite{Renner2005}, with the term $\vEA$ in the statement of  $\ve$-security.
	\end{remark}
	
	\FloatBarrier

	\subsection{Security of QKD protocol}\label{sec:SecurityProof}
	We can now present the key length for a QKD protocol using the EAT without introducing any smoothing. We note this depends on the construction of a max-tradeoff entropy accumulation theorem for the sandwiched R\'{e}nyi entropy $H^{\uparrow}_{\delta}$, which we provide in \Cref{app:EAT-Sec-without-Smoothing}.
	\begin{shaded}
		\begin{theorem}
			\label{thm:keylengthwithoutsmoothing} 
			Consider any QKD protocol which follows \Cref{prot:physicalDDQKD} and satisfies \Cref{assumption:block_diagonal}. Let $\vSec,\vEC,\vAcc \in (0,1]$ such that $\vAcc \geq \vSec + \vEC$. Let $h$ such that $f(\vect{q}) \geq h$ for all $\vect{q} \in \mathcal{Q}$ where $f$ is the min-tradeoff function generated by \Cref{alg:Algorithm1_min_tradeoff_function} or \Cref{alg:Algorithm2_min_tradeoff_function}. \\
			
			\noindent Let $\beta \in (1,2)$, $\delta \in (1/2,1)$ and $\alpha = \frac{-\beta +\delta}{-1+2\delta-\beta\delta}$.
			The QKD protocol is $\vAcc$-secure for key length
			\begin{equation}\label{eq:KeyLength_without_smoothing} 
				\begin{aligned}
					\ell \leq & n h - \leak_{\vEC} - n \frac{(\beta-1)\ln 2}{2}V^{2} - n(\beta-1)^{2} K_{\beta} - n\gamma\log|\mathcal{A} \times \mathcal{B}| \\
					& \hspace{3mm} 
					+\frac{\beta-\delta}{(\beta+1)(1-\delta)}\log(\vAcc)
					+ \frac{\alpha}{\alpha - 1} \log(\vSec) + 1
				\end{aligned}
			\end{equation}
			where 
			\begin{align*}
				V &= \sqrt{\Var{f}+2} + \log(2 d_{S}^{2}+1) \\
				K_\beta &= \frac{1}{6(2-\beta)^{3} \ln 2} 2^{(\beta-1)(\log d_{S} + \Max{f} - \MinSigma{{f})}} \ln^{3} \left( 2^{\log d_{S} + \Max{f} - \MinSigma{f}} + e^{2} \right)
			\end{align*}
			where
			$\mathcal{A}$ and $\mathcal{B}$ are the alphabets of private outcomes for Alice's and Bob's announcements excluding the symbol $\perp$, respectively, $d_{S} = (|\cA|+1)(|\cB|+1)$, $\leak_{\vEC}$ is the amount of information leakage during the error correction step.
		\end{theorem}
	\end{shaded}
	\begin{proof}
	       See \Cref{app:EAT-Sec-without-Smoothing}.
    \end{proof}
    \begin{remark}
        In the proof of \Cref{thm:keylengthwithoutsmoothing} in \Cref{app:EAT-Sec-without-Smoothing} there is another parameter, $\eta$, to optimize over. At the end of the proof, we make a natural choice for this parameter. As a result, the parameter $\eta$ is not stated in the above theorem.
    \end{remark}
    
    \begin{remark}
          While not necessary, it seems natural to set $\vAcc = \vEC + \vSec$. As shown in the proof of \Cref{thm:keylengthwithoutsmoothing}, the optimal choice of security parameters will always have $\vAcc \geq \vEC + \vSec$. In principle, one has no control over the input states, so it is necessary to prove the security for many states, which would require $\vAcc$ to be small. The scaling term of $\log(\vAcc)$ in the theorem is small unless $(\beta,\delta)$ is near $(2,1)$, which is clearly suboptimal as the $V$ and $K_{\beta}$ correction terms increase linearly and exponentially in $\beta$ respectively. Moreover, $\delta$ is effectively free as $\alpha \in [1,1.5]$ always. As such, it is reasonable to set $\vAcc = \vEC + \vSec$ as this results in the strongest security claim without changing $\vEC,\vSec$ and, one would expect, obtains similar key rates assuming $\vAcc,\vSec,\vEC$ were not originally significantly different orders.
    \end{remark}

	\begin{remark}
	    As noted in \Cref{sec:EATBackground}, by specific choice of parameters in using the EAT, the resulting bound on the entropy can scale as $nh - O(\sqrt{n})$. As such, with suitable choices of $\beta$ and $\delta$, \Cref{thm:keylengthwithoutsmoothing} gives the the key length that scales as $\ell(n) \leq nh - \leak_{\vEC} - O(\sqrt{n})$. If the min-tradeoff function is chosen appropriately, the key rate will reach the asymptotic key rate in the infinite-key limit.
	\end{remark}
	
	\subsubsection{Fixed number of test rounds}\label{sec:FixedLengthPE}
	
	We have just proven the security of \Cref{prot:physicalDDQKD} by proving the security of \Cref{prot:DDQKD} with the help of entropy accumulation subprotocol (\Cref{prot:DDEA}). However, traditionally device-dependent QKD protocols use fixed-length testing instead of probabilistic round-by-round testing. This leaves us with two options. First, the device-dependent QKD protocol could be altered to do the parameter estimation round by round as is described in \Cref{prot:DDQKD}. In this case, one can use the result of \Cref{thm:keylengthwithoutsmoothing} directly. However, if one wishes to use \Cref{thm:keylengthwithoutsmoothing} and apply it to QKD protocols with fixed-length testing, one must connect the failure probability of \Cref{prot:DDQKD} to the failure probability of the device-dependent QKD protocol actually implemented. Here we state the relation between the two failure probabilities in the case that $T_{i}$ is an independent Bernoulli random variable (e.g. determined by seeded randomness). This can then be used to calculate secure key length using \Cref{thm:keylengthwithoutsmoothing} for protocols with fixed-length testing as explained beneath the following theorem. In \Cref{app:EATtoDDQKDCorrespondence}, we present the derivation of this result. We note that, given the proof method, this result is exact rather than a bound.
	
	\begin{shaded}
	\begin{theorem}\label{thm:EATOutvsTrueOut}
		Let $\rho^{in}_{Q_{1}^{n}E} \in \mathrm{D}(Q_1^{n}E)$ be the input to the protocol. Let $\rho^{out}$ denote the output of \Cref{prot:physicalDDQKD} but with fixed-length parameter estimation on the input state $\rho^{in}_{Q_{1}^{n}E}$. Let $\rho^{out}_{EAT}$ denote the output of the EAT protocol where for each round the probability of testing is $\gamma$. Let $\Omega$ be the event of not aborting on parameter estimation, which to be the same in both protocols, can only accept when there are $m$ tests. Then $\rho^{out}_{EAT}[\Omega] = 2^{-nh(\gamma)} \binom{n}{m}\rho^{out}[\Omega]$. Furthermore, $\rho^{out}_{|\Omega} = \rho^{out}_{EAT|\Omega}$. (See \Cref{sec:notation} to recall notation.)
	\end{theorem}
	\end{shaded}
	
	\begin{proof}
		See \Cref{app:EATtoDDQKDCorrespondence}.
	\end{proof}
	
	In proving security, one wishes to consider the set of inputs $\rho^{in}$ which will be accepted except with probability $\ve$ in the testing. In EAT this probability is $\vEA$ and in fixed-length testing we will say this is $\vPE$. That is, one would like to consider the set of $\rho^{in}$ such that $\rho^{out}[\Omega] \geq \vEA$ (resp. $\vPE$). The above theorem tells us how the set of $\rho^{in}$ that satisfy these conditions changes when going between the considered fixed-length setting and probabilistic round-by-round testing. In other words, if one wishes to consider a protocol with a fixed-length testing that considers all input states that are not $\vPE$-filtered, it suffices to calculate the secure key length of the EAT with $\vEA = 2^{-nh(\gamma)} \binom{n}{m} \vPE$. As \Cref{thm:EATOutvsTrueOut} is tight, this completely closes the gap in this setting. Note, however, that this approach does make the second-order term in \Cref{thm:keylengthwithoutsmoothing} scale closer to that of the first-order term. That is, considering that $\log(\vBar \rho[\Omega])$ is replaced with $\log(\vBar \vEA)$ in applying \Cref{thm:EATv2} for  \Cref{thm:entRate}, we see $\log( \vBar \rho^{out}_{EAT}[\Omega]) \geq \log(\vBar \vPE) + \log( \frac{e^2}{\sqrt{2\pi}} \sqrt{\gamma(1-\gamma)n})$, as is shown in \Cref{app:EATtoDDQKDCorrespondence}.\footnote{One may verify this is the appropriate direction of bound as we are interested in the $\delta$ term of \Cref{thm:keylengthwithoutsmoothing}, and if $y \geq x$, then $n \beta - \sqrt{1-z-y} \geq n\beta - \sqrt{1-z-x}$, so the bound on the key length has only decreased.} This means that the correction term scales as $O(\sqrt{n}\log(\sqrt{n}))$ rather than $O(\sqrt{n})$, which may suggest this is not the ideal way to merge fixed-length testing and the ideas from EAT.

	\section{Construction of (Near-)Optimal Min-Tradeoff Functions}\label{sec:construction_min_tradeoff}
	
	The QKD key rates obtained by the EAT method crucially depend on the choice of min-tradeoff function. For any given block size, it is desirable to choose the min-tradeoff function that maximizes the key rate among all valid min-tradeoff functions. In the infinite-key limit, we would like to choose a min-tradeoff function such that the key rate obtained by the EAT method reproduces the expected asymptotic key rate. We also would like to make our method as general as possible so that it can be applied to a large family of protocols. Our framework can be used whenever the EAT maps have the specific tensor product form as explained in \Cref{subsec:ApplyingEATtoDDQKD}. Our first approach is to use the numerical framework for asymptotic QKD rate calculation \cite{Winick2018} to construct min-tradeoff functions (via a similar two-step procedure). As will be explained in depth later, the important observation here is that the dual problem of the linearization of the original optimization problem gives us the desired min-tradeoff functions. The linearization is a semidefinite program (SDP) and thus its dual problem can be efficiently solved. This method is conceptually simple and can give us a family of valid min-tradeoff functions. We then optimize the choice of min-tradeoff functions when we evaluate the key rate in the finite-key regime using this algorithm. On the other hand, the generation of each individual min-tradeoff function by this approach takes into account only the first-order information in the key rate expression. It is typically the case that the min-tradeoff function that gives the highest first-order term does not give the optimal finite-key rate when lower-order terms are included. This motivates us to derive the second algorithm that also considers the second-order terms. With the aid of Fenchel duality, we show that the second algorithm can also be written in terms of convex optimization.

	As a starting point, we review the asymptotic key rate optimization formulation in \cite{Winick2018}. We present our first algorithm that utilizes the essential idea of \cite{Winick2018} in \Cref{subsec:first_algorithm} and then discuss the second algorithm in \Cref{subsec:second_algorithm}.
	
		\subsection{Review of asymptotic key rate optimization}\label{subsec:review_asymptotic}
	
	To construct min-tradeoff functions, we establish the intimate relation that exists between the problem of generating a good min-tradeoff function for a given protocol and the problem of computing asymptotic key rates in QKD. It is shown \cite{Winick2018,George2021} that the latter problem can be rewritten as a convex optimization problem. The main idea is that, given a map $\wt \cM_i: Q_i \rightarrow S_iP_iX_i$ (from an input quantum system to the key, public announcement and testing registers), the function $r_{\infty}(\vect{q})$ gives the worst-case conditional entropy compatible with the given statistics $\vect{q}$. Explicitly, it is the result of the following convex optimization problem:
	\begin{aeq}\label{eq:asymptotic_primal}
		r_{\infty}(\q) = \minimize_{\rho_{QR}} \quad & W(\rho_{QR}) \\
		 \st & \Tr[\rho_Q M_x] = \q(x)\\
		& \rho_{QR} \geq 0
	\end{aeq}where the objective function
	\begin{IEEEeqnarray}{rL}
		W(\rho_{QR}) := H(S_i|P_iR)_{\widetilde \cM_i \otimes \id_R(\rho_{QR})}
	\end{IEEEeqnarray}
	is defined as the conditional entropy of the state obtained by applying the map $\wt\cM_i \otimes \id_R$ to the state $\rho_{QR}$ which is a purification of $\rho_Q$, and the constraints come from the POVM elements $M_x := \sum_{s,p:t(s,p) = x} M_{sp}$ associated to the map $\Tr_{S_iP_i} \circ \wt \cM_{i}$. Here we use register $R$ to refer to Eve's register in a single round as depicted in \Cref{fig:EATProcess}. The objective function can be written in terms of these POVM elements and without needing to optimize over the $R$ register directly as \Cref{prop:conditional_entropy_form} shows.
	
\begin{prop}\label{prop:conditional_entropy_form}
	Let $\{M_{sp}: s, p\}$ be Alice and Bob's joint POVM which is regrouped according to the public information $p$ and the value of the final key $s$. Let $M_p = \sum_s M_{sp}$. Then for $\rho \in \DM(Q)$,
	\begin{align}\label{eq:objective_conditional_entropy}
		W(\rho_{QR})& = \sum_{s,p} H\left(\cK_{sp}(\rho_{Q})\right) - \sum_p H\left( \cK_p(\rho_{Q})\right)
	\end{align}where $\cK_{sp}(\rho) := K_{sp}\rho K_{sp}^\dagger$ with $K_{sp} = \sqrt{M_{sp}}$,  and $\cK_{p}(\rho) := K_{p}\rho K_{p}^\dagger$ with $K_{p} = \sqrt{M_{p}}$. We note the subscripts highlight that the optimization does not depend on the $R$ register.
\end{prop}
\begin{proof}
        See \Cref{app:key_rate_formula}.
\end{proof}

To solve the convex optimization problem in \Cref{eq:asymptotic_primal}, numerical algorithms typically require the gradient information of the objective function. The gradient of $W$ for $\rho > 0$ is given as
\begin{aeq}
	\nabla W(\rho) &= - \sum_{s,p} \cK_{sp}^{\dagger} \Big(\log \cK_{sp}(\rho)\Big) + \sum_p \cK_p^{\dagger}\Big( \log \cK_p(\rho)\Big),
\end{aeq}where $\cK_{sp}^{\dagger}$ denotes the adjoint map of $\cK_{sp}$ and can be written as $\cK_{sp}^{\dagger}(\rho) := K_{sp}^{\dagger}\rho K_{sp}$. Similarly, $\cK_{p}^{\dagger}$ is the adjoint map of $\cK_{p}$. When $\rho$ is singular, this gradient is not well-defined. We can use the same perturbation technique used in \cite{Winick2018} for the quantum relative entropy formulation to define the gradient for every $\rho \geq 0$. In particular, we denote the depolarizing channel with a depolarizing probability $p$ by $\cD_p$, which is defined as
	\begin{aeq}\label{eq:depolarizing_channel}
		\cD_{p}(\rho) = (1-p) \rho + p \Tr(\rho) \frac{\idm}{d} \, ,
		\end{aeq}where $d$ is the dimension of the Hilbert space relevant for $\rho$. We denote the perturbed version of the objective function as $W_{\epsilon}(\rho)$ with a perturbation $\epsilon$, which is defined as
\begin{align}\label{eq:perturbed_objective_function}
	W_\epsilon(\rho)= \sum_{s,p} H\left(\cK^\epsilon_{sp}(\rho)\right) - \sum_p H\left( \cK^\epsilon_p(\rho)\right),
\end{align}where $\cK^\epsilon_{sp} = \cD_{\epsilon} \circ \cK_{sp} $ and $\cK^\epsilon_{p} = \cD_{\epsilon} \circ \cK_p$. In \Cref{app:key_rate_formula}, we also discuss the continuity of our objective function under this small perturbation. In particular, we have
\begin{align}
	\label{eq:continuity_bound}
	|W(\rho) - W_\epsilon(\rho)| \leq \eta_\epsilon\quad \text{with} \quad \eta_\epsilon = (\abs{S}+1)\abs{P} \epsilon(d-1)\log \frac{d}{\epsilon(d-1)},
\end{align}where $\abs{S}$ and $\abs{P}$ denote the size of alphabets for $S_i$ and $P_i$, respectively.

\subsection{An algorithm derived from the asymptotic numerical optimization}\label{subsec:first_algorithm}
		
\Cref{alg:Algorithm1_min_tradeoff_function} for finding a min-tradeoff function has the same spirit as the algorithm for finding the asymptotic key rate in Ref. \cite{Winick2018}. We prove it provides a valid min-tradeoff function in \Cref{prop:alg1_correctness}. In \Cref{prop:alg1:tightness}, we show that each construction of min-tradeoff function gives us tight asymptotic key rate for the observed statistics $\vect{q}_0$ that we use in the construction. 
\begin{breakablealgorithm}
	\caption{Algorithm for constructing the min-tradeoff functions based on the asymptotic key rate method}
\label[algorithm]{alg:Algorithm1_min_tradeoff_function}
	\textbf{Inputs:} \\[0.1cm]
	\hspace{0.3cm}
	
	\begin{tabular}{ l l}
		$\vect{q}_0$ & A given probability distribution in $\mbP(\cX)$ \\
		$\{M_x: x \in \cX\}$ & Bipartite POVM used for testing
	\end{tabular}
	
	\vspace{0.3cm}
	
	\textbf{Output:} \\[0.1cm]
	\hspace{0.3cm}
	\begin{tabular}{ l l}
		$\vect{y^\star}$ & A vector in $\mbR^\abs{\cX}$ which defines a min-tradeoff function by $f_\epsilon(\vect{q}) := \langle \vect{q},\vect{y^\star}\rangle - \eta_\epsilon$.
	\end{tabular}
	
	\vspace{0.3cm}
	
	\textbf{Algorithm:}
	\begin{enumerate}
		\addtocounter{enumi}{0}
		\itemsep0em
		\item Consider the convex optimization  $r_\epsilon(\vect{q}_0) := \min_{\rho \in \Sigma_i(\vect{q}_0)} W_\epsilon(\rho)$ with the true optimal solution $\rho^{\text{opt}}$. Solve the optimization (e.g. by Frank-Wolfe algorithm) and obtain a near-optimal solution $\rho_\epsilon^{\star}$ with the perturbation error $\epsilon \in (0,1/(e(d-1))]$ determined by the algorithm.
		
		\item Let $W_{\epsilon}^{\text{lin}}(\rho):= W_\epsilon(\rho_\epsilon^\star) + \Tr [\nabla W_\epsilon(\rho_\epsilon^\star) \cdot (\rho - \rho_\epsilon^\star)]$ be the linearization of the function $W_\epsilon$ at the point $\rho_\epsilon^\star$. It can be equivalently written as
		\begin{align}
			\label{eq: algorithm tmp2}
			W_{\epsilon}^{\text{lin}}(\rho) = \Tr[O_\epsilon\rho] \quad \text{with}\quad O_\epsilon:= \left(W_\epsilon(\rho_\epsilon^\star) - \Tr \big[\nabla W_\epsilon(\rho_\epsilon^\star) \cdot \rho_\epsilon^\star\big]\right) \idm +\nabla W_\epsilon(\rho_\epsilon^\star).
		\end{align} 
		Since $W_\epsilon(\rho)$ is a convex function in $\rho$, we know that $W_{\epsilon}^{\text{lin}}(\rho) \leq W_\epsilon(\rho)$, $\forall \rho$, and $W_{\epsilon}^{\text{lin}}(\rho_\epsilon^\star) = W_\epsilon(\rho_\epsilon^\star)$.
		
		\item Consider the SDP $\min_{\rho \in \Sigma_i(\vect{q}_0)} W_{\epsilon}^{\text{lin}}(\rho)$ whose dual SDP is given by
		\begin{align}\label{eq:EAT_algorithm_linearization}
			\max_y\ \langle \vect{q}_0, \vect{y}\rangle \quad \text{subject to} \quad \sum_{x \in \cX} \vect{y}(x) M_x \leq O_\epsilon,\ \vect{y} \in \mbR^{\abs{\cX}}.
		\end{align}
		Solve the dual program and obtain an optimal solution $\vect{y^\star}$.
		
		\item Construct the min-tradeoff function by $f_\epsilon(\vect{q}) := \langle \vect{q},\vect{y^\star}\rangle - \eta_\epsilon$.
	\end{enumerate}
\end{breakablealgorithm}

\begin{remark}
Note that the first term of $O_\epsilon$ in \Cref{eq: algorithm tmp2} always vanishes, i.e., $W_\epsilon(\rho_\epsilon^\star) - \Tr \big[\nabla W_\epsilon(\rho_\epsilon^\star) \cdot \rho_\epsilon^\star\big] = 0$ for any $\rho_\epsilon^\star$. This can be shown by expanding the definitions.
\end{remark}

\begin{remark}[Strong duality for SDP in \Cref{eq:EAT_algorithm_linearization}]\label{rmk: slater condition}
Let $\lambda_{\min}$ be the smallest eigenvalue of $O_\epsilon$. It follows that $(\lambda_{\min} - 1) \idm  < O_\epsilon$. Since $\sum_x M_x = \idm$, $(\lambda_{\min} - 1, \lambda_{\min} - 1, \cdots, \lambda_{\min} - 1)$ is a strictly feasible solution for~\Cref{eq:EAT_algorithm_linearization}. As long as $\Sigma_i(\q_0)$ is non-empty, the Slater's condition is satisfied \cite[Theorem 1.18]{Watrous2018}. Thus, the strong duality holds for the SDP in~\Cref{eq:EAT_algorithm_linearization}. This means that the dual problem in \Cref{eq:EAT_algorithm_linearization} gives the same optimal value as the primal problem which is $\min_{\rho \in \Sigma_i(\vect{q}_0)} W_{\epsilon}^{\text{lin}}(\rho)$.
\end{remark}

\begin{remark}
Note that the min-tradeoff function constructed by \Cref{alg:Algorithm1_min_tradeoff_function} depends on the input choice of $\q_0 \in \PP(\cX)$. In the end we need to optimize over $\q_0$ to get the best key rate (similar to~\cite[Eq.~(32)]{Arnon2018}).
\end{remark}

In the following we show that the function constructed in \Cref{alg:Algorithm1_min_tradeoff_function} is indeed a valid min-tradeoff function. 

\begin{shaded}
\begin{prop}[Correctness]\label{prop:alg1_correctness}
Let $\epsilon \in (0,1/(e(d'-1))]$ and assume that $\sum_i \vect{y^\star}(i) \Gamma_i \leq O_\epsilon$ is satisfied. Then $f_\epsilon$ constructed from \Cref{alg:Algorithm1_min_tradeoff_function} is a valid min-tradeoff function.
\end{prop}
\end{shaded}
\begin{proof}
To check whether $f_\epsilon$ is a valid min-tradeoff function, according to \Cref{defn:mintradeofffunction}, one needs to check that $f_\epsilon$ is an affine function and it satisfies \Cref{eqn:mintradeofffunction}. It is clear that $f_\epsilon$ is a real-valued affine function by construction. It remains to check \Cref{eqn:mintradeofffunction}. For any $\q \in \PP(\cX)$ and any $\rho\in \Sigma_i(\q)$, it holds
\begin{align}
  f_\epsilon(\q)  =\<\q,\vect{y}^\star\> - \eta_\epsilon & = \sum_i \vect{y}^\star(i) \vect{q}(i) - \eta_\epsilon \\
  & = \sum_i \vect{y}^\star(i) \Tr[\rho \Gamma_i] - \eta_\epsilon\\
  & = \Tr \big[\big(\sum_i \vect{y}^\star(i) \Gamma_i \big) \rho\Big] - \eta_\epsilon\\
  & \leq \Tr [O_\epsilon \rho] - \eta_\epsilon\\
  & = W_{\epsilon}^{\lin}(\rho) - \eta_\epsilon\\
  &  \leq W_\epsilon(\rho) - \eta_\epsilon\\
  & \leq W(\rho),
\end{align}
where the second line follows by the assumption that $\rho \in \Sigma_i(\q)$, the third line follows by the linearity of trace, the fourth line follows by the assumption that $\vect{y}^\star$ satisfying the constraint $\sum_i \vect{y}^\star(i)  \Gamma_i \leq O_\epsilon$, the fifth line follows by definition of $W_{\epsilon}^{\lin}$, the sixth line follows as $W$ is a convex function and $W_{\epsilon}^{\lin}$ is a linearization of $W$, and the last line follows by the continuity bound in \Cref{eq:continuity_bound}. Minimizing over all $\rho \in \Sigma_i(\q)$, we have
\begin{align}
  f_\epsilon(\q) \leq r(\q) := \min_{\rho \in \Sigma_i(\q)} W(\rho).
\end{align}
As this holds for any $\q \in \PP(\cX)$, we conclude that $f_\epsilon$ is a valid min-tradeoff function.
\end{proof}

\begin{prop}[Tightness]\label{prop:alg1:tightness}
Let $\epsilon\in (0,1/(e(d'-1))]$ and assume that $\Sigma_i(\q_0)$ is non-empty and the first step of \Cref{alg:Algorithm1_min_tradeoff_function} is solved exactly, i.e., $\rho_\epsilon^\star = \rho^{\opt}$. Then $\lim_{\epsilon \to 0^+} f_\epsilon(\q_0) = r(\q_0)$.
\end{prop}
\begin{proof}
Since $\rho_\epsilon^\star$ is the minimizer of $W_\epsilon$ over the convex set $\Sigma_i(\q_0)$, by~\cite[Lemma 2 and Eq.~(86),  equivalently Eq.~(95)]{Winick2018}, we know that $\min_{\rho \in \Sigma_i(\q_0)} \Tr [\nabla W_\epsilon(\rho_\epsilon^\star) \cdot (\rho - \rho_\epsilon^\star)] = 0$. Thus $ \min_{\rho \in \Sigma_i(\q_0)} W_{\epsilon}^{\lin}(\rho) = W_\epsilon(\rho_\epsilon^\star) = r_\epsilon(\q_0)$. Moreover, since $\Sigma_i(\q_0)$ is non-empty, by~\Cref{rmk: slater condition} the strong duality of \Cref{eq:EAT_algorithm_linearization} holds. Then we have $f_\epsilon(\q_0) + \eta_\epsilon = \langle\q_0,\vect{y}^\star\rangle = \min_{\rho \in \Sigma_i(\q_0)} W_{\epsilon}^{\lin}(\rho) = r_\epsilon(\q_0)$. This implies $\lim_{\epsilon \to 0^+} f_\epsilon(\q_0) = \lim_{\epsilon \to 0^+} r_\epsilon(\q_0) = r(\q_0)$.
\end{proof}

	\subsection{An alternative algorithm that uses second-order information}\label{subsec:second_algorithm}

	 To apply \Cref{alg:Algorithm1_min_tradeoff_function} in the finite-key rate calculation, we need to optimize the choice of min-tradeoff functions by heuristically picking different starting points. As such, while the previous algorithm will reproduce the asymptotic key rate in the infinite-key limit, it may behave poorly in the small block-size regime if the optimization over the starting point $\vect{q}_0$ is not done properly. This limitation motivates us to design a new algorithm that considers the effect of the choice of min-tradeoff function on second-order correction terms when constructing the min-tradeoff function. While ideally we would look at the key length expression in \Cref{thm:keylengthwithoutsmoothing} and collect all terms that depend on the choice of min-tradeoff function $f$ for our new objective function, this would involve an additional optimization over the choices of $\beta$ and $\delta$. This is because the terms $V$ and $K_{\beta}$ depend on both the min-tradeoff function $f$ and the choice of $\beta$, and there are constant terms that depend on both $\beta$ and $\delta$. In principle, one could optimize the min-tradeoff function $f$ and two parameters $\beta$ and $\delta$ simultaneously to obtain the best possible key rate. However, because such a joint optimization is challenging, we choose to consider a simpler scenario that we now explain.
	 
	 As we will claim the R\'{e}nyi entropy key length obtains better key lengths than the smooth min-entropy key length, we build an algorithm that should behave best for the smooth min-entropy key length, \Cref{thm:keyLengthWithSmoothing} [see \Cref{eq:KeyLength_with_smoothing}], which already has one fewer free parameters than in \Cref{thm:keylengthwithoutsmoothing}. To simplify further, we follow \cite{Dupuis2019} in fixing a specific choice of $\alpha$ that leads to a simplified statement of the EAT \cite[Theorem V.2, Eq. (28)]{Dupuis2019}. Again using the fact that $S_i$ is classical, and dividing relevant terms by the number of signals, $n$, we have the following candidate for the objective function to use to to generate a near-optimal min-tradeoff function:
	 
\begin{aeq}\label{eq:eat_min_tradeoff_related_terms_cc_version}
	\cL(f):= f(\vect{q}_0) - \frac{1}{\sqrt{n}}c(f)-\frac{1}{n} c'(f), 
	\end{aeq}where $c(f)$ and $c'(f)$ are defined as
	\begin{aeq}
    c(f) &= \sqrt{2\ln(2)}\Big[\log(2d_S^2+1)+\sqrt{2+\Var{f}}\Big]\sqrt{1-2\log(\vBar\rho[\Omega])},\\
    c'(f) &= \frac{35[1-2\log(\vBar\rho[\Omega])]}{[\log(2d_S^2+1)+\sqrt{2+\Var{f}}]^2}2^{\log(d_S) + \Max{f}-\MinSigma{f}}\ln^3\Big(2^{\log d_S+\Max{f}-\MinSigma{f}}+e^2\Big).
	\end{aeq}

	For a general min-tradeoff function $f$, $\Var{f}$ can be upper bounded by a function of $\Max{f}$ and $\Min{f}$ as 
\begin{aeq}
	\Var{f} \leq \frac{1}{4}[\Max{f}-\Min{f}]^2.
	\end{aeq}We note that in the application of EAT to security proofs of QKD protocols (see \Cref{thm:keyLengthWithSmoothing}), one replaces $\rho[\Omega]$ by $\vAcc$ and uses $\vBar/4$ in the place of $\vBar$. We also note while the term $c'(f)$ has a complicated dependence on the min-tradeoff function $f$, its contribution to the key rate is much smaller than the first two terms of \Cref{eq:eat_min_tradeoff_related_terms_cc_version} due to the $1/n$ dependence. Therefore, for simplicity of our method, we ignore the $c'(f)$ term in our objective function for the purpose of constructing min-tradeoff function. We make another simplification in the $c(f)$ term by dropping the term related to $\log(2d_S^2+1)$ since it does not depend on the min-tradeoff function. With all these simplifications, we would like to consider the following objective function:
	\begin{aeq}\label{eq:alg2_objective_generalf}
	\tilde{\cL}(f)&:= f(\vect{q}_0) - \frac{1}{\sqrt{n}}\sqrt{2\ln(2)}\sqrt{1-2\log(\vAcc\vBar/4)}\sqrt{2+\frac{1}{4}[\Max{f}-\Min{f}]^2}\\
	& = f(\vect{q}_0) - \frac{2}{\sqrt{n}}\sqrt{\ln(2)}\sqrt{1-2\log(\vAcc\vBar/4)}\sqrt{1+\frac{1}{8}[\Max{f}-\Min{f}]^2}.
	\end{aeq}Since a min-tradeoff function $f$ can be fully specified by a vector $\vect{f}$, it is the case that $\Max{f} = \max(\vect{f}) := \max_x \vect{f}(x)$ and similarly, $\Min{f} = \min(\vect{f}) := \min_x \vect{f}(x)$ (see \Cref{eq:definition_minmaxvar} for definitions). This leads to the following optimization problem
\begin{aeq}\label{eq:eat_second_optimization}
	\maximize_{\vect{f}} & \ \vect{f} \cdot \vect{q}_0  - c_0 \sqrt{1+ c_1^2 [\max(\vect{f}) - \min(\vect{f})]^2} \\
	\st 
	& \ \sum_{x} \vect{f}(x) \Tr(\rho M_x) \leq W(\rho) \  \forall \rho \in \DM(Q_i),
	\end{aeq}where $c_0, c_1$ are two constants to be set for generality, and $\{M_x\}$ is the POVM used for testing. In particular, the set of values, $c_0 =2\sqrt{\ln2}\sqrt{1-2\log(\vAcc\vBar/4)}/ \sqrt{n}, c_1=\frac{1}{2\sqrt{2}}$, correspond to the optimization of \Cref{eq:alg2_objective_generalf}. We emphasize that in deriving this simplified expression, we have made a heuristic choice. As we will show later, any vector $\vect{f}$ returned by this optimization gives a valid min-tradeoff function. This means that our heuristic choice does not affect the correctness of a min-tradeoff function. However, it might give sub-optimal min-tradeoff functions that lead to looser key rates. We note that it is possible to make further improvements, particularly for optimizing the min-tradeoff function for \Cref{thm:keylengthwithoutsmoothing}.

The reason that we introduce two constants $c_0$ and $c_1$ is to make our algorithm general enough to allow the construction of crossover min-tradeoff function (which is defined later in \Cref{def:crossover_minTF}) as well as the normal min-tradeoff function in the statements of EAT. If our algorithm is used to find a crossover min-tradeoff function $g$, which can be used to reconstruct a min-tradeoff function $f$ by \Cref{eq:crossover_to_normal_min_tradeoff_1,eq:crossover_to_normal_min_tradeoff_2}, then $\Var{f}$ is upper bounded by $ \frac{1}{\gamma}[\Max{g}-\Min{g}]^2$ according to \Cref{eq: min crossover min relation 4}. Moreover, $g(\vect{q}') =f(\vect{q})$ for every $\vect{q} \in \PP(\cX)$, where $\vect{q}'$ is renormalized after removing the position corresponding to the $\perp$ symbol from $\vect{q}$. Thus, it is the case that the problem for finding a crossover min-tradeoff function still has the form of \Cref{eq:eat_second_optimization}. For crossover min-tradeoff functions, these two constants take the following values: $c_0 = 2\sqrt{\ln2}\sqrt{1-2\log(\vAcc\vBar/4)}/ \sqrt{n}$, $c_1 = 1/\sqrt{2\gamma}$. 

 The optimization problem in \Cref{eq:eat_second_optimization} has an infinite number of constraints that we cannot really handle since the constraint needs to be held for every density operator. However, we can use the Fenchel duality \index{Fenchel duality}(see \Cref{sec:fenchel_duality}) to show it is the dual problem of some primal problem that we can actually solve. Let $\{M_x: x \in \cX\}$ denote the relevant bipartite POVM of a protocol. \Cref{app_sec:algorithm2details} presents a detailed derivation of the primal-dual problem relation including strong duality. The primal problem corresponding to \Cref{eq:eat_second_optimization} is 
\begin{aeq}\label{eq:algorithm2_primal}
	\minimize_{\rho, \vect{\xi}}\ & \  W(\rho)  - \sqrt{c_0^2-\Big[\sum_{x}\vect{\xi}(x)\Big]^2/(4c_1^2)}\\
	\text{subject to } \ & \  \Tr(\rho)=1\\
	& \  -\vect{\xi}(x) \leq \Tr(\rho M_x) - \vect{q}_0(x)\leq \vect{\xi}(x) \\
	& \ \sum_x \vect{\xi}(x) \leq 2c_0 c_1\\
	& \ \rho \geq 0 \ , \ \vect{\xi} \in \mbR^{\abs{\cX}} \ .
	\end{aeq}
	
	We note that the primal problem in \Cref{eq:algorithm2_primal} is very similar to the primal problem in the asymptotic case in \Cref{eq:asymptotic_primal}, but the difference is that the state $\rho$ is not required to reproduce the statistics $\vect{q}_0$ exactly. Instead there is a penalty term in the objective function when $\Tr[\rho M_x] \neq \q_0(x)$ and the additional constraint ensures that the penalty term is well-defined. We also note that one may need to use the perturbed version of $W$. The same perturbation procedure used for \Cref{alg:Algorithm1_min_tradeoff_function} is applicable here for the function $W$. For simplicity of the presentation, we ignore the perturbation here. 
	
	To solve the primal problem in \Cref{eq:algorithm2_primal}, it is often useful to use the gradient information. As the gradient of $W(\rho)$ is already discussed in the previous algorithm (including necessary perturbation), we just write the derivative of $\sqrt{c_0^2-\Big[\sum_{j}\vect{\xi}(j)\Big]^2/(4c_1^2)}$ with respect to $\vect{\xi}(k)$ here as
\begin{aeq}
	\frac{\partial}{\partial \vect{\xi}(k)} \sqrt{c_0^2-\Big(\sum_{j}\vect{\xi}(j)\Big)^2/(4c_1^2)} = -\frac{1}{4 c_0^2 c_1^2} \frac{1}{\sqrt{1-[\sum_j \vect{\xi}(j)]^2/(4c_0^2c_1^2)}} \sum_{j} \vect{\xi}(j) \ .
	\end{aeq}

	We can follow a similar two-step procedure as in \cite{Winick2018} to solve the primal problem in \Cref{eq:algorithm2_primal}. In the first step, we try to obtain a near-optimal solution $\rho^\star$ and in the second step, we solve the dual problem of the linearization of the objective function at the point $\rho^\star$. The dual problem of the linearization at a point $\rho^\star \in \DM(Q_i)$ is
\begin{aeq}\label{eq:algorithm2_linearized_dual}
	\maximize_{\vect{f}} \ & \ \vect{f}   \cdot \vect{q}_0 -  c_0 \sqrt{1+ c_1^2 [\max(\vect{f}) - \min(\vect{f})]^2} \\
	\st \ & \ \sum_{x} \vect{f}(x) M_{x} \leq \nabla W(\rho^\star) \ .
 	\end{aeq}We rewrite this problem as an SDP by introducing slack variables $u, v, t$:
 \begin{aeq}\label{eq:algorithm2_linearized_dualSDP}
 	\maximize_{\vect{f},u,v,t} \ & \ \vect{f} \cdot \vect{q}_0  -t \\
 	\st \ & \  \sum_{x} \vect{f}(x) M_{x} \leq \nabla W(\rho^\star) \\
 	& \ v\vect{1} \leq \vect{f} \leq u\vect{1} \\
 	& \ \begin{pmatrix}
 		t- c_0 & c_0c_1(u-v)\\
 		c_0c_1(u-v) & t+c_0
 		\end{pmatrix} \geq 0 \ .
 	\end{aeq}

	We now present our second algorithm for constructing min-tradeoff function in \Cref{alg:Algorithm2_min_tradeoff_function}.
	\begin{breakablealgorithm}\caption{The second algorithm for constructing min-tradeoff functions}
    \label[algorithm]{alg:Algorithm2_min_tradeoff_function}
	\textbf{Inputs:} \\[0.1cm]
	\hspace{0.3cm}
	
	\begin{tabular}{ l l}
		$\vect{q}_0$ & A given probability distribution in $\mbP(\cX)$ \\
		$c_0, c_1$ & Two constants related to the EAT correction terms \\
		$\{M_x: x \in \cX\}$ & Bipartite POVM used for testing
	\end{tabular}
	
	\vspace{0.3cm}
	
	\textbf{Output:} \\[0.1cm]
	\hspace{0.3cm}
	\begin{tabular}{ l l}
		$\vect{y^\star}$ & A vector in $\mbR^\abs{\cX}$ which defines a min-tradeoff function by $f(\vect{q}) := \langle \vect{q},\vect{y^\star}\rangle $.
	\end{tabular}
	
	\vspace{0.3cm}
	
	\textbf{Algorithm:}
	\begin{enumerate}
		\addtocounter{enumi}{0}
		\itemsep0em
		\item Use either the Frank-Wolfe method or an interior-point method to solve \Cref{eq:algorithm2_primal} and obtain a nearly optimal solution $\rho^\star$.
		\item Solve the dual SDP problem of the linearization at the point $\rho^\star$ in \Cref{eq:algorithm2_linearized_dualSDP} to obtain $\vect{y}^\star=\vect{f}$.
		\end{enumerate}
	\end{breakablealgorithm}

	Our next task is to show that \Cref{alg:Algorithm2_min_tradeoff_function} constructs a valid min-tradeoff function.
	\begin{shaded}
	\begin{prop}[Correctness]\label{prop:alg2_correctness}
	 Assuming that $\sum_{x} \vect{f}(x) M_{x} \leq \nabla W(\rho^\star)$ is satisfied, the min-tradeoff function $f$ constructed from $\vect{f} = \vect{y}^\star$ returned by \Cref{alg:Algorithm2_min_tradeoff_function} is a valid min-tradeoff function.
	\end{prop}
	\end{shaded}
\begin{proof}
        From the assumption  $\sum_{x} \vect{f}(x) M_{x} \leq \nabla W(\rho^\star)$, it follows that for any $\rho \in \DM(Q_i)$,
        \begin{aeq}
        \sum_x \vect{f}(x) \Tr(\rho M_x) \leq \Tr(\rho \nabla W(\rho^\star)) \leq W(\rho),
        \end{aeq}where the last inequality follows from the linearization of the function $W$ at the point $\rho^\star$ since $W$ is a convex function. We note that this is exactly the condition for a valid min-tradeoff function since the left-hand side is the min-tradeoff function evaluated at the statistics produced by a state $\rho$ and the right-hand side is the conditional entropy evaluated at the state $\rho$. As this inequality is true for any state, it follows that $f(\vect{q}) \leq \min_{\nu \in \Sigma_i(\vect{q})}H(S_i|P_iR)_{\nu}$ for every $\vect{q} \in \mbP(\cX)$ such that $\Sigma_i(\vect{q}) \neq \emptyset$. We also note that if  $\Sigma_i(\vect{q}) = \emptyset$, the minimum on the right-hand side of \Cref{eqn:mintradeofffunction} is defined as $\infty$ \cite{Dupuis2016} so that $f(\vect{q}) \leq \min_{\nu \in \Sigma_i(\vect{q})}H(S_i|P_iR)_{\nu} = \infty$ is still satisfied in this case. 
\end{proof}
\begin{remark}
    Due to the numerical precision of any solver, $\sum_{x} \vect{f}(x) M_{x} \leq \nabla W(\rho^\star)$ may not be exactly satisfied. In implementation, we relax this constraint by  $\sum_{x} \vect{f}(x) M_{x} + \epsilon \idm \leq \nabla W(\rho^\star)$ for some small $\epsilon$ that is slightly larger than the solver precision. By doing so, we make the value $f(\vect{q})$ smaller than it could be if there were no numerical precision issue. Thus, the correctness is guaranteed even when one takes into account of the numerical precision. 
\end{remark}
\begin{remark}
    By taking into account some second-order correction terms in the objective function, as we will see later, \Cref{alg:Algorithm2_min_tradeoff_function} can produce similar or better key rates than \Cref{alg:Algorithm1_min_tradeoff_function} \emph{without} optimizing the initial choice of $\vect{q}_0 \in \PP(\cX)$, which can save computational time. Moreover, it is also possible to optimize the choice of $\vect{q}_0$ with this algorithm and choose the best one among all selected choices of $\vect{q}_0$. In addition, as we mentioned previously, this algorithm can be further improved since we made several simplifications and ignored some second-order correction terms. While it is possible to do so, adding back more terms will definitely make the optimization problem more complicated and thus a more sophisticated problem formulation is potentially needed. We leave any potential improvement for a future work.  
\end{remark}

\subsection{Crossover min-tradeoff function}\label{sec:crossover-min-tradeoff}
	In practice, the number of testing rounds is typically chosen to be a small fraction of the total signals sent in the QKD protocol. When the number of test rounds becomes sufficiently smaller than the total number of signals, the original version of the EAT (\Cref{prop:EAT2}) is generally dominated by the second-order term because it scales inversely with the testing probability $\gamma$. A solution for this issue is given in Ref. \cite{Dupuis2019}, where authors of \cite{Dupuis2019} present the `crossover min-tradeoff function', which may be used to induce a proper min-tradeoff function that does not generally become dominated by the second-order term when testing probability $\gamma$ is small. For this reason, it is often advantageous to construct the crossover min-tradeoff function first and then reconstruct a normal min-tradeoff function from the crossover version. We review the definitions from \cite{Dupuis2019} for completeness of our presentation and refer to \cite[Section V.A]{Dupuis2019} for further discussion.
	\begin{definition}[Channel with infrequent sampling]
		A channel with testing probability $\gamma \in [0,1]$ is an EAT channel $\cM_{i,Q_i\to S_iP_iX_i}$ such that $\cX = \cX' \cup \{\perp\}$ and that can be expressed as
		\begin{align}
			\cM_{i,Q_i\to S_iP_iX_i}(\cdot) = \gamma \cM_{i,Q_i\to S_iP_iX_i}^{\testround}(\cdot) + (1-\gamma) \cM_{i,Q_i\to S_iP_i}^{\genround}(\cdot) \ox \ketbra{\perp}{\perp}_{X_i},
		\end{align}
		where $\cM_i^{\testround}$ never outputs the symbol $\perp$ on $X_i$.
	\end{definition}
	
	In our case $\M^{\gen}_{i,Q_i \rightarrow S_iP_i}$ 
	is given by the protocol description where $S_i = \overline A_i$ and $P_i = \tA_i\tB_i$. 
	The testing map $\M^{\testround}_{i,Q_i\rightarrow S_iP_iX_i}$ is given by $S_i = \ol{A}_i\ol{B}_i$ and $P_i = \wt{A}_{i}\wt{B}_{i}$ as per \Cref{prot:DDEA}. 
	
	\begin{definition}[Crossover min-tradeoff function]\label{def:crossover_minTF}
		Let $\cM_i$ be a channel with testing probability $\gamma$ as defined above. The crossover min-tradeoff function for $\cM_{i}$ is an affine function $g: \PP(\cX') \to \mathbb{R}$ satisfying
		\begin{align}
			\label{eq: definition of crossover min tradeoff function}
			g(\vect{q}') \leq \min_{\nu \in \Sigma_i'(\vect{q}')} H(S_i|P_iR)_{\nu} \quad \forall \vect{q}' \in \PP(\cX'),
		\end{align}
		where the set of quantum states
		\begin{align}
			\Sigma_i'(\vect{q}'):= \left\{\nu_{S_iP_iX_iR} = (\cM_i \ox \id_{R})(\omega_{Q_{i}R}): \omega \in \DM(Q_{i}\ox R) \, \& \, [(\cM_i^{\testround} \ox \cI_{R})(\omega_{Q_iR})]_{X_i} = \vect{q}'\right\}.
		\end{align}
	\end{definition}
	We note that the difference between the crossover min-tradeoff function and the original min-tradeoff function defined in \Cref{defn:mintradeofffunction} is that we only require the testing rounds to give the correct frequency distribution. 
	
	For each $x \in \cX$, let $\delta_{x} \in \PP(\cX)$ denote the frequency distribution with $\delta_x(x)=1$ and $\delta_x(x') = 0$ for all other $x' \in \cX$ such that $x' \neq x$. The crossover min-tradeoff function $g$ automatically defines a min-tradeoff function $f: \PP(\cX) \to \mathbb{R}$ by \cite{Dupuis2019}:
	\begin{align}\label{eq:crossover_to_normal_min_tradeoff_1}
		f(\delta_x) & = \Max{g} + \frac{1}{\gamma} [g(\delta_x) - \Max{g}] \qquad \forall x \in \cX'\\
		f(\delta_{\perp}) & = \Max{g} \, . \label{eq:crossover_to_normal_min_tradeoff_2}
	\end{align}
	
	Moreover, we have the relations \cite{Dupuis2019}:
	\begin{align}
		\Max{f} & = \Max{g} \label{eq: min crossover min relation 1}\\
		\Min{f} & = (1-\frac{1}{\gamma}) \Max{g} + \frac{1}{\gamma} \Min{g} \label{eq: min crossover min relation 2}\\
		\MinSigma{f} & \geq \Min{g} \label{eq: min crossover min relation 3}\\
		\Var{f} & \leq \frac{1}{\gamma} [\Max{g}-\Min{g}]^2.\label{eq: min crossover min relation 4}
	\end{align}
	
	\subsection{Procedure for key rate calculation}
	
	We now provide an instruction for the finite-key length $\ell$ calculation using \Cref{alg:Algorithm1_min_tradeoff_function} and \Cref{thm:keylengthwithoutsmoothing}.
	\begin{enumerate}
		\item We first pick a frequency distribution $\vect{q}_0$ and then apply \Cref{alg:Algorithm1_min_tradeoff_function} to construct a crossover min-tradeoff function $g$. By solving the dual SDP of the linearized problem, the algorithm returns us a list of dual variables, which are coefficients of the min-tradeoff function $g$.
		\item We construct the min-tradeoff function $f$ needed for EAT by \Cref{eq:crossover_to_normal_min_tradeoff_1} and \Cref{eq:crossover_to_normal_min_tradeoff_2}, and then compute $\min_{\vect{q} \in \mathcal{Q}} f(\vect{q})$ to get the first-order term $h$. 
		\item We evaluate $\Max{g}$, $\Min{g}$ by simply taking the max and min of coefficients.    
		\item We apply \Cref{thm:keylengthwithoutsmoothing} with the relations in \Cref{eq: min crossover min relation 1,eq: min crossover min relation 2,eq: min crossover min relation 3,eq: min crossover min relation 4} to obtain a lower bound. To do so, we optimize the choice of $\beta$ and $\delta$ in \Cref{thm:keylengthwithoutsmoothing} using MATLAB built-in {\tt fmincon} function. This optimization gives us the optimal second-order correction terms for the given choice of min-tradeoff function. 
		\item We repeat this process with a different frequency distribution $\vect{q}_0$ to generate a different min-tradeoff function. We optimize the choice of min-tradeoff functions in a simple heuristic way by picking several different $\vect{q}_0$'s. 
	\end{enumerate}
	Similarly, we can use \Cref{alg:Algorithm2_min_tradeoff_function} and \Cref{thm:keylengthwithoutsmoothing}. The procedure is similar to the above except that we do not need to optimize the initial choice $\vect{q}_0$ (although one can still do it if it can give a better choice of the min-tradeoff functions). To apply \Cref{thm:keyLengthWithSmoothing}, we also have a similar procedure except that we optimize the choice of $\alpha$ in the statement of \Cref{thm:keyLengthWithSmoothing} with MATLAB built-in {\tt fminbnd} function.

\section{Examples}\label{sec:example}

We first present examples with announcements based on only seeded randomness and then consider an example with more sophisticated announcements. In the first example, we apply our method to the entanglement-based BB84 protocol with an ideal entangled photon source. In the second example, we provide a finite key analysis of six-state four-state protocol \cite{Tannous2019}. We use both these two examples to demonstrate the effectiveness of our method, compare performances of two algorithms and compare two versions of the EAT (smoothed min-entropy versus sandwiched R\'{e}nyi entropy). In the third example, we then show the key rates of high-dimensional  protocols with two mutually unbiased bases (MUBs), i.e.\ analogs of BB84 using qudit systems. We show that for these protocols the Entropy Accumulation Theorem can outperform the postselection technique \cite{Christandl2009}. In the last example, we show the key rates of the entanglement-based BB84 with a realistic entangled photon source. 
	
In all examples, for the purpose of illustration, we {make the following choices in the numerics. We set $\vBar=\vPA=\vEC = \vEA = \frac{1}{4} \times 10^{-8}$. To apply Theorem \ref{thm:keylengthwithoutsmoothing}, we set $\vSec =\vBar+\vPA$ and $\vAcc= \vSec + \vEC$. To apply Theorem \ref{thm:keyLengthWithSmoothing}, we set $\vAcc = \vBar+\vPA + \vEC$. This results in both cases using the same value of $\vAcc$ and thus providing the same security guarantee (namely, guaranteeing the same $\vAcc$-secure key).
In all examples, to estimate the cost of error correction $\leak_{\vEC}$, we set $\leak_{\vEC}=n f_{\EC}H(Z|\hat{Z}) + \log(2/\vEC)$, where $f_{\EC}$ is the inefficiency of an error correction code and $H(Z|\hat{Z})$ is the von Neumann entropy of Alice's key (in a single round) conditioned on Bob's guess $\hat{Z}$. In the simulation, we set $f_{\EC} = 1.16$ for all examples.

We also remark that our implementation allows us to define the acceptance set $\Q$ (see \Cref{prot:physicalDDQKD}) as $\Q = \{F\in \mbP(\cX) : \norm{F-\bar{F}}_1 \leq \xi_t \}$, where $\bar{F}$ is the expected frequency distribution in an honest implementation and $\xi_t$ is the acceptance threshold. As mentioned in \Cref{subsec:completeness}, key rates should decrease as the $\xi_t$ grows. We present an example of this scaling for qubit-based BB84 in \Cref{sec:example_BB84} and in other examples we specify the choice for $\xi_{t}$ we have chosen.

We highlight that we optimize both $\beta$ and $\delta$ in the statement of \Cref{thm:keylengthwithoutsmoothing} when we use the key-length expression from this theorem. Similarly, we optimize the choice of $\alpha$ in the statement of \Cref{thm:keyLengthWithSmoothing} when using it. As discussed previously, this optimization is done with \textsc{Matlab}'s {\tt fmincon} function for the former case and {\tt fminbnd} function for the latter case.
	
\subsection{Qubit-based BB84}\label{sec:example_BB84}	
We apply our method to analyze a simple entanglement-based BB84 example based on the qubit implementation to compare different EAT statements, two algorithm variants for the construction of min-tradeoff functions, and study the scaling of the rate as a function of the size of our acceptance set. We assume that Alice's system and Bob's system are qubits and do not consider loss for this example. 
	
	\subsubsection{Protocol description and simulation}
 We consider the following setup for this protocol: 
 \begin{enumerate}
 \item[(1)] Alice chooses the $Z$ basis with a probability $p_z$ and the $X$ basis with a probability $1-p_z$. Bob chooses to measure in the $Z$ basis with a probability $p_z$ and in $X$ basis with a probability $1-p_z$. 
 \item[(2)] Key-generation rounds are where they both choose $Z$ basis. The testing rounds are where they both choose $X$ basis. They discard rounds with mismatched basis choices. 
 \item[(3)] They perform parameter estimation before error correction. For parameter estimation, we use the phase error POVM $\{E_X, \idm-E_X\}$ where $E_X$ is the $X$-basis error operator. This corresponds to statistics $\{e_x, 1-e_x\}$ where $e_x$ is the $X$-basis error rate. 
	\end{enumerate}
	
	We note that in this protocol setup, the testing probability $\gamma$ is given by the probability that both Alice and Bob choose the $X$ basis, that is, $\gamma = (1-p_z)^2$. The sifting factor for the key rate is $p_z^2$. We consider an efficient version of BB84 \cite{Lo2005efficient} by choosing $p_z$ to be close to $1$. This also corresponds to infrequent testing in our setup. We remark that since basis choices are made based on seeded random numbers and they are chosen independently in each round, their announcements trivially satisfy the Markov condition. 
	
	In our simulation, we use the depolarizing channel to model noises. The simulated state that we use to calculate the observed statistics is
	\begin{equation}\label{eq:rho_sim_dp}
		\begin{aligned}
			\rho^{\text{sim}} = (1-\frac{3}{2}Q) \dyad{\Phi^+}{\Phi^+} + \frac{Q}{2}(\dyad{\Phi^-}{\Phi^-}+\dyad{\Psi^+}{\Psi^+} + \dyad{\Psi^-}{\Psi^-}),
		\end{aligned}
	\end{equation}where $\ket{\Phi^+}, \ket{\Phi^-},\ket{\Psi^+}$ and $\ket{\Psi^-}$ are Bell states, and $Q$ is the quantum bit error rate. The statistics $\vect{q}_0$ that we need to give as an input to the min-tradeoff function construction algorithm is then given by $\vect{q}_0(j) = \Tr(\rho^{\text{sim}}M_j)$ for Alice and Bob's joint POVM $\{M_j\}$.

	\subsubsection{Results}
	
		\begin{figure}[ht]
		\centering
		\includegraphics[width = 0.75\linewidth]{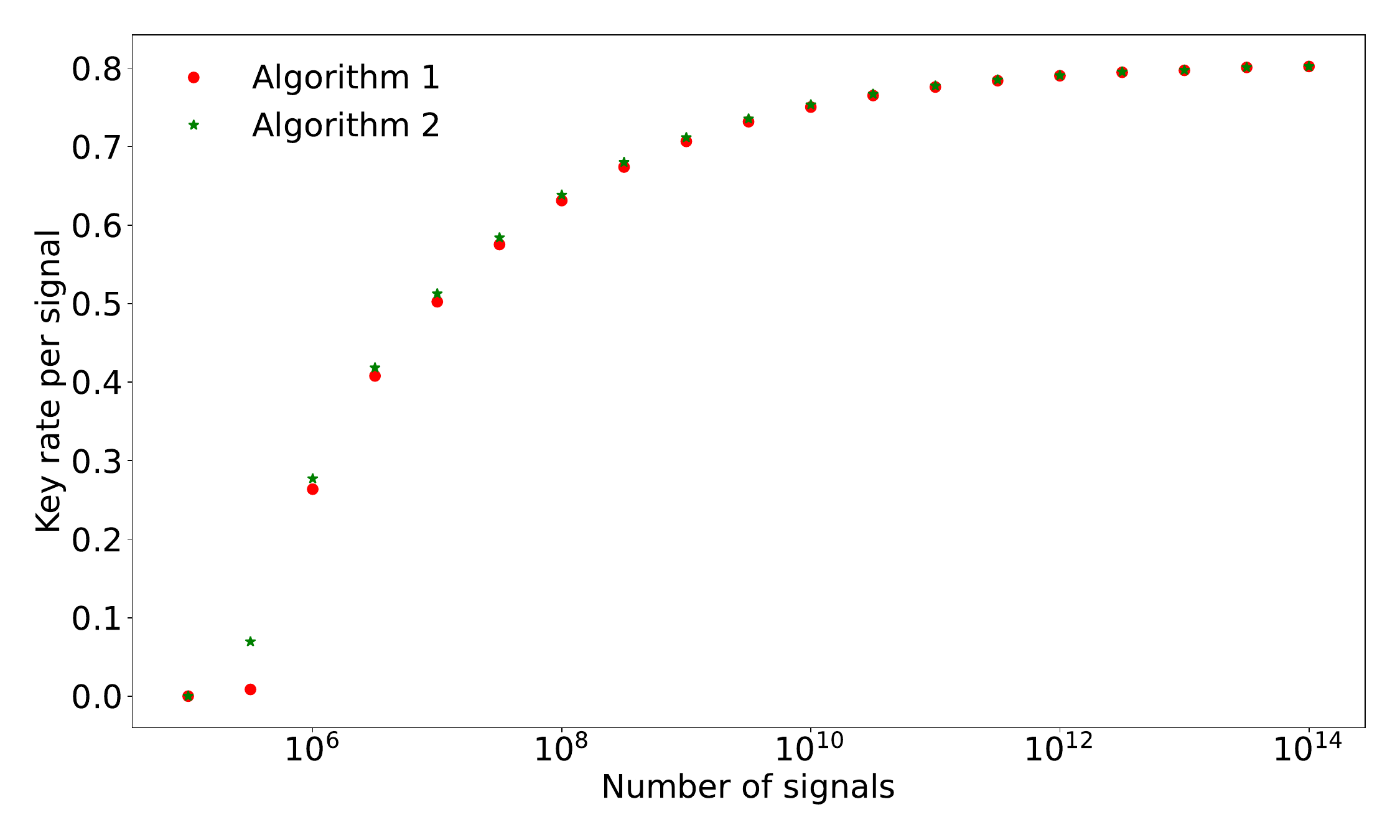}
		\caption{\label{fig:BB84algcomp}  Key rate versus the number of signals for the qubit BB84 protocol to compare two algorithms for the generation of min-tradeoff functions. The quantum bit error rate is set to $Q = 0.01$ in the simulation, and $\zeta_{t} = 0.005$ for defining the acceptance set $\cQ$. The red circle marker corresponds to \Cref{alg:Algorithm1_min_tradeoff_function} while the green star marker corresponds to \Cref{alg:Algorithm2_min_tradeoff_function}. The key rate formula is based on \Cref{thm:keylengthwithoutsmoothing}. Other protocol parameters are optimized as described in the main text.}
	\end{figure}
	
	When applying \Cref{alg:Algorithm1_min_tradeoff_function} to the finite-key rate calculation, we optimize the min-tradeoff functions by choosing different $\vect{q}_0 = (Q, 1-Q)$ where $Q$ is searched over the interval $[0.005, 0.07]$ with a step size $0.005$. For each min-tradeoff function generated from a particular value of $Q$, we calculate the key rate, which is the key length $\ell$ divided by the total number of signals $n$. We then choose the maximum key rate among all possible choices of min-tradeoff functions generated in this way. This way of coarse-grained search over $\vect{q}_0$ is a heuristic approach to optimize the choice of min-tradeoff functions. We find these choices of $\vect{q}_0$ in general give us good results. However, a more sophisticated optimization over $\vect{q}_0$ might potentially improve the results presented here. 
	
	For \Cref{alg:Algorithm2_min_tradeoff_function}, we use the interior-point method from \textsc{Matlab}'s {\tt fmincon} function for the first step and then use \textsc{CVX} for the linearized dual problem. We note that we can also use the Frank-Wolfe algorithm for the first step instead of the and the interior-point method. When we do so, results are similar (slightly worse) for this example.
	
	For both algorithms, we optimize $p_z$ by optimizing $\gamma =(1-p_z)^2 = 10^{-k}$ where $k$ is chosen from the interval $[2,4]$ with a step size of $0.1$. For block sizes larger than or equal to $10^{10}$, we allow $p_z$ to be closer to 1 by searching $k$ in the interval $[3,7]$ with a step size of $0.2$. Again, those choices are heuristic and could be potentially improved. Nevertheless, they are sufficient for our purpose.

	\begin{figure}[h]
		\centering
		\includegraphics[width = 0.7\linewidth]{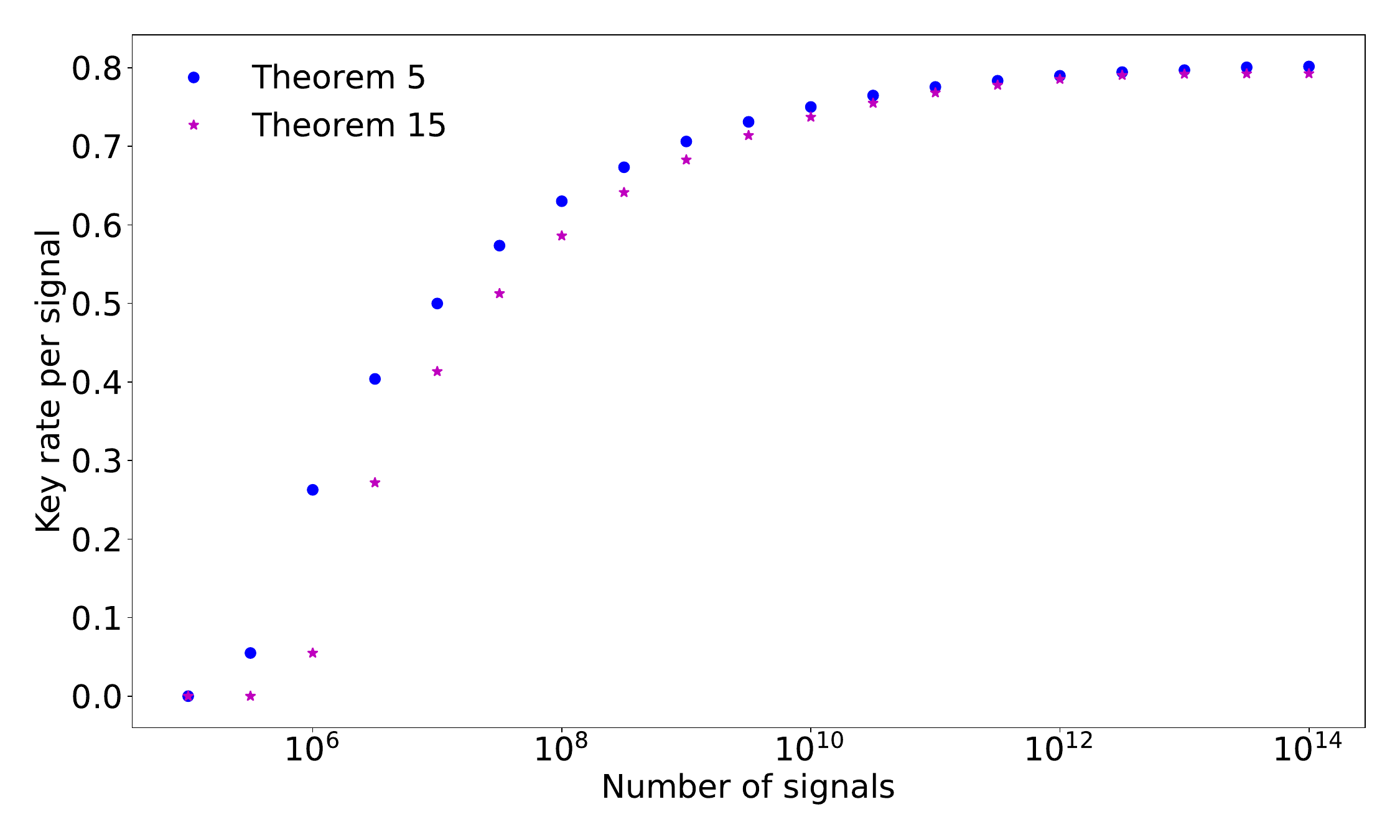}
		\caption{\label{fig:BB84thmcomp} Key rate versus the number of signals for the qubit BB84 protocol to compare different second-order correction terms using \Cref{alg:Algorithm2_min_tradeoff_function}. The quantum bit error rate is set to $Q = 0.01$ in the simulation, and $\zeta_{t} = 0.005$ for defining the acceptance set $\cQ$. The blue circle marker corresponds to the key rate formula given in \Cref{thm:keylengthwithoutsmoothing} \cite{Dupuis2021} while the magenta star marker corresponds to the key rate formula from \Cref{thm:keyLengthWithSmoothing} \cite{Dupuis2019}.}
	\end{figure}
	
	In \Cref{fig:BB84algcomp}, we compare the key rates obtained from these two algorithms with \Cref{thm:keylengthwithoutsmoothing}. Interestingly, both algorithms give similar results while \Cref{alg:Algorithm2_min_tradeoff_function} seems to be slightly better in terms of the smallest number of signals for nonzero key rates. Intuitively, we expect \Cref{alg:Algorithm2_min_tradeoff_function} to behave better as it takes into account some second-order correction terms, while \Cref{alg:Algorithm1_min_tradeoff_function} only looks for the min-tradeoff function that gives the highest leading-order term. As we perform an optimization of the choice of min-tradeoff function by different initial $\vect{q}_0$'s for \Cref{alg:Algorithm1_min_tradeoff_function}, we observe that the optimal finite key rate from \Cref{alg:Algorithm1_min_tradeoff_function} is often given by a min-tradeoff function that does not give the highest value for the leading-order term. Due to the optimization of $\vect{q}_0$, the running time of \Cref{alg:Algorithm1_min_tradeoff_function} is much longer than that of \Cref{alg:Algorithm2_min_tradeoff_function}.
	
	In \Cref{fig:BB84thmcomp}, we compare key rates given by \Cref{thm:keyLengthWithSmoothing,thm:keylengthwithoutsmoothing} when we use \Cref{alg:Algorithm2_min_tradeoff_function}. One can see that \Cref{thm:keylengthwithoutsmoothing} gives better key rates. This confirms our conjecture that the EAT based on the sandwiched $\alpha$-R{\'e}nyi entropy is tighter than the EAT based on the smooth min-entropy for lower-order correction terms. The intuition for this is that the Entropy Accumulation Theorem for smooth entropies is first derived in terms of $H_{\alpha}$ sandwiched R\'{e}nyi entropies and then converted to statements about smooth entropies \cite{Dupuis2016, Dupuis2019}. It follows that avoiding the conversion to smooth entropies should only improve the key rate.

    \begin{figure}[h]
		\centering
		\subfloat[]{\includegraphics[width = 0.5\linewidth]{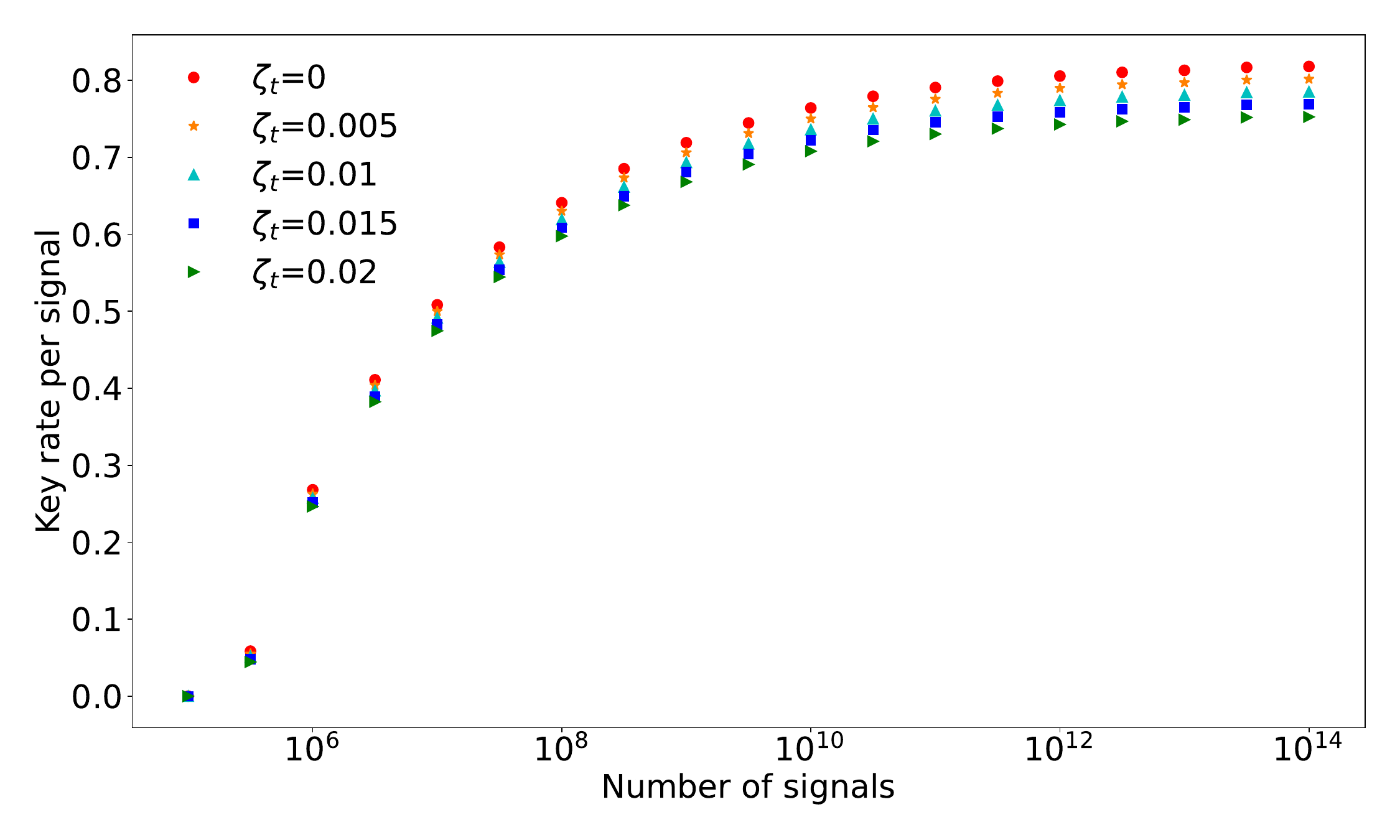}}
        \subfloat[]{\includegraphics[width = 0.5\linewidth]{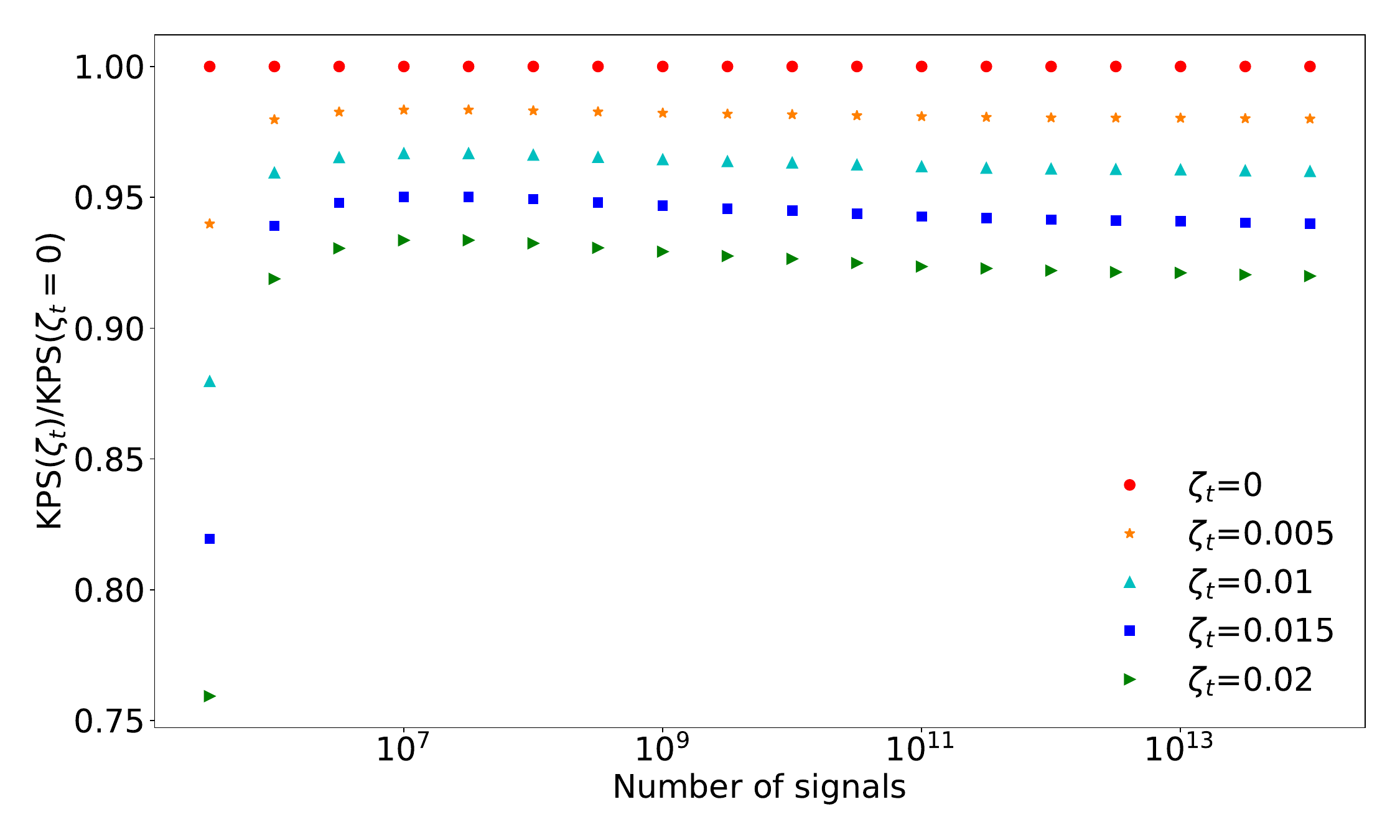}}
		\caption{\label{fig:BB84tplot}(a) Key rate versus the number of signals for the qubit BB84 protocol for different `sizes' of acceptance set controlled by parameter $\zeta_{t}$. (b) The fraction of the key rate per signal (KPS) at $\zeta_{t} = 0$ achieved for other values of $\zeta_{t}$. In both subplots, the quantum bit error rate is set to $Q = 0.01$ in the simulation. Other protocol parameters as described in the main text.}
	\end{figure}

    In \Cref{fig:BB84tplot}, we compare the key rate as a function of the size of the acceptance set, which is relevant when designing practical protocols where there exists a tradeoff between completeness and key rate. In \Cref{fig:BB84tplot}(a), it appears that the fluctuation of key rate per signal is minimal for smaller number of signals and then separates as the signal length grows. In \Cref{fig:BB84tplot}(b), the key rate per signal is normalized to the case $\zeta_{t}=0$ to look at this behaviour more closely. In this case we indeed see that at sufficiently large block size the fraction between each change in $\zeta_{t}$ is relatively constant. However, at the smaller block sizes, there seems to be less of a fluctuation between they key rates. It is possible this has to do with the relevance of the second-order term in small block sizes, but it also may be a property of our algorithm. The ability to investigate such properties with our numerics may be useful in designing protocols with small block size.

	\subsection{Six-state four-state protocol}
	Another interesting example is the six-state four-state protocol \cite{Tannous2019}. In the free space implementation of QKD protocols with the polarization encoding, there is naturally one axis that is stable against turbulence while other axes are slowly drifting. The idea of reference-frame-independent QKD \cite{Laing2010} was motivated to address this issue and it was shown that such a protocol can be robust to slow drifts. The basic idea is that if the reference frame drift can be described by a unitary rotation, by using the information from an additional basis, one can effectively undo this unitary rotation. Here we consider the six-state four-state protocol \cite{Tannous2019} which has the reference-frame-independent feature. In particular, we consider the entanglement-based version and assume Alice and Bob have qubits. We do not consider losses in this example. 
	
	\subsubsection{Protocol and simulation}
	We analyze the entanglement-based version of the six-state four-state protocol \cite{Tannous2019} assuming that Alice and Bob each receive a qubit in each round for simplicity of calculation. In this protocol, Alice measures the state in one of the $X, Y$ and $Z$ bases according to the probability distribution $((1-p_z)/2,(1-p_z)/2, p_z)$, while Bob measures in one of $X$ and $Z$ bases with the probability distribution $(1-p_z, p_z)$. Similarly to the previous qubit BB84 example, when both Alice and Bob choose the $Z$ basis, this round is used for key generation. When Alice chooses $X$ or $Y$ basis and Bob chooses $X$ basis, this round is used for parameter estimation. All other rounds are discarded. We consider an efficient version by setting $p_z$ to be close to $1$. In the testing step of the protocol, we use the POVM that contains error rates in the $XX$ and $YX$ bases.
	
	For the simulation, we assume the $Z$ basis is free of misalignment. The misalignment happens in the $X$-$Y$ plane of the Bloch sphere. Thus, on top of the qubit depolarizing channel, we also apply a unitary rotation along the $Z$ axis to Bob's qubit in order to model the misalignment. We choose the angle of rotation to be  $11^{\circ}$ in the simulation. The state used to obtain the observed statistics from this simulation is 
	\begin{aeq}
	\rho^{\text{sim}} = (\idm_A \ox e^{i\theta \sigma_Z}) \rho^{\text{dp}}(\idm_A \ox e^{-i\theta \sigma_Z}),
	\end{aeq}where $\sigma_Z$ is the Pauli-$Z$ matrix, $\theta$ is the angle of rotation and $\rho^{\text{dp}}$ is the state given in \Cref{eq:rho_sim_dp} (that is, the simulated state in the qubit-based BB84 example).

	\subsubsection{Results}
	
	In the application of \Cref{alg:Algorithm1_min_tradeoff_function} to the finite-key rate calculation, we optimize the min-tradeoff functions by choosing different $\vect{q}_0$'s. We adopt a heuristic approach to optimize the choice of min-tradeoff functions. Each $\vect{q}_0$ is created by choosing a different depolarizing probability $Q$, which is searched over the interval $[0.005, 0.07]$ with a step size $0.005$. It is a heuristic choice that serves our purpose. For each min-tradeoff function generated from a particular value of $Q$, we calculate the key rate and then choose the maximum key rate among them. For \Cref{alg:Algorithm2_min_tradeoff_function}, we use same procedure as for the qubit BB84 example in \Cref{sec:example_BB84} including the optimization over the choice of $p_z$.

	\begin{figure}[ht]
		\centering
		\includegraphics[width = 0.7\linewidth]{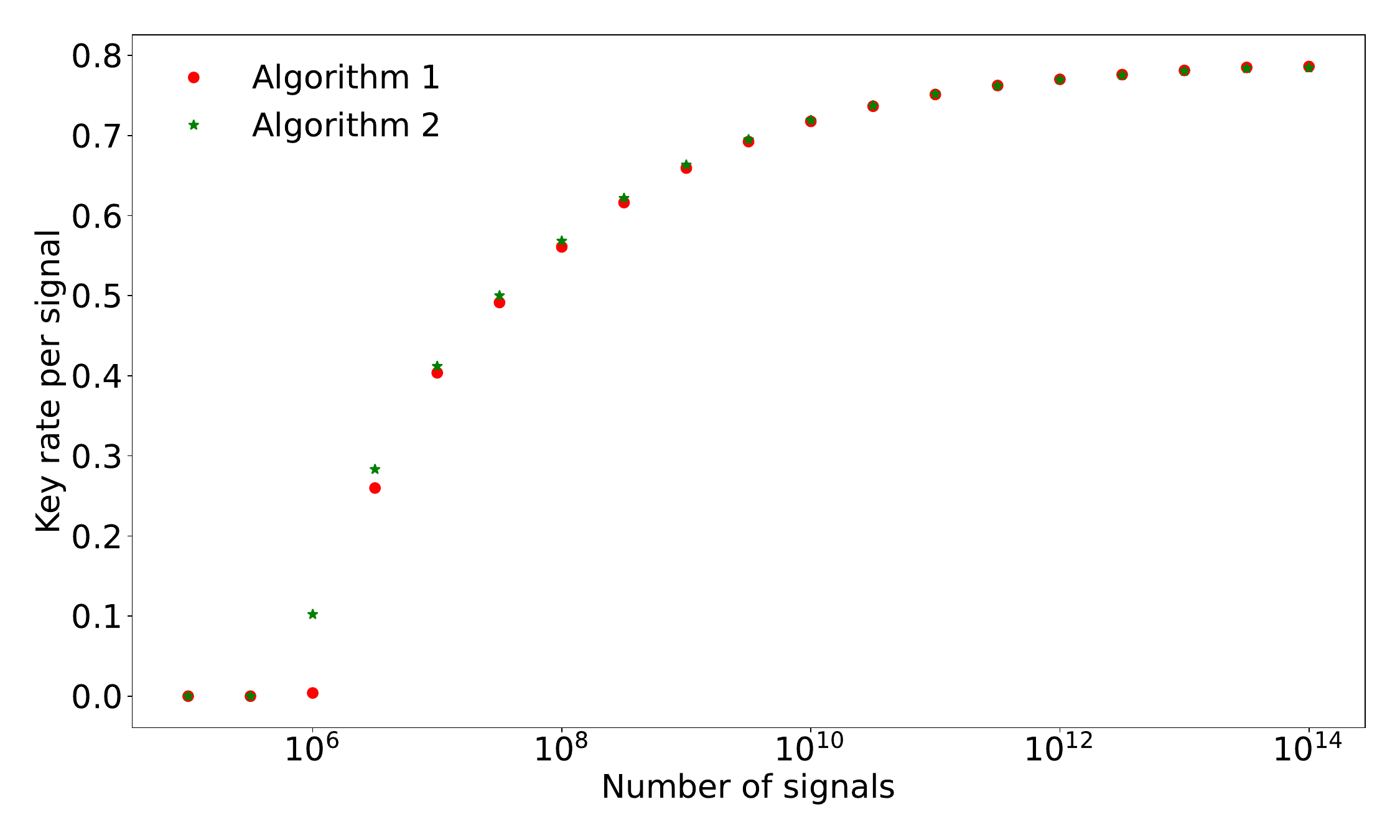}
		\caption{\label{fig:sixfour_alg_cmp} Key rate versus the number of signals for the six-state four-state protocol with two algorithms. The quantum bit error rate is set to $Q = 0.01$ in the simulation, and $\zeta_{t} = 0.005$ for defining the acceptance set $\cQ$. This plot is obtained by using the key rate formula given in \Cref{thm:keylengthwithoutsmoothing} \cite{Dupuis2021}. The red circle marker corresponds to the min-tradeoff function construction by \Cref{alg:Algorithm1_min_tradeoff_function}  while the green star marker corresponds to the min-tradeoff function construction by \Cref{alg:Algorithm2_min_tradeoff_function}.} 
	\end{figure}
	
	In \Cref{fig:sixfour_alg_cmp}, we compare two algorithms for the min-tradeoff function construction. For this plot, similar to qubit-based BB84 example, we perform additional initial point $\q_0$ optimization for \Cref{alg:Algorithm1_min_tradeoff_function} while we do not optimize $\q_0$ for \Cref{alg:Algorithm2_min_tradeoff_function}.  Like \Cref{fig:BB84algcomp}, both algorithms give similar key rates.
	
	\begin{figure}[ht]
		\centering
		\includegraphics[width = 0.7\linewidth]{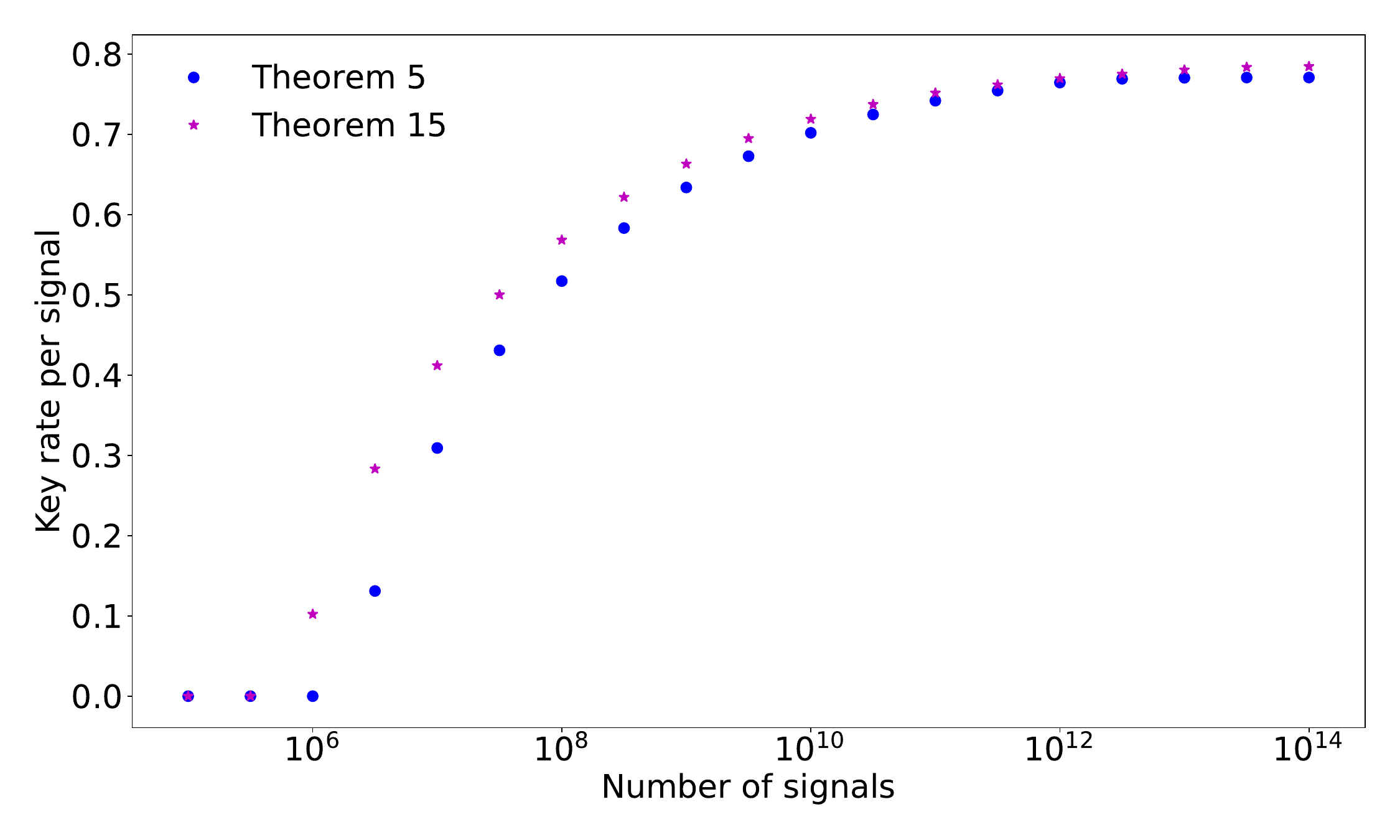}
		\caption{\label{fig:sixfour_thm_cmp} Key rate versus the number of signals for the six-state four-state protocol with different second-order correction terms. The quantum bit error rate is set to $Q = 0.01$ in the simulation, and $\zeta_{t} = 0.005$ for defining the acceptance set $\cQ$. This plot is obtained by using \Cref{alg:Algorithm2_min_tradeoff_function}. The red circle marker corresponds to the key rate formula given in \Cref{thm:keylengthwithoutsmoothing} \cite{Dupuis2021} while the green star marker corresponds to the key rate formula from \Cref{thm:keyLengthWithSmoothing} \cite{Dupuis2019}.}
	\end{figure}
	
	In \Cref{fig:sixfour_thm_cmp}, we compare key rates for \Cref{thm:keyLengthWithSmoothing,thm:keylengthwithoutsmoothing} and observe the same behavior as the qubit-based BB84 example in \Cref{fig:BB84thmcomp}. As explained in the qubit-based BB84 example in \Cref{fig:BB84thmcomp}, this behavior is not surprising since the sandwiched R{\'e}nyi entropy was used in the middle step of the proofs of \Cref{thm:EATv1,thm:EATv2} in Refs. \cite{Dupuis2016,Dupuis2019} before converting to the smooth min-entropy by an additional inequality. One would expect that bypassing the smooth min-entropy gives tighter key rates due to the removal of an inequality.

	\subsection{High-dimensional 2-MUB protocol}
	
	To demonstrate an advantage of our approach in the EAT framework, we analyze an interesting family of protocols which are the high-dimensional analog of the BB84 protocol. In BB84, two mutually unbiased bases (MUBs) are used. We consider qudit systems with 2 MUBs. We compare our calculation with the postselection technique \cite{Christandl2009} combined with the numerical approach of \cite{George2021}. We use this example to demonstrate that EAT can give better key rates compared to the postselection technique \cite{Christandl2009}, especially when the dimension $\dim(\cH_{AB})$ in the protocol is large. The key idea is that the reduction via the post-selection technique involves a security parameter scaling exponentially in the dimension, which leads to a correction to the key length. This correction term does not arise when applying the EAT. While perhaps intuitive, it is unclear how to prove this results in an advantage at all blocklengths and all dimensions and, to the best of our knowledge, this is the first work that has discussed this difference in security proof methods or numerically investigated them.
    
	\subsubsection{Protocol and simulation}
	To properly define the protocol setup, recall that the discrete Weyl operators are defined as
	\begin{aeq}
		U_{jk} = \sum_{s=0}^{d-1}\omega^{sk} \dyad{s+j}{s}
	\end{aeq}for $j,k \in \{0,1,\dots, d-1\}$ where $\omega = e^{2\pi i /d}$ is a $d$th root of unity. We note that $U_{01}$ is the generalized Pauli-$Z$ matrix and $U_{10}$ is the generalized Pauli-$X$ matrix. We define the qudit version of $Z$ and $X$ operators as $Z:=U_{01}$ and $X:=U_{10}$. The generalized Bell states are 
	\begin{aeq}
		\ket{\Phi_{jk}} =\frac{1}{\sqrt{d}} \sum_{s= 0}^{d-1} \omega^{sk} \ket{s, s+j} = \idm \ox U_{jk} \ket{\Phi_{00}}.
	\end{aeq}In the 2-MUB protocol, Alice measures in the eigenbasis of either $U_{01}$ or $U_{10}$. Bob similarly measures in the eigenbasis of either $U_{01}^*$ or $U_{10}^*$. The eigenbasis of the operator $Z:=U_{01}$ is the computational basis $\{\ket{s}: 0\leq s \leq d-1\}$. The eigenbasis of the operator $X:=U_{10}$ is $\{\ket{\psi^X_j}: 0 \leq j \leq d-1\}$ where
	\begin{aeq}
		\ket{\psi^X_j} := \sum_{s=0}^{d-1}\frac{\omega^{-js}}{\sqrt{d}}\ket{s}.
	\end{aeq}They each choose to measure in the $Z$ basis with the probability $p_z$ and choose to measure in the $X$ basis with the probability $1-p_z$. In the classical phase, Alice and Bob announce their basis choices and discard rounds with mismatched bases.  We allow an asymmetric basis choice, i.e., setting $p_z$ close to 1. In this protocol, all public announcements are based on seeded randomness. For simplicity of calculation, the testing rounds are those when they both choose $X$ basis and the key generation rounds are rounds when they both choose $Z$ basis.  
	
	It is interesting to note that the state $\ket{\Phi_{00}} = \frac{1}{\sqrt{d}}\sum_s \ket{s,s}$ is invariant under any $U_{jk} \ox U^*_{jk}$. In an honest implementation, the source is supposed to prepare the state $\ket{\Phi_{00}}$, and then to distribute one half to Alice and the other half to Bob. If the quantum channel is ideal, then they are supposed to obtain perfectly correlated results just like the qubit case.

	Following the reasoning of \cite{Kraus2005,Renner2005information}, Alice's and Bob's joint density operator can be taken as diagonal in the generalized Bell basis: 
	\begin{aeq}
		\rho_{AB}^{\text{BG}}=\sum_{j,k} \lambda_{jk}\dyad{\Phi_{jk}}{\Phi_{jk}}, 
	\end{aeq}where 
	\begin{aeq}
		\lambda_{jk} = \frac{1}{d} \Big(\sum_s q_{1s}^{(sj-k \bmod d)} + q_{01}^{(j)}-1\Big)
	\end{aeq}with $q_{jk}^{(i)}$ being the $i$th entry of the error vector $\vect{q}_{jk}$. See also \cite{Sheridan2010} for a detailed discussion. For our purpose of simulating the frequency distribution $\vect{q}_0$ needed for \Cref{alg:Algorithm1_min_tradeoff_function} or \Cref{alg:Algorithm2_min_tradeoff_function}, we take the simulation state $\rho_{AB}^{\text{sim}}$ used to generate the full statistics as $\rho_{AB}^{\text{sim}}=\rho_{AB}^{\text{BG}}$. 
	
	We follow the simulation in \cite{Sheridan2010} by considering the following observation for error vector $\vect{q}_{jk}$ in each basis $U_{jk}$, which is based on the natural generalization of the qubit depolarizing channel with a depolarizing probability $Q$: 
	\begin{aeq}
		\vect{q}_{jk}(Q) := \{1-Q, Q/(d-1), \dots, Q/(d-1)\}.
	\end{aeq}

	\subsubsection{Results}
	
	\begin{figure}[h]
		\centering
		\subfloat[]{\includegraphics[width =0.5\linewidth]{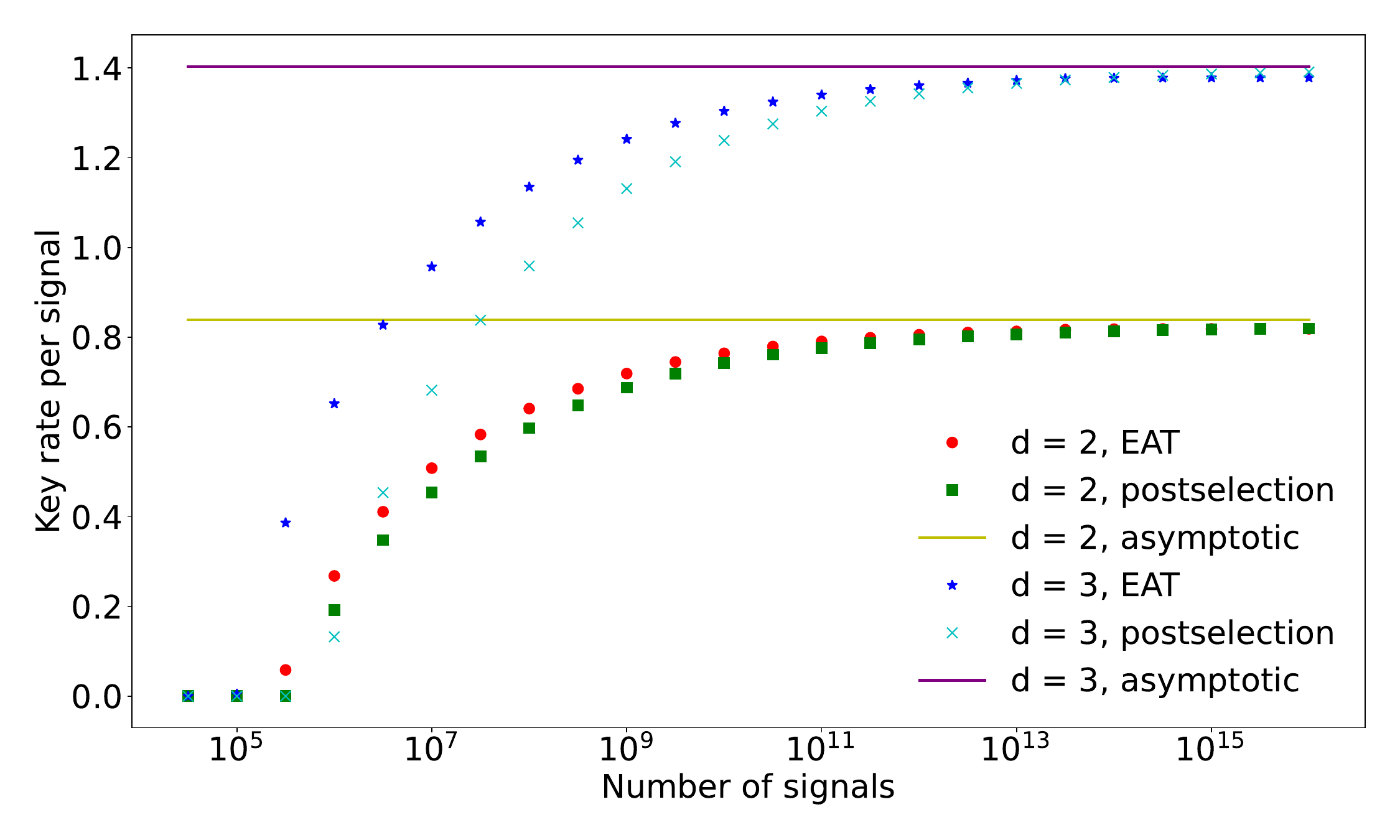}}
		\subfloat[]{\includegraphics[width =0.5\linewidth]{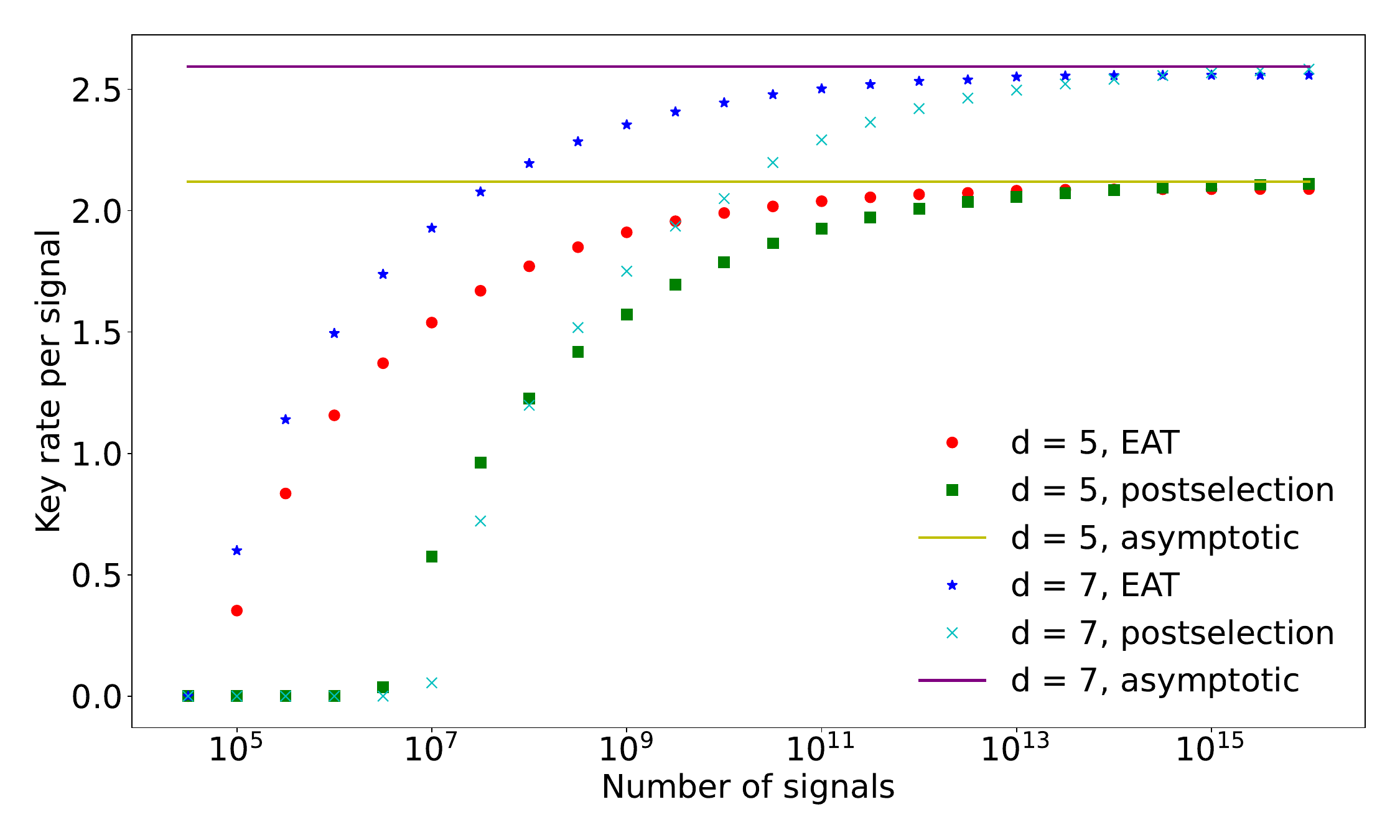}}
		\caption{\label{fig:hd2mub} 2-MUB protocols with qudits for various values of dimension $d$. The depolarizing probability is set to $Q = 0.01$ in the simulation. The calculation is done with \Cref{alg:Algorithm2_min_tradeoff_function} and the finite-key rate formula is from \Cref{thm:keylengthwithoutsmoothing} \cite{Dupuis2021}. The calculation of postselection technique is done using the numerical approach of \cite{George2021} in combination of \cite{Christandl2009}. For simplicity in implementing the post-selection technique code, we use $\zeta_{t}=0$ in these plots. The asymptotic key rate formula is given in \cite{Ferenczi2012,Sheridan2010}.} 
	\end{figure}

	In \Cref{fig:hd2mub}, we compare our results using \Cref{thm:keylengthwithoutsmoothing} \cite{Dupuis2021} with results obtained by the postselection technique \cite{Christandl2009} for 2-MUB protocols in dimensions $d = 2,3,5$ and $7$. For simplicity in implementation of the post-selection technique, we use $\zeta_{t} = 0$ in this simulation, which we stress means the protocol is $\sim1$-complete at reasonable block sizes (See Definition \ref{def:completeness}) and thus not practical.\footnote{This is not to imply the post-selection technique cannot be used more generally.} We note that 2-MUB protocols exist in any dimensions $d \geq 2$ and our proof method can work for any dimension. Here, we restrict to prime dimensions due to our choice of data simulation method. For both EAT and the postselection technique, we optimize the probability of choosing $Z$ basis. The probability of testing is set to $\gamma = (1-p_z)^2$. We optimize the probability of choosing $Z$-basis in the same way as in the qubit-based BB84 example.

	It can be seen that both EAT and postselection technique can approach the expected asymptotic key rate in the infinite-key limit. Also, our method based on EAT outperforms the postselection technique for 2-MUB protocols in any dimensions. We also observe that for larger dimensions, our method can give much higher key rate than the postselection technique for small block sizes. This makes our method attractive since small block sizes are of particular interests for experimental implementations. 
	
	In the same plot, we also show the asymptotic key rate for 2-MUB protocol. The asymptotic key rate formula is given in \cite{Ferenczi2012,Sheridan2010}. Both our method and the postselection technique can approach the asymptotic key rate for sufficiently large block size. We also note that for block size larger than $10^{17}$, there seems to be a small constant deviation from the asymptotic key rate in both postselection technique and our method. This deviation is mainly due to our optimization over $p_z$, which is done by choosing a set of values. The asymptotic key rate formula that we use assumes that we can set $p_z$ to be arbitrarily close to $1$ such that the sifting factor is $1$. On the other hand, numerical optimization over $p_z$ in our method cannot take $p_z$ too close to one. The reason for our EAT method is that our testing probability is set to be $(1-p_z)^2$. When the testing probability is small, the variance of the min-tradeoff function, $\Var{f}$, becomes large as shown in \Cref{eq: min crossover min relation 4}. Since $\Var{f}$ shows up in the second-order correction term, for a fixed block size, there is always a limit on how small the testing probability $\gamma$ can be before we start to lose key rates due to its adverse effect on the second-order correction term. For a fair comparison between our EAT approach and the postselection technique, we also use the same optimization of $p_z$ in the calculation with the postselection technique.

	\subsection{BB84 with a realistic spontaneous parametric downconversion source}\label{subsec:opticalBB84}
 We consider an example where the Markov chain conditions are not simply based on seeded randomness. This example considers an optical implementation of entanglement-based BB84. The photon-pair source is a spontaneous parametric downconversion source where there is a non-negligible probability that the source emits vacuum or more than one photon pair. Due to photon loss during the transmission and non-unity detector efficiencies, there are no-detection events. In the protocol, Alice and Bob announce these events and discard the corresponding rounds. We need to verify that this type of announcements do not violate the Markov chain conditions.  
	
	\subsubsection{Protocol and simulation method}	
	In this protocol, a type-II parametric down-conversion (PDC) source emits a state with polarization encoding \cite{Kok2000, Ma2007} which is 
	\begin{equation}
		\begin{aligned}
			\ket{\Psi} = (\cosh(\chi))^{-2} \sum_{n=0}^{\infty} \sqrt{n+1}\tanh^n(\chi)\ket{\Phi_n}	
		\end{aligned}
	\end{equation}where $\ket{\Phi}_n$ is the state of an $n$-photon pair which can be written as
	\begin{equation}
		\begin{aligned}
			\ket{\Phi_n} = \frac{1}{\sqrt{n+1}}\sum_{m=0}^{n} (-1)^m \ket{n-m,m}_a\ket{m,n-m}_b.
		\end{aligned}
	\end{equation}The average number of photon pairs generated by one pump pulse is $2\lambda$, where $\lambda = \sinh(\chi)$. In this protocol, Alice and Bob each have a set of single-photon detectors. We consider the BB84 detector setup with a passive basis choice; that is, each measurement setup consists of an initial 50/50 beam splitter and each output port of this beam splitter is directed to a polarizing beam splitter with two single-photon detectors. We assume Alice's two detectors have the same detector efficiency $\eta_A$ and the same dark count probability $Y_{0A}$. Similarly, we assume that Bob's two detectors have the same detector efficiency $\eta_B$ and the same dark count probability $Y_{0B}$. Using the similar choice as in other examples, whenever Alice and Bob choose $X$-basis, the round is used for testing. The key generation round is when they both choose $Z$-basis. Their probabilities of choosing $Z$ basis are $p_{zA}$ and $p_{zB}$, respectively. The probability of testing is set to $\gamma = (1-p_{zA})(1-p_{zB})$.

	Our simulation is based on \cite{Ma2007}. In this simulation, there are three main contributions to quantum bit error rate: (i) background counts of detectors, which are random noises with $e_0=1/2$; (ii) intrinsic detector error $e_d$, which is the probability that a photon enters the erroneous detector and is used to characterized the alignment and stability of optical system between Alice's and Bob's detection systems; (iii) errors due to multiphoton-pair states: (a) Alice and Bob may detect different photon pairs and (b) double clicks from detectors. Alice and Bob assign a random bit for each double-click event in order to use the squashing model \cite{Beaudry2008, Gittsovich2014}.

	In particular, the overall gain, $Q_{\lambda}$, as a function of the average number of photons, $\lambda$, in each mode, dark counts and detector efficiencies is given by \cite[Eq. (9)]{Ma2007}
	\begin{equation}
		\begin{aligned}
			Q_{\lambda} = 1 - \frac{1-Y_{0A}}{(1+\eta_A \lambda)} - \frac{1-Y_{0B}}{(1+\eta_B \lambda)} +\frac{(1-Y_{0A})(1-Y_{0B})}{(1+\eta_A \lambda + \eta_B \lambda - \eta_A \eta_B \lambda)^2}. 
		\end{aligned}
	\end{equation}
	
	The overall quantum bit error rate ($E_{\lambda}$) is given by \cite[Eq. (10)]{Ma2007}
	\begin{equation}
		\begin{aligned}
			E_{\lambda}Q_{\lambda} = e_0 Q_{\lambda} - \frac{2(e_0-e_d)\eta_A\eta_B\lambda(1+\lambda)}{(1+\eta_A\lambda)(1+\eta_B\lambda)(1+\eta_A\lambda + \eta_B\lambda - \eta_A \eta_B \lambda)}.
		\end{aligned}
	\end{equation}
	To reduce the number of free parameters in the protocol setup, we set $p_{zA} = p_{zB} =p_z$ and the testing probability is set to $\gamma = (1-p_z)^2$. We optimize the choice of $p_z$ in the same way as in \Cref{sec:example_BB84}.
	
	\subsubsection{Assumption on announcements}
	
	We need to verify for Markov chain conditions that the probability of a detection event only depends on the total photon number, but not on the particular $n$-photon state. We show that this holds when Alice's (Bob's) detectors consist of two single-photon detectors with an identical detection efficiency $\eta_A$ ($\eta_B$), and a basis-independent dark count probability $Y_{0A}$ ($Y_{0B}$). We note that the measurement performed by Bob (Alice) is block diagonal in the total photon number basis, as is the case in all discrete-variable protocols.
	
	Under our assumption about the detectors, we can treat all the imperfections of detectors as a part of the quantum channel and then assume ideal detectors in our analysis. Doing so only strengthens Eve's power. After assigning double-click events to random bits, we note that the measurement setup in this protocol admits a squashing model \cite{Beaudry2008,Gittsovich2014}. In our analysis, we can use the effective qubit measurement for Alice (Bob) as the target measurement. This target measurement acts on the Hilbert space that consists of a one-dimensional vacuum space and a two-dimensional qubit space. In particular, the announcement about detection and no-detection corresponds to the POVM $\{M_{\text{det}},M_{\text{no-det}}\}$, which is defined as
	\begin{aeq}
	M_\text{det} = \begin{pmatrix}
	0 & 0 & 0\\
	0 & 1 & 0\\
	0 & 0 & 1
	\end{pmatrix}, \, 	M_\text{no-det} = \begin{pmatrix}
	1 & 0 & 0\\
	0 & 0 & 0\\
	0 & 0 & 0
	\end{pmatrix}, 
	\end{aeq}where they are represented in the basis $\{\ket{\text{vac}},\ket{0},\ket{1}\}$. Here, $\ket{\text{vac}}$ is the vacuum state and $\ket{0},\ket{1}$ are the computational basis states of a qubit. Clearly, this POVM is weekly dependent according to \Cref{defn:weakly-dependent}. Thus, this announcement is allowed in applying \Cref{thm:keylengthwithoutsmoothing}.

	\subsubsection{Results}

		\begin{figure}[h]
		\centering
		\includegraphics[width =0.7\linewidth]{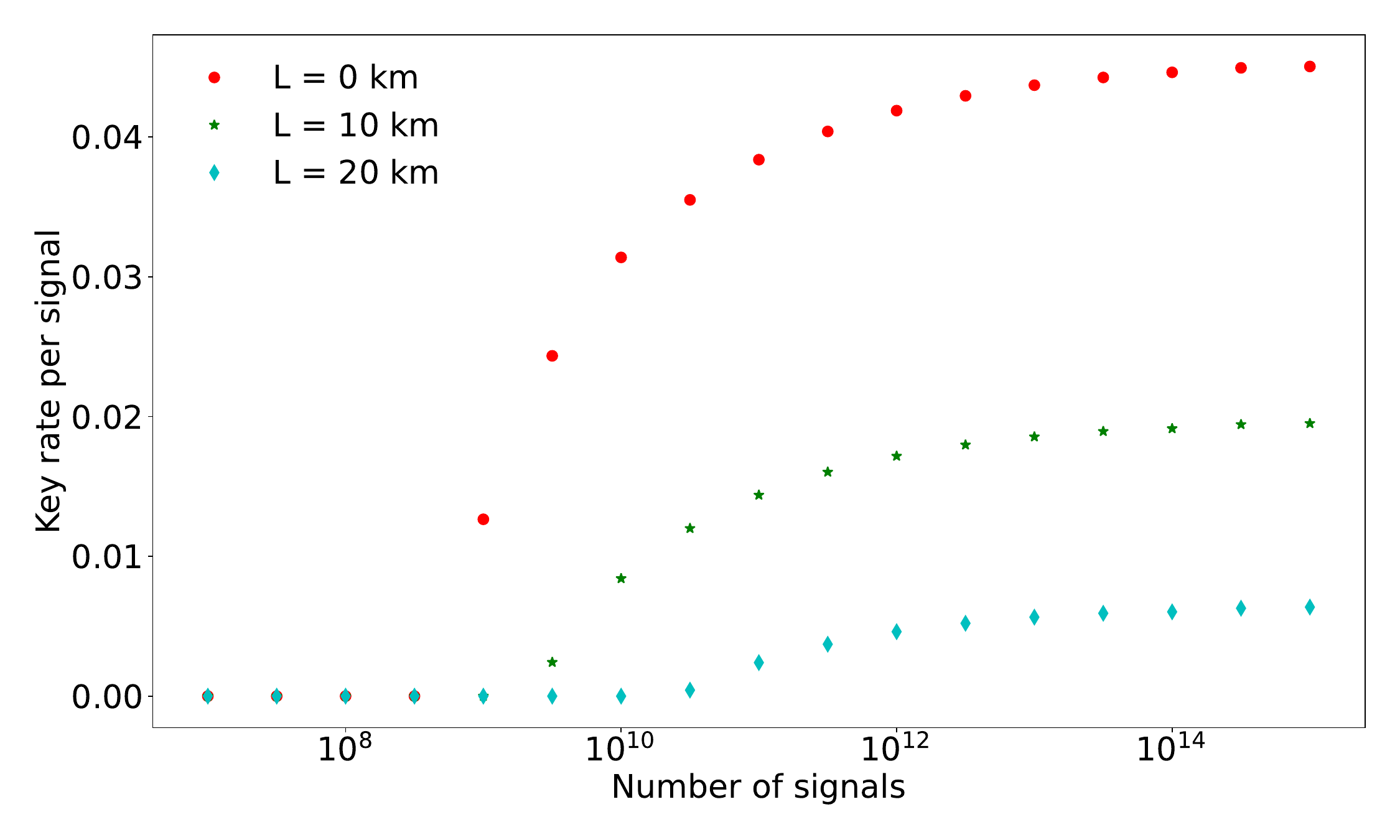}
		\caption{\label{fig:opticalBB84} Key rate per signal sent versus the number of signals for BB84 with a spontaneous parametric downconversion source for various distances between Alice and Bob. This plot is obtained with \Cref{alg:Algorithm1_min_tradeoff_function} and \Cref{thm:keylengthwithoutsmoothing}. The acceptance set threshold $\zeta_t = 0.005$ is used for defining the acceptance set $\cQ$. Parameters are specified in the main text.}
	\end{figure}
	
	In \Cref{fig:opticalBB84}, we show the key rate of this protocol for different distances ($L$ in kilometers) between Alice and Bob. We assume that the source is located in the middle and at an equal distance from Alice and Bob. We choose detector efficiencies $\eta_A = \eta_B = 0.8$ and dark counts $Y_{0A} = Y_{0B} =10^{-7}$ for the simulation in this plot.\footnote{This choice of detector parameters may be realized by superconducting nanowire single-photon detectors. However, the purpose here is to demonstrate that our method can handle some imperfections and loss. Our method also has limitations in the amount of loss it can handle.} We also set the intrinsic detector error $e_d$ as $1\%$. On the one hand, we can see that our method works for optical implementations from this figure. On the other hand, as the distance between Alice and Bob increases, the minimal number of signals for positive key rate also increases significantly, and the key rate drops quickly. In our key rate expression from \Cref{thm:keylengthwithoutsmoothing} (also from \Cref{thm:keyLengthWithSmoothing}), the term $-n\gamma \log(\abs{\cA})$ decreases the key length and becomes more problematic for this protocol since the entropy only accumulates in rounds where both Alice and Bob successfully detected photons. However, this term does not scale with the probability of detection. For long distances, this seems to suggest that the cost for parameter estimation is higher than the entropy accumulated from rounds with successful detection. Unfortunately, this counter-intuitive effect comes from the limitation of our proof method in dealing with parameter estimation registers. We hope that a better approach to handle the information leakage from the parameter estimation step can be found in the future to significantly improve the key rate.

	\section{Discussion and Conclusion}\label{sec:conclusions}
	In this work, we have adapted the EAT to entanglement-based device-dependent QKD protocols. To do so, we introduced new tools. First, we constructed new sufficient conditions on the public announcements of the protocol to guarantee the Markov conditions necessary for the EAT. These conditions capture the intuition that if Eve would always know some information for each round of the protocol, then announcing that information cannot change the security. The interesting point is that this guarantees the Markov conditions on Eve's optimal attack.
	
	Second, we proposed two variants of a numerical algorithm to construct min-tradeoff functions, both of which are efficient and one which considers second-order effects. Both methods build off of previous work \cite{Coles2016,Winick2018}, but the ability to construct min-tradeoff functions that take into account second-order correction information is novel and we expect it could be useful in other settings. We note this second-order correction algorithm relies on using Fenchel duality, which to the best of our knowledge has not been used in quantum information theory previously.
	
	Third, we derived our key length bound (\Cref{thm:keylengthwithoutsmoothing}) using Dupuis' privacy amplification for sandwiched R\'{e}nyi entropies $H^{\uparrow}_{\alpha}$ \cite{Dupuis2021}. In that work, Dupuis demonstrates one can obtain simpler error exponents for privacy amplification using the R\'{e}nyi version of the EAT along with his R\'{e}nyi leftover hashing lemma \cite[Theorem 9]{Dupuis2021} than if one were to apply the smooth min-entropy leftover hashing lemma \cite{Renner2005}. Here we have shown an alternative advantage: by avoiding the conversion of the R\'{e}nyi EAT into smooth min-entropy terms, we can tighten our bound on the key rate.
	
	We then apply our methods to several examples. First, we applied it to ideal qubit-based BB84 and six-state four-state protocols where we show the application of both our min-tradeoff construction algorithms and that using our R\'{e}nyi entropy rate improves the finite-size key rate to considering the smooth min-entropy rate, at least in the current proof method. We next considered the high-dimensional two mutually unbiased bases protocols which exemplified an improvement in the key rate over the postselection technique \cite{Christandl2009}, an alternative proof method for coherent-attack security which is limited in how it scales with respect to the dimension of Alice and Bob's quantum states. This confirms that there is a regime in which the postselection technique is ``loose'', further suggesting the importance of the application of the EAT. Lastly, we demonstrate our method for an optical implementation. This example demonstrates that our method is also applicable to practical protocols instead of restricting to theoretically simple ones. On the other hand, due to unsatisfactory results in small block sizes, we also observe that the EAT currently appears to require more improvements in handling loss and noise in order to have good key rates for experimentally feasible block sizes.
	
	Given these results, it is natural to consider where improvements might be made or further avenues to explore. First and foremost, we note that we were restricted to considering entanglement-based protocols, as we need the entanglement-based protocol framework to use our algorithms. In previous work \cite{George2021}, this has not been limiting as we could use the source-replacement scheme \cite{Curty2004,Ferenczi2012} to convert prepare-and-measure protocols to entanglement-based ones. However, it is not clear how the source-replacement scheme interacts with the Markov chain requirements for EAT. This has been partially resolved in concurrent work \cite{Metgers-2022a,Metgers-2022b}. Second, we have seen that while the EAT scales well in terms of the dimension of the states, it seems limited by loss and noise. In particular, the ability to handling high loss parameter regime is of particular interest for realistic implementations. As such, a natural question is whether the EAT can be more robust to loss and noise, at least in the device-dependent setting. We note follow-up work \cite{Kamin-2024a} has numerical results that show the (generalised) EAT scales worse in loss than for i.i.d. collective attacks. This mathematically makes sense as coherent attacks are more general. However, to the best of our knowledge, no example of such an attack separating the cases is known, so an optimistic researcher could still hope for further improvements.

	\section*{Acknowledgments} 
	\addcontentsline{toc}{section}{Acknowledgements}
	The code used in this work is publicly available at \href{https://github.com/Optical-Quantum-Communication-Theory/Finite-Key-Analysis-of-Quantum-Key-Distribution-with-Characterized-Devices}{this GitHub Repository.} I.G. thanks Fr\'{e}d\'{e}ric Dupuis, Masahito Hayashi, and Marco Tomamichel for helpful discussions. The authors thank Jamie Sikora for helpful suggestions on numerical optimization and thank Ernest Y.~-Z Tan for technical suggestions given an earlier draft. The work has been performed at the Institute for Quantum Computing (IQC), University of Waterloo, which is supported by Innovation, Science and Economic Development Canada. J.L. acknowledges the support of Mike and Ophelia Lazaridis Fellowship from IQC. I.G. acknowledges the support of an Illinois Distinguished Fellowship during this work. K.F. is supported by the National Natural Science Foundation of China (Grant No. 92470113 and 12404569), the Shenzhen Science and Technology Program (Grant No. JCYJ20240813113519025), the Shenzhen Fundamental Research Program (Grant No. JCYJ20241202124023031), the General R\&D Projects of 1+1+1 CUHK-CUHK(SZ)-GDST Joint Collaboration Fund (Grant No. GRDP2025-022), and the University Development Fund (Grant No. UDF01003565). The research has been supported by Natural Sciences and Engineering Research Council of Canada (NSERC) under the Discovery Grants Program, Grant No. 341495, and by NSERC under the Collaborative Research and Development Program, Grant No. CRDP J 522308-17. Financial support for this work has been partially provided by Huawei Technologies Canada Co., Ltd.
		
	\addcontentsline{toc}{section}{References}
	\bibliographystyle{unsrtnat}
	\bibliography{main}
	
	\onecolumn\newpage
	\appendix
	
\section{A Sufficient Condition for the Markov Chain Condition}\label{app:AppendixMarkov}

In this section, we prove \Cref{thm:Markov_block_diag} from the main text which gives sufficient conditions for ensuring the Markov chain conditions hold on the optimal attack. We note the statement in this section (\Cref{prop:Markov_block_diag2}) includes an additional equivalent condition which is quick to verify and so may be of use.

Recall that we consider EAT channels with a special tensor product structure. We consider $n$ CPTP maps $\widetilde\cM_i:Q_i \rightarrow S_iP_iX_i$ acting in tensor product on the $n$ independent systems $\{Q_i\}_i$. The maps are defined by the POVM $\{M_{sp}\}$ such that 
\begin{equation}
	\widetilde\cM_i (\rho) = \sum_{s,p}\Tr[ \rho M_{sp}] \ketbra{s,p}_{S_iP_i} \otimes \ketbra{t(s,p)}_{X_i}\,.
\end{equation}

From these we define the corresponding EAT channels $\cM_i: R_{i-1} \rightarrow R_iS_iP_iX_i$ acting on the quantum systems $R_i = Q_{i+1}^n$, as $\cM_i = \widetilde {\cM}_i \otimes \id_{Q_{i+1}^n}$ (see main text for further details). In this case, it is simple to prove that the output state of the protocol $\rho_{S_1^nP_1^nX_1^nE} = \cM_1\circ \dots \circ \cM_n \otimes \id_E(\rho_{R_0E})$ takes on the simple form
\begin{equation}
	\label{eq:OutputStateTensor}
	\rho_{S_1^nP_1^nX_1^nE} = \bigotimes_{i=1}^n \widetilde \cM_i \otimes \id_E (\rho_{Q_1^nE}) 
\end{equation}
where $R_0 = Q_1^n$.

We now want to prove the following proposition.
\begin{prop}\label{prop:Markov_block_diag2}
	Let $\{M_{sp}: s \in \cS, p\in \cP\}$ be the POVM associated to a given quantum-to-classical CPTP map $\widetilde \cM_i: Q_i\rightarrow S_iP_iX_i$ and let $M_p := \sum_s M_{sp}$ correspond to the POVM elements of the map $\Tr_{S_iX_i}\circ \widetilde \cM_i: Q_i \rightarrow P_i$. If either \cref{condition:A} or \cref{condition:B} below holds for each CPTP map $\widetilde \cM_i$, then without loss of generality the optimal attack by an eavesdropper is of a block-diagonal form such that the Markov chain conditions, \Cref{eq:Markov_cond}, hold and so the EAT may be applied (\Cref{thm:EATv2}). 
	\begin{enumerate}[label=(\Alph*)]
		\item \label[condition]{condition:A} There exists a decomposition $Q_i = \bigoplus_\lambda V^\lambda$ of the space $Q_i$ into orthogonal subspaces $\{V^\lambda\}_\lambda$ such that 
		\begin{itemize}
			\item[(1)] For all $(s,p) \in \cS \times \cP$, $M_{sp}$ is block diagonal:  $M_{sp} = \bigoplus_\lambda M^{(\lambda)}_{sp}$, where $M_{sp}^{(\lambda)}$ acts on the subspace $V^\lambda$.
			\item[(2)] For all $p \in \cP$, $M_p$ is block diagonal and proportional to the identity in each subspace: there exist constants $m_p^\lambda \in [0,1]$ such that $M_p = \bigoplus_\lambda m_p^\lambda \idm_{V^\lambda}$,
			where $\idm_{V^\lambda}$ is the identity operator on $V^\lambda$.
		\end{itemize}
		\item \label[condition]{condition:B} For all $s \in \cS$, $p,p' \in \cP$, $M_{sp}$ and $M_{p'}$ commute: $[M_{sp},M_{p'}] = 0$.
	\end{enumerate}
	Furthermore, \labelcref{condition:A} and \labelcref{condition:B} are equivalent.
\end{prop}
\begin{remark}
    The statement of \Cref{thm:Markov_block_diag} just states \cref{condition:A} in \Cref{prop:Markov_block_diag2} under the name of `weakly dependent' (\Cref{defn:weakly-dependent}), which is of primary interest. While equivalent, \cref{condition:B} is stated here as it could be of use when one knows the measurement operators, as it is then easy to check \cref{condition:B} to determine if \Cref{prop:Markov_block_diag2} may be applied.
\end{remark}

One way to prove the EAT theorem applies under \cref{condition:A} or \cref{condition:B} would be to show the output state in \Cref{eq:OutputStateTensor} satisfies the Markov chain conditions, but this does not work in general. Instead, we show that we can assume the initial state is of a certain block diagonal form without loss of generality due to the block diagonal structure of the EAT channels (\Cref{lemma:block_attack}). Then we show that this block diagonal form satisfies the Markov chain conditions (Proof of \Cref{prop:Markov_block_diag2}), and, as assuming this form was without loss of generality, this is sufficient. The reduction to a quantum state with block-diagonal structure is well-known in discrete-variable QKD where the measurement operators that model single-photon detectors are block diagonal in the total photon number basis. In such a setting, this block diagonal structure implies Eve's optimal attack includes implementing a QND measurement on the total number of photons before sending out the states, thereby resulting in the state being block-diagonal without loss of generality. One may view \Cref{lemma:block_attack} and \Cref{prop:Markov_block_diag2} as a generalization of such a method, but with the further insight that such structure implies the Markov chain conditions hold. 

\begin{lemma}\label{lemma:block_attack}
	Let each $\widetilde \cM_i$ be a quantum-to-classical CPTP map whose associated POVM elements are block diagonal; i.e.,\ for each $i$ there exists a decomposition $Q_i = \bigoplus_\lambda V^\lambda$ such that $M_{sp}  =\bigoplus _{\lambda} M_{sp}^{(\lambda)}$ with $M_{sp}^{(\lambda)}$ acting on $V^\lambda$. Then for any initial state $\rho_{Q_1^nE}$, there exists another state $\nu_{Q_1^nE'}$ such that
	\begin{itemize}
		\item[(1)] the state $\nu_{Q_1^nE'}$ has the block diagonal form
		\begin{equation}
			\label{eq:block_attack}
			\nu_{Q_1^{n}E\Lambda_1^{n}} = \sum_{\vect{\lambda} = (\lambda_1,\cdots \lambda_n)} \rho_{Q_1^nE}{(\vect{\lambda})} \otimes \ketbra{ \vect{\lambda}}_{\Lambda_1^n}\,,
		\end{equation}where Eve's registers $E' = (E,\Lambda_1,\cdots \Lambda_n)$ are composed of a quantum memory $E$ and $n$ classical registers $\Lambda_i$ indicating the subspace, so that, for all $\vect{\lambda}$, the state $\Tr_E[\rho_{Q_1^nE}{(\vect{\lambda})}]$ is defined on $V^{\vect{\lambda}} := \bigotimes_i V^{\lambda_i}$, and
		\item[(2)] the output states $\nu_{S_1^nP_1^nX_1^nE'}$ and $\rho_{S_1^nP_1^nX_1^nE}$ are related by
		\begin{equation}
			\label{eq: block_attack_trace}
			\Tr_{\Lambda_1^n} [\nu_{S_1^nP_1^nX_1^nE\Lambda_1^n}] = \rho_{S_1^nP_1^nX_1^nE}\,.
		\end{equation}
		
	\end{itemize}
\end{lemma}

\begin{proof}
	To construct the state $\nu_{Q_1^nE'}$ from the state $\rho_{Q_1^nE}$, we consider the dephasing map $\Delta_i: Q_i \rightarrow  Q_i\Lambda_i$
	\begin{equation}
		\Delta_i(\rho) = {\sum}_\lambda P_\lambda \rho P_\lambda \otimes \ketbra{\lambda}_{\Lambda_i}\,,
	\end{equation}where $P_\lambda$ is a projector on the subspace $V^\lambda$. This map $\Delta_i$ projects the state onto each subspace $V^\lambda$ without affecting the coherence inside the subspace $V^\lambda$ and writes the result in $\Lambda_i$.  We then define $\nu_{Q_1^nE'} := \nu_{Q_1^nE\Lambda_1^n} := \left(\bigotimes_i \Delta_i\otimes \id_E\right)(\rho_{Q_1^nE})$. The state $\nu_{Q_1^n E \Lambda_1^n}$ is indeed of the form \Cref{eq:block_attack}, as required. 
	
	Because of the block diagonal structure of the maps $\widetilde\cM_i$, it is simple to see that the dephasing map does not affect the measurement statistics. We can observe that the maps $\Delta_i$ and $\widetilde\cM_i$ are related by $\Tr_{\Lambda_i} \circ \widetilde \cM_i \circ \Delta_{i} =  \widetilde \cM_i \circ \Tr_{\Lambda_i} \circ \Delta_{i} = \widetilde \cM_i$. Consequently, we have
	\begin{IEEEeqnarray}{rL}
		\Tr_{\Lambda_1^n} [\nu_{S_1^nP_1^nX_1^nE\Lambda_1^n}] 
		&= \Tr_{\Lambda_1^n} \circ \left(\bigotimes_{i} \widetilde \cM_i \otimes \id_{E\Lambda_1^n}\right)\circ \left(\bigotimes_i \Delta_i \otimes \id_E\right)(\rho_{Q_1^nE})\\
		&= \left(\bigotimes_i \widetilde \cM_i \otimes \id_E\right)(\rho_{Q_1^nE})\\
		&=: \rho_{S_1^nP_1^nX_1^nE}.
	\end{IEEEeqnarray}
\end{proof}

\begin{proof}[Proof of \Cref{prop:Markov_block_diag2}]
	Let $\rho_{Q_1^nE}$ be some initial state used in the EAT (\Cref{thm:EATv2}). We assume that the maps $\widetilde \cM_i$ satisfy the \cref{condition:A} instead of the Markov chain conditions. Let us show that the Markov chain conditions follow from this condition. 
	
	Using the block diagonal structure of \cref{condition:A}, we can use~\Cref{lemma:block_attack} to show that there exists a block diagonal state $\nu_{Q_1^nE'}$ with the same passing probability $\nu[\Omega] = \rho[\Omega]$, which satisfies $H_{\min}^\epsilon(S_i^n|P_i^nE)_{\rho_{|\Omega}} \geq H_{\min}^\epsilon(S_i^n|P_i^nE')_{\nu_{|\Omega}}$. Indeed, this follows by using \Cref{eq: block_attack_trace} and the strong subadditivity of smooth min-entropies. We conclude that whatever min-entropy lower-bound statement holds for $\nu_{Q_1^nE'}$ must also hold for $\rho_{Q_1^nE}$.
	
	To prove \Cref{prop:Markov_block_diag2}, we then show that the new state $\nu_{Q_1^nE}$ satisfies the Markov chain conditions in \Cref{eq:Markov_cond}. We can do this by showing that given the state $\nu_{S_1^{i-1} P_1^{i-1} E'}$, there exists a recovery map $\mathcal R_{P_iE\leftarrow E'}$ which acts on Eve's registers $E' = (E,\Lambda_1^n)$ and recovers the register $P_i$, i.e., $\mathcal R (\nu_{S_1^{i-1} P_1^{i-1} E'}) = \nu_{S_1^{i-1} P_1^{i} E}$.  Since the POVM elements $M_{p}$ are proportional to the identity in every subspace $V^\lambda$, the probability distribution of $P_i$ only depends on the subspace $V^\lambda$, which is indicated in Eve's register $\Lambda_i$. This means the variable $P_i$ can therefore be recovered from $\Lambda_i$ by simply generating it with the distribution $\Pr_{P}(p|\lambda) = m_{p}^\lambda$. Moreover, as we have shown the Markov chain conditions for $\nu_{Q^{n}_{1}E'}$ and without loss of generality, Eve can only do better using a state of the form $\nu$, we can conclude we may use the EAT (\Cref{thm:EATv2}) to prove security.
	
	Having shown that the Markov chain conditions follow from \cref{condition:A}, it remains to show that \labelcref{condition:A} and \labelcref{condition:B} are equivalent. To do so, we first note that the commutation relation $[M_{sp},M_{p'}] = 0$ in \cref{condition:B} is a direct consequence of the block diagonal conditions in \labelcref{condition:A}. Conversely, we note that by construction the POVM elements $M_p$'s are linear combinations of the measurement operators $M_{sp}$'s. Therefore, the commutation relation in \cref{condition:B} implies the commutation relation $[M_p,M_{p'}] = 0$ for all $p,p' \in \cP$. This consequently implies the existence of a joint basis that simultaneously diagonalizes the operators $M_p$'s. We can thus find an orthogonal decomposition $V^\lambda$ so that $M_p = \bigoplus_\lambda m_p^\lambda \idm_{V^\lambda}$. Using the commutation relation in \cref{condition:B} a second time implies that the operators $M_{sp}$ must be block diagonal along the same subspaces $V^{\lambda}$'s.
\end{proof}
\begin{remark}
    We note this proof does not change for R\'{e}nyi entropies as they also satisfy strong subadditivity.
\end{remark}

\section{Security of QKD using EAT without Smoothing}\label{app:EAT-Sec-without-Smoothing}

In this appendix, we present the necessary technical lemmas and the resulting proof for \Cref{thm:keylengthwithoutsmoothing}. The primary tool is Dupuis' recent result on applying privacy amplification without smoothing \cite{Dupuis2021}. However, we will also need to prove a max-tradeoff function upper bound on $H^{\uparrow}_{\alpha}(S_{1}^{n}|P_{1}^{n}E)_{\rho_{|\Omega}}$. We note that in the original work \cite{Dupuis2016} a max-tradeoff function upper bound for $H^{\uparrow}_{\alpha}$ is implied to be straightforward given the max-tradeoff result presented in said work and simply not presented because upper bounding an alternative entropy will result in a tighter bound on $H^{\ve}_{\max}$ \cite[Footnote 13]{Dupuis2016}. For completeness, we present our proof of such a bound. Throughout this appendix we present results using the Greek letters $(\alpha,\beta,...)$ that will be ultimately used in our key length security result rather than always using $\alpha$ for the R\'{e}nyi entropy, so that it is easier to see how the results combine. We refer to \Cref{sec:notation} for a summary of common notations. We do however use $A_{i}$ and $B_{i}$ rather than $S_{i}$ and $P_{i}$ respectively from the main text.

\subsection{Background R{\'e}nyi divergence results}

 Here we present the special cases of the R\'{e}nyi entropy chain rules which we will need in this section. We refer to \Cref{sec:notation} for basic entropic notation used in this work.
\begin{theorem}[{\cite[Theorem 1]{Dupuis2015}}] \label{thm:chain_rule}
	Let $\rho_{ABC} \in \Density(ABC)$, $\tau_{C} \in \Density(C)$, $\alpha,\beta,\gamma \in (1/2,1) \cup (1,\infty)$ such that $\frac{\alpha}{\alpha-1} = \frac{\beta}{\beta-1} + \frac{\gamma}{\gamma-1}$.
	\begin{enumerate}
		\item Case 1: If $(\alpha-1)(\beta-1)(\gamma-1) > 0$, \quad
		$ H_{\beta}^{\uparrow}(A|BC)_{\rho} \leq H_{\alpha}(\rho_{ABC}||\tau_{C}) - H_{\gamma}(\rho_{BC}||\tau_{C}) \ . $
		\item Case 2: If $(\alpha-1)(\beta-1)(\gamma-1) < 0$, \quad
		$ H_{\alpha}(\rho_{ABC}||\tau_{C}) \leq H^{\uparrow}_{\beta}(A|BC)_{\rho} + H_{\gamma}(\rho_{BC}||\tau_{C}) \ . $
	\end{enumerate}
\end{theorem}
\begin{remark}
    In Ref. \cite{Dupuis2015}, the result is presented with the product $(\alpha-1)(\beta-1)(\gamma-1)$ to be bounded by one rather than zero. If one checks the conditions used to derive the theorem (\cite[Propositions 7 and 8]{Dupuis2015}), it only depends on the sign of the product. This has been realized in previous works, e.g.\ \cite{Tomamichel2016}.
\end{remark}

\begin{corollary}[Special Cases of Chain Rules]\label{prop:chain_rule_special_case}\ 
	Let $\rho_{ABC} \in \Density(ABC)$, $\tau_{C} \in \Density(C)$. Let $\alpha,\beta,\delta,\eta,\gamma,\zeta \in (1/2,1) \cup (1,\infty)$. 
	\begin{enumerate}
		\item Let $(\eta-1)(\delta-1)(\zeta-1) > 0$, $\frac{\eta}{\eta-1} = \frac{\delta}{\delta-1} + \frac{\zeta}{\zeta-1}$ and $B$ be a trivial (one-dimensional) register. Then
		$ H_{\delta}^{\uparrow}(A|C)_{\rho} \leq H_{\eta}(\rho_{AC}||\tau_{C}) + D_{\zeta}(\rho_{C}||\tau_{C}) \ . $
		\item Let $(\alpha-1)(\beta-1)(\delta-1) < 0$ and $\frac{\beta}{\beta-1} = \frac{\alpha}{\alpha-1} + \frac{\delta}{\delta-1}$. Then $H^{\uparrow}_{\beta}(AB|C)_{\rho} \leq H^{\uparrow}_{\alpha}(A|BC)_{\rho} + H^{\uparrow}_{\delta}(B|C)_{\rho} \ . $
	\end{enumerate}
\end{corollary}
\begin{proof}
	Item 1 follows from Item 1 of \Cref{thm:chain_rule} by noting if $B$ is trivial, the $B$ register may be dropped. Because $H_{\gamma}(\rho_{BC}||\tau_{C}) = -D_{\gamma}(\rho_{C}|| \idm_{\mathbb{C}} \otimes \tau_{C})$, the negative signs cancel. Dropping the register $B$ and setting $\beta \to \delta$, $\alpha \to \eta$, and $\gamma \to \zeta$ completes the proof. Item 2 follows from Item 2 of \Cref{thm:chain_rule} by letting $\tau_{C}$ be the optimizer for $H^{\uparrow}_{\alpha}(AB|C)_{\rho}$. Then by definition, $H^{\uparrow}_{\gamma}(B|C)_{\rho} \geq H_{\gamma}(\rho_{BC}||\tau_{C})$. Relabeling by $\alpha \to \beta$, $\beta \to \alpha$,  $\gamma \to \delta$ completes the proof.
\end{proof}

\subsection{Applying privacy amplification without smoothing}
In Ref. \cite{Dupuis2021}, it was shown that privacy amplification could be performed without smoothing. While it was stated as holding for strongly 2-universal hashing functions, it is straightforward to show that the proof is the same for a general family of 2-universal hash functions. Recall that a family of 2-universal hash functions is the pair $(p_{F},\cF)$ where $\cF$ is a set of functions from discrete alphabet $\cX$ to discrete alphabet $\cZ$ and $p_{f}$ is a probability distribution over $\cF$ such that $\Pr_{f}[f(x) = f(x')] \leq |\cZ|^{-1}$ for all $x \neq x' \in \cX$. Altering notation to better fit for our purposes, Theorem 8 of Ref. \cite{Dupuis2021} is in terms of $\mathbb{E}_{f \in \cF} \left \| \rho_{ZE}^{f} - \frac{1}{|\{0,1\}^{\ell}|} \mathbb{1}_{Z} \otimes \rho_{E} \right \|_{1}$ where $\rho^{f}_{ZE}$ is $\rho_{XE}$ after the hash function $f$ has been applied to register $X$ and $\cZ = \{0,1\}^{\ell}$. Using that announcing the choice of hash function $f \in \cF$ will make the state block diagonal in the register $F$ that stores this announcement, it is clear that 
\begin{align*}
\|\rho_{ZEF} - \frac{1}{|\{0,1\}^{\ell}|} \mathbb{1}_{Z} \otimes \rho_{EF} \|_{1} &= \sum_{f} p(f) \|\rho^{f}_{ZE} -  \frac{1}{|\{0,1\}^{\ell}|} \mathbb{1}_{Z} \otimes \rho_{E} \|_{1} \\
&= \mathbb{E}_{f} \|\rho^{f}_{ZE} - \frac{1}{|\{0,1\}^{\ell}|} \mathbb{1}_{Z} \otimes \rho_{E} \|_{1} \ , 
\end{align*}
so we state the result in the fashion more commonly used in QKD.
\begin{prop}[{Special Case of \cite[Theorem 8]{Dupuis2021}}]
	Let $\rho_{XE} \in \mathrm{D}(XE)$ be a classical-quantum state where $X \cong \mathbb{C}^{|\cX|}$ where $\cX$ is some finite set. Let $(p_{f},\cF)$ be a family of two-universal hash functions from $\cX \to \cZ$ and $\cZ = \{0,1\}^{\ell}$. Let $\alpha \in (1,2)$. Then
	\begin{align*}
		\left \| \rho_{ZEF} - \frac{1}{|\{0,1\}^{\ell}|} \mathbb{1}_{Z} \otimes \rho_{EF} \right \|_{1} \leq  2^{\frac{2}{\alpha} - 1} \cdot 2^{\frac{1-\alpha}{\alpha} \left( H_{\alpha}^{\uparrow}(X|E)_{\rho} - \ell \right)} \ ,
	\end{align*}
	where $F$ is the register storing the announcement of the hash function.
\end{prop}
\noindent From this we obtain the following corollary immediately.
\begin{corollary}\label{corr:LHL_without_smoothing}
	Let $\rho_{XE{|\Omega}} \in \mathrm{D}(XE)$ be a classical-quantum state. Let $\varepsilon_{sec} := \frac{1}{2} \cdot 2^{\frac{2}{\alpha} - 1} \cdot 2^{\frac{1-\alpha}{\alpha} \left( H_{\alpha}^{\uparrow}(X|E)_{\rho} - \ell \right)}$ where $\alpha \in (1,2)$. Then for an $\varepsilon_{sec}$-secret key, the length of the key, $\ell$, must satisfy the following bound:
	\begin{align*}
		\ell \leq H^{\uparrow}_{\alpha}(X|E)_{\rho} - \frac{\alpha}{\alpha-1}\log(\frac{1}{\vSec}) + 2\ .
	\end{align*}
\end{corollary}
\begin{proof}\begingroup
	Starting from how we defined $\vSec$:
	\allowdisplaybreaks
	\begin{align*}
		& \vSec \leq 2^{\frac{2}{\alpha} - 2} \cdot 2^{\frac{1-\alpha}{\alpha} \left( H_{\alpha}^{\uparrow}(X|E)_{\rho} - \ell \right)} \\
		\Rightarrow & \log(\vSec) \leq \frac{2}{\alpha} - 2 + \frac{1-\alpha}{\alpha} \left( H_{\alpha}^{\uparrow}(X|E)_{\rho} - \ell \right) \\
		\Rightarrow & \frac{1-\alpha}{\alpha}\ell \leq \frac{2}{\alpha} - 2 + \frac{1-\alpha}{\alpha} H_{\alpha}^{\uparrow}(X|E)_{\rho} + \log(\frac{1}{\vSec})\\
		\Rightarrow & \ell \leq H^{\uparrow}_{\alpha}(X|E)_{\rho} + \frac{\alpha}{1-\alpha}\log(\frac{1}{\vSec}) + 2 \ .
	\end{align*}
	Pulling a negative sign out of the $\frac{\alpha}{1-\alpha}$ term completes the proof.
	\endgroup
\end{proof}
\begin{remark}
	Note we added the $\frac{1}{2}$ term in defining $\vSec$ because trace distance should be divided by two for the operational interpretation.
\end{remark}

\subsection{Entropy Accumulation Theorem for upper bound on \texorpdfstring{$H_{\delta}^{\uparrow}$}{Lg}}
In this section we prove an upper bound on $H^{\uparrow}_{\delta}$ in terms of a max-tradeoff function. While almost all of the work is done by \cite{Dupuis2016}, we split it into two pieces as we will need to use a chain rule and also as it allows us to specify what results of \cite{Dupuis2016} we use at each step. As explained earlier, this result is implied to be expected to exist in some form in \cite[Footnote 13]{Dupuis2016}. We present it in its entirety for rigour.

As we switched the secret and public register notation, we emphasize in the current notation the Markov conditions are:
\begin{align}\label{eqn:Markov_cond_alt}
		A_1^{i-1} \leftrightarrow B_1^{i-1}  E \leftrightarrow B_{i} \qquad \forall i \in \{1, \ldots, n\} \ .
\end{align}
We proceed by using the original entropy price function which we know is sub-optimal in the second-order terms, but that will not matter for our purposes. For every $i$, we use a CPTP map $\mathcal{D}_{i}: X_{i} \to X_{i}D_{i}\overline{D}_{i}$
$$ \mathcal{D}_{i}(W_{X_{i}}) = \sum_{x \in \cX} \bra{x} W_{X_{i}} \ket{x} \dyad{x} \otimes \tau(x)_{D_{i}\overline{D}_{i}} \ , $$
where $\tau(x)_{D_{i}\ol{D}_{i}}$ is such that $H(D_{i}|\ol{D}_{i})_{\tau(x)} = \overline{g} -f(\delta_{x})$, $f$ is a min-tradeoff function, $[g_{\min},g_{\max}]$ is the smallest real interval containing $f(\PP(\cX))$ and $\overline{g} := (g_{\min} + g_{\max})/2$. This is possible by letting $\tau(x)_{D_{i}\overline{D}_{i}}$ be a mixture of the maximally entangled state and the maximally mixed state. Following the argument from \cite[Beginning of Proof of Proposition 4.5]{Dupuis2016}, this can be done if $|D_{i}| \leq \exp(\lceil \|\nabla f\|_{\infty} \rceil)$. We also note the general definition $\ol{\rho} := (\cD_{n} \circ \hdots \circ \cD_{1})(\rho)$. With all of this specified, we begin with our equivalent of Claim 4.6 of \cite{Dupuis2016}.

\begin{lemma}\label{lem:claim46}
	Let $\delta \in [1/2,\infty)$. Assuming the Markov chain conditions are satisfied [\Cref{eqn:Markov_cond_alt}], then
	$$ H^{\uparrow}_{\delta}(A^{n}_{1}|B^{n}_{1}E)_{\rho_{|\Omega}} \leq H^{\uparrow}_{\delta}(A^{n}_{1}D^{n}_{1}|B^{n}_{1}E\ol{D}^{n}_{1})_{\overline{\rho}_{|\Omega}} + nh - n\ol{g} \ ,$$
	where $h \in \mathbb{R}$ that satisfies the conditions assumed in \Cref{thm:EATv2} and $\ol{g}$ is as defined above.
\end{lemma}
\begin{proof}
	First note that
	$$ \overline{\rho}_{A^{n}_{1}X^{n}_{1}B^{n}_{1}E\overline{D}^{n}_{1}|\Omega} = \frac{1}{\rho[\Omega]} \sum_{x^{n}_{1} \in \Omega} \dyad{x^{n}_{1}} \otimes \rho_{A^{n}_{1}B^{n}_{1}E,x^{n}_{1}} \otimes \Tr_{D^{n}_{1}}(\tau(x^{n}_{1})_{D^{n}_{1}\ol{D}^{n}_{1}}) \ . $$
	Note that $\Tr_{D_{i}}(\tau(x_{i})_{D_{i}\ol{D}_{i}}) = \frac{1}{\abs{\ol{D}_{i}}} \idm_{\ol{D}_{i}}$ for all $i$ and $x_{i}$ because $\tau(x_{i})$ is a mixture of the maximally mixed and maximally entangled state. It follows we know that register is independent of the rest of the registers and that
	\begin{align}\label{eq:marginal_indep} \overline{\rho}_{A^{n}_{1}X^{n}_{1}B^{n}_{1}E\overline{D}^{n}_{1}|\Omega} = \frac{1}{\rho[\Omega]} \sum_{x^{n}_{1} \in \Omega} \dyad{x^{n}_{1}} \otimes \rho_{A^{n}_{1}B^{n}_{1}E,x^{n}_{1}} \otimes \overline{\rho}_{\ol{D}^{n}_{1}} \ .
	\end{align}
	It follows the state is invariant under tracing off $\overline{\rho}_{\ol{D}^{n}_{1}}$ and re-appending it. Denote this map $\mathcal{R}$. By data-processing,
	\begin{align}\label{lemma1-entprice1:ineq}
		-D_{\delta}(\rho_{A^{n}_{1}X^{n}_{1}B^{n}_{1}E\ol{D}^{n}_{1}}||\idm_{A^{n}_{1}X^{n}_{1}}\otimes \sigma_{B^{n}_{1}E\ol{D}^{n}_{1}})_{\ol{\rho}_{|\Omega}} \, \leq & \, -D_{\delta}(\mathcal{R}(\rho_{A^{n}_{1}X^{n}_{1}B^{n}_{1}E\ol{D}^{n}_{1}})||\mathcal{R}(\idm_{A^{n}_{1}X^{n}_{1}}\otimes  \sigma_{B^{n}_{1}E\ol{D}^{n}_{1}}))_{\ol{\rho}_{|\Omega}} \nonumber \\
		= & \, 
		-D_{\delta}(\rho_{A^{n}_{1}X^{n}_{1}B^{n}_{1}E\ol{D}^{n}_{1}}||\idm_{A^{n}_{1}X^{n}_{1}} \otimes \ol{\sigma}_{B^{n}_{1}E\ol{D}^{n}_{1}})_{\ol{\rho}_{|\Omega}} \ ,
	\end{align}
	where $\ol{\sigma}_{B^{n}_{1}E\ol{D}^{n}_{1}}= \sigma_{B^{n}_{1}E} \otimes \ol{\rho}_{\ol{D}^{n}_{1}}$. This implies that $H^{\uparrow}_{\delta}(A^{n}_{1}X^{n}_{1}|B^{n}_{1}E\ol{D}^{n}_{1})_{\overline{\rho}_{|\Omega}}$ is obtained by a maximizer of the form $\ol{\sigma}_{B^{n}_{1}E\ol{D}^{n}_{1}}$. 
	
	Next we recall Lemma 3.1 of \cite{Dupuis2016} (with terms moved around):
	$$ -D_{\delta}(\rho_{A_{1}B}||\mathbb{1}_{A_{1}} \otimes \sigma_{B}) +  H_{\delta}(A_{2}|A_{1}B)_{\nu} = -D_{\delta}(\rho_{A_{1}A_{2}B}||\mathbb{1}_{A_{1}A_{2}} \otimes \sigma_{B})  \ , $$
	for any $\sigma_{B}$ where 
	\begin{align*}
		\nu_{A_{1}A_{2}B} = \nu^{\frac{1}{2}}_{A_{1}B}\rho_{A_{2}|A_{1}B}\nu^{\frac{1}{2}}_{A_{1}B} \, \text{ where } \, \nu_{A_{1}B} = \frac{\left(\rho^{\frac{1}{2}}_{A_{1}B}\sigma_{B}^{\frac{1-\alpha}{\alpha}}\rho^{\frac{1}{2}}_{A_{1}B} \right)^{\alpha}}{\tr\left(\rho^{\frac{1}{2}}_{A_{1}B}\sigma_{B}^{\frac{1-\alpha}{\alpha}}\rho^{\frac{1}{2}}_{A_{1}B} \right)^{\alpha}} \ .
	\end{align*} Making the substitutions $A_{1} = A^{n}_{1}X^{n}_{1}$, $A_{2} = D^{n}_{1}$, $B=B^{n}_{1}E\overline{D}^{n}_{1}$ and letting $\sigma = \ol{\sigma}_{B^{n}_{1}E\ol{D}^{n}_{1}}$ be the optimizer of $H^{\uparrow}_{\delta}(A^{n}_{1}X^{n}_{1}|B^{n}_{1}E\ol{D}^{n}_{1})_{\ol{\rho}_{|\Omega}}$, we obtain
	\begin{align*}
		H^{\uparrow}_{\delta}(A^{n}_{1}X^{n}_{1}|B^{n}_{1}E\ol{D}^{n}_{1})_{\overline{\rho}_{|\Omega}} + H_{\delta}(D^{n}_{1}|A^{n}_{1}X^{n}_{1}B^{n}_{1}E\ol{D}^{n}_{1})_{\nu} & = -D_{\delta}(\overline{\rho}_{A^{n}_{1}X^{n}_{1}D^{n}_{1}B^{n}_{1}E\overline{D}^{n}_{1}|\Omega} || \mathbb{1}_{A^{n}_{1}X^{n}_{1}D^{n}_{1}} \otimes \ol{\sigma}_{B^{n}_{1}E\ol{D}^{n}_{1}}) \\
		& \leq H^{\uparrow}_{\delta}(A^{n}_{1}X^{n}_{1}D^{n}_{1}|B^{n}_{1}E\overline{D}^{n}_{1})_{\overline{\rho}_{|\Omega}} \ ,
	\end{align*}
	where the second inequality is by the definition of $H^{\uparrow}_{\delta}$. Now the goal is to simplify.
	\begin{enumerate}
		\item $H^{\uparrow}_{\delta}(A^{n}_{1}X^{n}_{1}|B^{n}_{1}E\ol{D}^{n}_{1})_{\overline{\rho}_{|\Omega}} = H^{\uparrow}_{\delta}(A^{n}_{1}|B^{n}_{1}E)_{\rho_{|\Omega}}$ as \cite[Lemma B.7]{Dupuis2016} allows us to remove the $X^{n}_{1}$ register and $\ol{D}^{n}_{1}$ is independent of everything as shown in \Cref{eq:marginal_indep}.
		\item $H_{\delta}(D^{n}_{1}|A^{n}_{1}X^{n}_{1}B^{n}_{1}E\overline{D}^{n}_{1})_{\nu} = H_{\delta}(D^{n}_{1}|\ol{D}^{n}_{1}X^{n}_{1})_{\nu}$ as \cite[Eq. (46)]{Dupuis2016} still holds because our $\nu$ is of the same form.
		\item $H^{\uparrow}_{\delta}(A^{n}_{1}X^{n}_{1}D^{n}_{1}|B^{n}_{1}E\ol{D}^{n}_{1})_{\overline{\rho}_{|\Omega}} = H^{\uparrow}_{\delta}(A^{n}_{1}D^{n}_{1}|B^{n}_{1}E\ol{D}^{n}_{1})_{\overline{\rho}_{|\Omega}}$ by \cite[Eq. (43)]{Dupuis2016}, as the proof for this equation still holds in our setting.
	\end{enumerate}
	Therefore we have simplified to
	\begin{align*}
		H^{\uparrow}_{\delta}(A^{n}_{1}|B^{n}_{1}E)_{\rho_{|\Omega}} + H_{\delta}(D^{n}_{1}|\ol{D}^{n}_{1}X^{n}_{1})_{\nu} \leq H^{\uparrow}_{\delta}(A^{n}_{1}D^{n}_{1}|B^{n}_{1}E\ol{D}^{n}_{1})_{\ol{\rho}_{|\Omega}} \ .
	\end{align*}
	Using the fact that $H_{\delta}(D^{n}_{1}|\ol{D}^{n}_{1}X^{n}_{1})_{\nu} \geq n\overline{g} - nh$ as before \cite[Below Eq. (52)]{Dupuis2016} and by pushing that term over to the other side,
	\begin{align*}
		H^{\uparrow}_{\delta}(A^{n}_{1}|B^{n}_{1}E)_{\rho_{|\Omega}}  \leq H^{\uparrow}_{\delta}(A^{n}_{1}D^{n}_{1}|B^{n}_{1}E\ol{D}^{n}_{1})_{\ol{\rho}_{|\Omega}} - n\ol{g} + nh \ .
	\end{align*}
\end{proof}

\begin{theorem}\label{thm:uparrowmaxtradeoff}
	Let $\delta \in (1/2,1)$, $\eta \in (1,1+2/V)$, $\zeta \in (1/2,1) \cup (1,\infty)$,  such that $\frac{\eta}{\eta-1} = \frac{\delta}{\delta-1} + \frac{\zeta}{\zeta-1}$ and $(\eta-1)(\delta-1)(\zeta-1) > 0$. Assuming the Markov chain conditions are satisfied [\Cref{eqn:Markov_cond_alt}], then
	\begin{align*}H^{\uparrow}_{\delta}(A^{n}_{1}|B^{n}_{1}E)_{\rho_{|\Omega}} \leq nh - \frac{\delta}{1-\delta}\log(\rho[\Omega]) + n\frac{\eta - 1}{4}V^{2} \ ,
	\end{align*}
	where $V := 2\lceil \|\nabla f\|_{\infty} \rceil + 2\log(1+2d_{A})$.
\end{theorem}
\begin{proof}
	Combining \Cref{lem:claim46} with \cite[Lemma B.5]{Dupuis2016} applied to $\ol{\rho}$, which we can do by requiring $\delta \in (1/2,1)$ and noting $\ol{\rho} = \rho[\Omega]\ol{\rho}_{|\Omega} + (\ol{\rho} - \rho[\Omega]\ol{\rho}_{|\Omega})$,
	one obtains
	$$ H^{\uparrow}_{\delta}(A^{n}_{1}|B^{n}_{1}E)_{\rho_{|\Omega}} \leq H^{\uparrow}_{\delta}(A^{n}_{1}D^{n}_{1}|B^{n}_{1}E\ol{D}^{n}_{1})_{\overline{\rho}} - \frac{\delta}{\delta-1}\log(\frac{1}{\rho[\Omega]}) + nh - n\ol{g} \ .$$
	Next we use \Cref{prop:chain_rule_special_case} Item 1 with the substitutions $A = A^{n}_{1}D^{n}_{1}$, $C = B^{n}_{1}E\overline{D}^{n}_{1}$, $\tau_{C} = \ol{\rho}_{B^{n}_{1}E\overline{D}^{n}_{1}}$, where we will require $\eta \in (1,1 + 1/\log(1+d_{A}d_{D}))$.
	\begin{align*}
		H^{\uparrow}_{\delta}(A^{n}_{1}D^{n}_{1}|B^{n}_{1}E\ol{D}^{n}_{1})_{\overline{\rho}} \leq  H_{\eta}(A^{n}_{1}D^{n}_{1}|B^{n}_{1}E\ol{D}^{n}_{1})_{\overline{\rho}} + D_{\gamma}(\ol{\rho}_{B^{n}_{1}E\ol{D}^{n}_{1}}||\ol{\rho}_{B^{n}_{1}E\ol{D}^{n}_{1}} ) 
		=  H_{\eta}(A^{n}_{1}D^{n}_{1}|B^{n}_{1}E\ol{D}^{n}_{1})_{\overline{\rho}} \ ,
	\end{align*}
	where the equality is by positive definiteness of R{\'e}nyi divergences.
	This gives us 
	$$H^{\uparrow}_{\delta}(A^{n}_{1}|B^{n}_{1}E)_{\rho_{|\Omega}} \leq  H_{\eta}(A^{n}_{1}D^{n}_{1}|B^{n}_{1}E\ol{D}^{n}_{1})_{\overline{\rho}} -  \frac{\delta}{\delta-1}\log(\frac{1}{\rho[\Omega]}) + nh - n\ol{g} \ .$$
	We just need to bound the $H_{\eta}$ term, which follows the same type of argument as leading to \cite[Eq. (40)]{Dupuis2016}, as noted at the end of the proof of \cite[Proposition 4.5]{Dupuis2016}. We have
	\begin{align*}
		H_{\eta}(A^{n}_{1}D^{n}_{1}|B^{n}_{1}E\ol{D}^{n}_{1})_{\overline{\rho}} & \leq \sum_{i} \underset{\omega_{R_{i-1}R}}{\sup} H_{\eta}(A_{i}D_{i}|B_{i}\overline{D}_{i}R)_{(\mathcal{D}_{i}\circ \mathcal{M}_{i})(\omega)} \\ 
		&  \leq \sum_{i} \underset{\omega_{R_{i-1}R}}{\sup} H(A_{i}D_{i}|B_{i}\overline{D}_{i}R)_{(\mathcal{D}_{i}\circ \mathcal{M}_{i})(\omega)} + n(\eta - 1)\log^{2}(1+2d_{A}d_{D}) \\
		& \leq \sum_{i} \underset{\omega_{R_{i-1}R}}{\sup} H(A_{i}D_{i}|B_{i}\overline{D}_{i}R)_{(\mathcal{D}_{i}\circ \mathcal{M}_{i})(\omega)} + n\frac{\eta-1}{4}V^{2} \ ,
	\end{align*}
	where the first line is because the Markov conditions are satisfied even including the $D_{i}\ol{D}_{i}$ registers due to our assumption that the Markov conditions in \Cref{eqn:Markov_cond_alt} hold\footnote{We omit the proof of this implication; see the proof of \cite[Proposition 4.5]{Dupuis2016} which contains proof of this.} so we can apply \cite[Corollary 3.5]{Dupuis2016}, the second is by applying \cite[Lemma B.9]{Dupuis2016} which we can do by our restriction on $\eta$, and the third is by definition of $V$. Lastly one can show $H(A_{i}D_{i}|B_{i}\ol{D}_{i}R)_{(\mathcal{D}_{i}\circ \mathcal{M}_{i})(\omega)} \leq \ol{g}$ \cite[End of proof of Proposition 4.5]{Dupuis2016}. Combining these points gets us our theorem.
\end{proof}

\subsection{QKD security proof with R\'{e}nyi entropies}
We now can present the proof of \Cref{thm:keylengthwithoutsmoothing}. In the main text we present physical and virtual QKD protocols (\Cref{prot:physicalDDQKD} and \Cref{prot:DDQKD} respectively). As explained in the main text, they are equivalent and so the following security proof holds for both. From a technical viewpoint, we use the EAT which holds for sequential processes, so the security is most directly proven for the virtual QKD protocol because it is sequential.

\begin{proof}[Proof of \Cref{thm:keylengthwithoutsmoothing}]
    By the equivalence of \Cref{prot:physicalDDQKD} and \Cref{prot:DDQKD}, which is by construction, proving the security of \Cref{prot:physicalDDQKD} is equivalent to proving the security of \Cref{prot:DDQKD}. We refer to \Cref{prot:DDQKD} in the proof, but this is just for the ease of the reader to link to the main text and \Cref{prot:DDQKD}'s sequential presentation. Let $\Omega'$ be the event that \Cref{prot:DDEA} does not abort (i.e. $\Omega$ is satisfied) and error correction verification accepts (i.e.  $\hash(\mathbf{Z}) = \hash(\mathbf{\widehat{Z}})$). Let $\Omega''$ be the event that \Cref{prot:DDQKD} does not abort and error correction verification passes. Let $\tilde{\rho}_{|\Omega''}$ be the output of the \Cref{prot:DDQKD} conditioned on not aborting and error correction verification passing and $\rho_{|\Omega'}$ be the output of \Cref{prot:DDEA} conditioned on not aborting \textit{and} error correction verification passing.

	Fundamentally, we are interested in the entropy of the raw key which excludes the registers discarded in the general sifting step or the registers used for parameter estimation. As defined in \Cref{prot:DDQKD}, $\mathbf{Z} \equiv \ol{A}^{n}_{1} \setminus (\ol{A}_{\mathcal{S}} \cup \ol{A}_{\mathcal{T}})$. However, we note that Alice announces $\ol{A}_{\mathcal{T}}$ in \Cref{prot:DDQKD}, making it part of what is conditioned on, and $\ol{A}_{\mathcal{S}}$ is determined entirely by the registers $\widetilde{A}^{n}_{1}$, $\widetilde{B}^{n}_{1}$, and $T^{n}_{1}$. In other words, it is clear by \Cref{prot:DDQKD} that there exists a deterministic function $f$ which can reconstruct $\ol{A}_{\cS},\ol{A}_{\cT}$ using a map of the form $\sum_{(s,p)} \dyad{s} \otimes \dyad{p} (\cdot) \dyad{s} \otimes \dyad{p} \otimes \dyad{f(s,p)}_{\ol{A}_{\cS}\ol{A}_{\cT}}$. Therefore, by \cite[Lemma B.7]{Dupuis2016},
	$$ H_{\alpha}^{\uparrow}(\mathbf{Z}|\widetilde{A}^{n}_{1}\widetilde{B}^{n}_{1} T^{n}_{1} \ol{A}_{\mathcal{T}} C E)_{\tilde{\rho}_{|\Omega''}} = H_{\alpha}^{\uparrow}(\ol{A}^{n}_{1}|\widetilde{A}^{n}_{1}\widetilde{B}^{n}_{1} T^{n}_{1} \ol{A}_{\mathcal{T}} C E)_{\tilde{\rho}_{|\Omega''}} \ . $$
		
	We can handle the classical communication cost of error correction by \cite[Eq. (5.96)]{Tomamichel2016}:
	\begin{align}
		H^{\uparrow}_{\alpha}(\ol{A}^{n}_{1}|\widetilde{A}^{n}_{1}\widetilde{B}^{n}_{1} T^{n}_{1} \ol{A}_{\mathcal{T}} C E)_{\tilde{\rho}_{|\Omega''}} &\geq H^{\uparrow}_{\alpha}(\ol{A}^{n}_{1}|\widetilde{A}^{n}_{1}\widetilde{B}^{n}_{1} T^{n}_{1} \ol{A}_{\mathcal{T}} E)_{\tilde{\rho}_{|\Omega''}} - {\leak}_{\vEC} \\
		& = H^{\uparrow}_{\alpha}(\ol{A}^{n}_{1}|\widetilde{A}^{n}\widetilde{B}^{n} T^{n}_{1} \ol{A}_{\mathcal{T}} E)_{\rho_{|\Omega'}} - {\leak}_{\vEC} \label{eqn:alt_thm_starting_point}
	\end{align}
	where $C$ is the register storing the classical communication due to error correction, and the equality comes from noting that $\rho_{\ol{A}^{n}_{1}\widetilde{A}^{n}_{1}\widetilde{B}^{n}_{1}T^{n}_{1}X^{n}_{1}E{|\Omega'}} = \tilde{\rho}_{\ol{A}^{n}_{1}\widetilde{A}^{n}_{1}\widetilde{B}^{n}_{1}T^{n}_{1}X^{n}_{1}E{|\Omega''}}$. We stress that Bob's private register, $\ol{B}^{n}_{i}$, for both the EAT and the QKD protocol, is excluded in the equality as they differ. We note that $\ol{A}_{\mathcal{T}}$ is conditioned upon as Alice publicly announced it, but $X^{n}_{1}$ is not conditioned upon as Bob computes it locally and does not need to announce it unless aborting.
		
With all of this addressed, the entropy term that we would like to calculate to perform privacy amplification is
	\begin{align}\label{eq:goal_without_smoothing}
		H^{\uparrow}_{\alpha}(\ol{A}^{n}_{1}|\widetilde{A}^{n}_{1}\widetilde{B}^{n}_{1} T^{n}_{1} \ol{A}_{\mathcal{T}} E)_{\rho_{|\Omega'}} \ .
	\end{align}This imposes one major difficulty: We need to find a way to relate this term to the term $H^{\uparrow}_{\beta}(\ol{A}^{n}_{1} \ol{B}^{n}_{1}|\widetilde{A}^{n}\widetilde{B}^{n} T^{n}_{1} E)_{\rho_{|\Omega'}}$ on the left-hand side of \Cref{eq:entropyRate}, where the Entropy Accumulation Theorem can be applied.
	
	Following derivations in the DI-QKD setting \cite{Arnon2018}, but also handling the $\ol{A}_{\mathcal{T}}$ register, we use the following series of claims:
	\begin{align}
		H^{\uparrow}_{\alpha}(\ol{A}^{n}_{1}|\widetilde{A}^{n}_{1} \widetilde{B}^{n}_{1} T^{n}_{1} \ol{A}_{\mathcal{T}} E)_{\rho_{|\Omega'}} 
		&\geq H^{\uparrow}_{\alpha}(\ol{A}^{n}_{1}|\widetilde{A}^{n}_{1} \widetilde{B}^{n}_{1} T^{n}_{1} \ol{A}_{\mathcal{T}} \ol{B}^{n}_{1} E)_{\rho_{|\Omega'}}  \notag \\
		&\geq H^{\uparrow}_{\beta}(\ol{A}^{n}_{1} \ol{B}^{n}_{1} \ol{A}_{\cT}|\widetilde{A}^{n}_{1}\widetilde{B}^{n}_{1} T^{n}_{1} E)_{\rho_{|\Omega'}} - H^{\uparrow}_{\delta}(\ol{B}^{n}_{1}\ol{A}_{\cT}|\widetilde{A}^{n}_{1}\widetilde{B}^{n}_{1} T^{n}_{1} E)_{\rho_{|\Omega'}} \notag \\
		& \geq H^{\uparrow}_{\beta}(\ol{A}^{n}_{1} \ol{B}^{n}_{1} |\widetilde{A}^{n}_{1}\widetilde{B}^{n}_{1} T^{n}_{1} E)_{\rho_{|\Omega'}} - H^{\uparrow}_{\delta}(\ol{B}^{n}_{1}\ol{A}_{\cT}| T^{n}_{1} E)_{\rho_{|\Omega'}}
		\label{eqn:mid_inequality} \ .
	\end{align}
	The first inequality is the data processing inequality for sandwiched R{\'e}nyi entropies \cite[Corollary 5.5]{Tomamichel2016} for the specific choice of tracing out the $\ol{B}^{n}_{1}$ register. The second inequality is Case 2 of \Cref{prop:chain_rule_special_case} with the substitutions $A = \ol{A}^{n}_{1}$, $B = \ol{A}_{\cT}\ol{B}^{n}_{1}$, $C = \wt{A}^{n}_{1}\wt{B}^{n}_{1}T^{n}_{1}E$, where we also demand $\alpha \in(1,2)$, $\delta \in (1/2,1)$. The last inequality is again using the data-processing inequality for the $H^{\uparrow}_\delta$ term and \cite[Lemma B.7]{Dupuis2016} for the $H^{\uparrow}_\beta$ term using that $\ol{A}_{\cT}$ may be generated using a deterministic mapping from $\ol{A}^{n}_{1}$ and $T^{n}_{1}$.
		
	As \Cref{eqn:mid_inequality} includes the term in \Cref{eq:goal_without_smoothing}, we need to bound the $H^{\uparrow}_{\delta}$ term in \Cref{eqn:mid_inequality}. We will do this using \Cref{thm:uparrowmaxtradeoff}. We use the replacements $A_{i} = \ol{A}_{\cT_{i}}\ol{B}_{i}$, $B_{i} = T_{i}$, and $E = E$, from which it is clear how to define the EAT maps for this process. We note it is okay to use $\ol{A}_{\mathcal{T}}$ as we may treat it as if it were actually $n$ registers where it is $\perp$ whenever $T_{i} \neq 0$. The Markov conditions trivially hold since $T_{i}$ uses seeded randomness. Then using the fact that $\ol{A}_{\cT_{i}}\ol{B}_{i} = \perp \times \perp$ except when $T_{i}=1$ which happens with probability $\gamma$, we construct the max-tradeoff function analytically by the following observation:
	\begin{align}\label{eq:maxtradeoffconstruction}
		H(\ol{A}_{\cT_{i}} \ol{B}_{i}|T_{i}R')_{\cM_{i}(\omega)} & \leq H(\ol{A}_{\cT_{i}}\ol{B}_{i}|T_{i})_{\cM_{i}(\omega)}\notag\\
		& = (1-\gamma)H(\ol{A}_{\cT_{i}}\ol{B}_{i}|T_{i}=0)_{\cM_{i}(\omega)} + \gamma H(\ol{A}_{\cT_{i}}\ol{B}_{i}|T_{i}=1)_{\cM_{i}(\omega)} \notag \\
		& \leq \gamma \log|\mathcal{A} \times \mathcal{B}|
	\end{align}
	where $\mathcal{A}$ and $\mathcal{B}$ are the alphabets of private outcomes Bob utilizes during parameter estimation. This tells us we can let the max-tradeoff function  $f_{\max}(\vect{q}) = \gamma \log |\mathcal{A} \times \mathcal{B}|$ for all $\vect{q} \in \mathbb{P}(\cX)$. This also implies $\|\nabla f_{\max}\|_{\infty} = 0$. Note that for \Cref{thm:uparrowmaxtradeoff} to hold, we require $\delta \in (1/2,1)$, which is why we demanded this in the first chain rule. Therefore, using \Cref{thm:uparrowmaxtradeoff}
	we get 
	\begin{align}\label{eq:deltamaxtradeoff}
		& H^{\uparrow}_{\delta}(\ol{A}_{\cT} \ol{B}^{n}_{1}|T^{n}_{1} E)_{\rho_{|\Omega'}} \leq n\gamma\log|\mathcal{A} \times \mathcal{B}| - \frac{\delta}{1-\delta}\log(\vAcc) + n(\eta-1)\log(1+2d_{S})
	\end{align}
	where we have replaced $\rho[\Omega]$ with $\rho[\Omega']$ and then assumed $\vAcc \leq \rho[\Omega']$.\footnote{Note in principle this is a free parameter, but it does imply a lower bound on the probability that the input will pass parameter estimation and error correction verification, which is why we label it $\vAcc$ for accept.} We have also defined $d_{S} = (|\cA|+1)(|\cB|+1)$.
		
	Combining \Cref{eqn:mid_inequality} and \Cref{eq:deltamaxtradeoff},
	\begin{align*}
		H^{\uparrow}_{\alpha}(\ol{A}^{n}_{1}| \widetilde{A}^{n}_{1} \widetilde{B}^{n}_{1} T^{n}_{1} \ol{A}_{\mathcal{T}} E)_{\rho_{|\Omega'}} 
		&\geq H^{\uparrow}_{\beta}(\ol{A}^{n}_{1} \ol{B}^{n}_{1} |\widetilde{A}^{n}_{1}\widetilde{B}^{n}_{1} T^{n}_{1} E)_{\rho_{|\Omega'}} - n\gamma\log|\mathcal{A} \times \mathcal{B}| \notag\\
		& \hspace{0.75cm}
		+\frac{\delta}{1-\delta}\log(\vAcc) - n(\eta-1)\log(1+2d_{S}) \ .
	\end{align*}
	We now apply the equivalent of \Cref{thm:entRate} to the $H_{\beta}^{\uparrow}$ term. We can do this by our restriction on $\beta$ earlier along with noting on the relevant marginal that any frequency distribution $\q$ accepted in the event $\Omega'$ is by definition also accepted in the event $\Omega$ so the value of $h$ will hold. Therefore, the only formal change in \Cref{thm:entRate} is using the replacement $\vEA \to \vAcc$. Therefore, we have that either the input state aborts with probability greater than $1-\vAcc$ or 
	\begin{equation}
	\begin{aligned}
		H^{\uparrow}_{\alpha}(\ol{A}^{n}_{1}| \widetilde{A}^{n}_{1} \widetilde{B}^{n}_{1} T^{n}_{1} \ol{A}_{\mathcal{T}} E)_{\rho_{|\Omega'}} 
		&>  n h - n \frac{(\beta-1)\ln 2}{2}V^{2} - \frac{\beta}{\beta-1} \log \frac{1}{\vAcc} - n(\beta-1)^{2} K_{\beta} \notag \\
	    & \hspace{0.75cm} - n\gamma\log|\mathcal{A} \times \mathcal{B}|
		+ \frac{\delta}{1-\delta}\log(\vAcc) \notag \\
		& \hspace{1.25cm} - n(\eta-1)\log(1+2d_{S}) \ .
	\end{aligned}
	\end{equation}
	Grouping terms, the $\vAcc$ correction term is of the form $\frac{\beta-\delta}{(\beta+1)(1-\delta)}\log(\vAcc)$. Then, plugging this back into \Cref{eqn:alt_thm_starting_point} and applying the leftover hash lemma for non-smooth entropies \cite{Dupuis2021}, \Cref{corr:LHL_without_smoothing}, resolves everything except that one has many conditions on $\alpha,\beta,\delta,\gamma,\eta,\zeta$.
		
	To get the simplified conditions on $\alpha,\beta,\delta,\gamma,\eta,\zeta$ in the theorem, note the conditions on $\alpha,\beta,\delta$ are satisfied for $\beta \in (1,2),\delta \in (1/2,1), \alpha = \frac{-\beta +\delta}{-1+2\delta-\beta\delta}$. The conditions on $\delta \in (1/2,1),\eta \in (1,1+\log^{-1}(1+d_{S})),\zeta \in (1/2,1) \cup (1,\infty)$ are satisfied for the whole relevant range so long as $\zeta = \frac{-\delta + \eta}{-1+2\eta-\delta\eta}$, which does not enter into our key length. Finally, we note that $\eta > 1$, $\log(1+2d_{S}) > 1$, so $-n(\eta-1)\log(1+2d_{S}) < 0$, so this term is always a penalty on the key rate. In principle we would need to minimize $\eta$, but noting that $-n(\eta-1)\log(1+2d_{S})$ goes to $0$ as $\eta$ goes to $1$, and $\eta$ is free to tend towards $1$, we choose $\eta$ such that $-n(\eta-1)\log(1+2d_{S}) = -1$.
	
	Lastly, we need to explain the security claim in the statement of the theorem. There are two cases. Either $\rho[\Omega] < \vAcc$ or $\rho[\Omega] \geq \vAcc$. If $\rho[\Omega] < \vAcc$, then for that input, the security is trivially bounded by $\vAcc$ by using \Cref{eq:security-scaled-by-accept} with the replacement $\Pr(\text{accept}) = \rho[\Omega] < \vAcc$. Therefore the focus is on the case $\rho[\Omega] \geq \vAcc$. First we consider the secrecy. Starting from the subnormalized state\footnote{Subnormalization is due to aborting the protocol on the input.} $\rho_{S_{A}EF}$ where $S_{A}:= h(\mathbf{Z})$ is the hashed key, $h$ is the hash function, and $F$ is the announcement of the hash function,
	\begin{align*}
	     \frac{1}{2} \| \rho_{S_{A}EF} - |S_{A}|^{-1}\idm_{S_{A}} \otimes \rho_{EF} \|_{1}
	    = \frac{1}{2} \rho[\Omega'] \| \rho_{S_{A}F_{|\Omega'}} - |S_{A}|^{-1}\idm_{S_{A}} \otimes \rho_{EF_{|\Omega'}} \|_{1} 
	    \leq \rho[\Omega'] \vSec \leq \vSec \ ,
	\end{align*}
	where we started from the subnormalized state, normalized it by the condition of passing, and in the last line used that \Cref{corr:LHL_without_smoothing} tells us the trace norm divided by two is at most $\vSec$. Therefore any input such that $\rho[\Omega'] \geq \vAcc$, the key is $\vSec$-secret.  Finally, we need to address correctness which is an upper bound on $\Pr[S_{A} \neq S_{B} \, \& \Omega']$. By the analysis of error correction verification using 2-universal hash functions \cite{Renner2005}, $\vEC$ is exactly what bounds this term, so we have the key is $\vEC$-correct. Therefore, when $\rho[\Omega'] \geq \vAcc$, the key is $\vSec + \vEC$-secure. Combining these points, we have the protocol is $\max\{\vAcc,\vSec+\vEC\}$-secure. Finally, noting that if $\vAcc < \vEC + \vSec$ you could increase to $\vAcc = \vEC + \vSec$ to increase the key length without changing the security claim, so without loss of generality $\vAcc \geq \vEC + \vSec$. As such, we can just say the protocol is $\vAcc$-secure.
	\end{proof}
	
We remark on an alternative key length that may be obtained largely following the same proof that allows us to express things in terms of $\vEA$ directly. This does not seem to provide a clear advantage over the previous proof. As such, we note that the use of \cite[Lemma B.5]{Dupuis2016} could presumably be replaced by \cite[Lemma 10]{Tomamichel2017} for the smoothed version, but we omit such a proof in this work.
\begin{prop}
    Consider any QKD protocol which follows \Cref{prot:physicalDDQKD} and satisfies \Cref{assumption:block_diagonal}. $\vEA,\vSec,\vEC \in (0,1]$, $\vECRob\in[0,1]$, such that $\vEA \geq \vSec + \vEC$. Let $h$ such that $f(\vect{q}) \geq h$ for all $\vect{q} \in \mathcal{Q}$ where $f$ is the min-tradeoff function generated by \Cref{alg:Algorithm1_min_tradeoff_function} or \Cref{alg:Algorithm2_min_tradeoff_function}. \\
			
			\noindent Let $\beta \in (1,2)$, $\delta \in (1/2,1)$ and $\alpha = \frac{-\beta +\delta}{-1+2\delta-\beta\delta}$.
			The QKD protocol is $\vEA$-secure for key length
			\begin{equation}\label{eq:KeyLength_without_smoothing_alt} 
				\begin{aligned}
					\ell \leq & n h - \leak_{\vEC} - n \frac{(\beta-1)\ln 2}{2}V^{2} - n(\beta-1)^{2} K_{\beta} - n\gamma\log|\mathcal{A} \times \mathcal{B}| \\
					& \hspace{3mm} 
					+\frac{\beta-\delta}{(\beta-1)(1-\delta)}\log(\vEA) 
					+ \frac{\alpha}{\alpha - 1} \log(\ve_{\mathrm{sec}}\cdot (1-\vECRob)) + 1
				\end{aligned}
			\end{equation}
			where 
			\begin{align*}
				V &= \sqrt{\Var{f}+2} + \log(2 d_{S}^{2}+1) \\
				K_\beta &= \frac{1}{6(2-\beta)^{3} \ln 2} 2^{(\beta-1)(\log d_{S} + \Max{f} - \MinSigma{{f})}} \ln^{3} \left( 2^{\log d_{S} + \Max{f} - \MinSigma{f}} + e^{2} \right) \ ,
			\end{align*}
			where
			$\mathcal{A}$ and $\mathcal{B}$ are the alphabets of private outcomes for Alice's and Bob's announcements excluding the symbol $\perp$, respectively, $d_{S} = (|\cA|+1)(|\cB|+1)$, $\leak_{\vEC}$ is the amount of information leakage during the error correction step.
\end{prop}
\begin{proof}
    We largely follow the previous proof with some modifications which we explain. We start at \Cref{eqn:alt_thm_starting_point}. First we note that $\rho_{|\Omega} =\rho[\Omega'|\Omega]\rho_{|\Omega'} + (\rho_{|\Omega}-\rho[\Omega'|\Omega]\rho_{|\Omega'})$, where $[\Omega'|\Omega]$ denotes the event $\Omega'$ occurring conditioned on $\Omega$ occurring, i.e.\ the probability that error correction succeeds given we pass parameter estimation.\footnote{As the notation is a slight abuse of notation from conditional probability which would have $\rho[\text{EC accepts}|\Omega]$, we clarify. Recall one starts with the output after error correction $\rho$. We can partition the total state according to the events error correction and parameter estimation pass, $\Omega'$, just parameter estimation passes, $\Omega \setminus \Omega'$, and everything else, so that $\rho = \rho[\Omega']\rho_{|\Omega'} + \rho[\Omega\setminus \Omega']\rho_{\Omega \setminus \Omega'} + (\rho - \rho[\Omega']\rho_{|\Omega'} + \rho[\Omega\setminus \Omega']\rho_{\Omega \setminus \Omega'})$. Then by definition, $\rho_{|\Omega} = \rho[\Omega]^{-1}(\rho[\Omega']\rho_{|\Omega'} + \rho[\Omega' \setminus \Omega] \rho_{|\Omega' \setminus \Omega})$. Recalling $\Omega'$ is $\Omega$ and error correction passes, we can think of $\rho[\Omega'|\Omega]:=\rho[\Omega']/\rho[\Omega]$ as the probability of error correction passing given parameter estimation.} Moreover, since ultimately $\alpha > 1$ will always hold on the conditions that complete the proof, we can apply \cite[Lemma B.5, Eq. (80)]{Dupuis2016} to determine
    \begin{align*}
        H^{\uparrow}_{\alpha}(\ol{A}^{n}_{1}|\widetilde{A}^{n}\widetilde{B}^{n} T^{n}_{1} \ol{A}_{\mathcal{T}} E)_{\rho_{|\Omega'}} & \geq H^{\uparrow}_{\alpha}(\ol{A}^{n}_{1}|\widetilde{A}^{n}\widetilde{B}^{n} T^{n}_{1} \ol{A}_{\mathcal{T}} E)_{\rho_{|\Omega}} - \frac{\alpha}{\alpha-1}\log(\frac{1}{\rho[\Omega'|\Omega]}) \\
        & \geq  H^{\uparrow}_{\alpha}(\ol{A}^{n}_{1}|\widetilde{A}^{n}\widetilde{B}^{n} T^{n}_{1} \ol{A}_{\mathcal{T}} E)_{\rho_{|\Omega}} + \frac{\alpha}{\alpha-1}\log(1-\vECRob)
        ,
    \end{align*}
    where the second inequality follows from the assumption $\rho[\Omega'|\Omega] \geq 1-\vECRob$.\footnote{The choice of notation is as follows. The $\ve$-robustness of a protocol on an input $\rho$ is the probability of it aborting on that input \cite{Renner2005}. Ideally, we want our error correction scheme to be $0$-robust whenever parameter estimation has passed, i.e. $\rho[\Omega'|\Omega] = 1$ is the ideal case. As such, the assumption we have made is roughly that if parameter estimation passes, then the error correction will be $\vECRob$-robust.} Everything through \Cref{eq:deltamaxtradeoff} is the same argument but with $\rho_{|\Omega}$. The only change this results in is that any term with $\vAcc \to \vEA$. Since we consider $\rho_{|\Omega}$, we can apply \Cref{thm:entRate} directly, so we replace $\vAcc$ in previous proof with $\vEA$ again. Combining things and simplifying the greek alphabet parameters is the largely the same as before. This completes the proof.
    
    As for the security claim. We split into two cases in the same fashion. In the case $\rho[\Omega] < \vEA$, the state is trivially $\vEA$ secure. In the case $\rho[\Omega] \geq \vEA$, we have identical analysis to the previous case. The correctness argument is the same as before, so the security claim is $\ve = \max\{\vEA,\vSec+\vEC\}$.
\end{proof}

\section{Relation between Fixed-Length and Probabilistic Round-by-Round Parameter Estimation Protocols} \label{app:EATtoDDQKDCorrespondence}

In this appendix we provide the proof of \Cref{thm:EATOutvsTrueOut}. Effectively, we need to describe the QKD protocol in enough detail that we can view the output as a (somewhat complicated) CQ state whose probability of certain sequences can be determined using the structure of the protocol. To do this, we first describe the general description of a device-dependent QKD protocol through the parameter estimation step in the most explicit formal fashion, i.e. in terms of conditional states. We then show how this relates to the output of the EAT subprotocol, and finally we relate conditional states of these two cases. This allows us to determine exactly how the security parameter relates in the two cases. We stress that the notation in this appendix slightly differs from that of the rest of the paper so as to keep things concise. In this appendix, we use registers $Q_AQ_B$ instead of register $Q$ to separate Alice's and Bob's registers. We try to keep the notation as close to the main text for consistency. 

As explained in \Cref{subsec:ApplyingEATtoDDQKD}, since we are considering device-dependent QKD protocols, we may make the assumption that the measurements by Alice and Bob are performed in an $n$-fold manner, i.e. $\mathcal{E}^{M} = \Phi_{\mathcal{P}}^{\otimes n}$ where $\Phi_{\mathcal{P}} : Q_AQ_B \to \mathbb{C}^{\abs{\cA \times \cB}}$ where, following the notation of \Cref{thm:keylengthwithoutsmoothing}, $\cA \times \cB$ is a joint finite alphabet of Alice and Bob's outcomes for parameter estimation, and we dropped the subscript $i$ from the Hilbert spaces as traditionally device-dependent QKD assumes the dimension to be fixed in each round. Given this assumption, we may write the output state of the measurement as
\begin{align*}
	(\mathcal{E}^{M} \otimes \id_{E})(\rho_{Q^{n}_{A,1} Q^{n}_{B,1}E}) = \sum_{(a^{n}_{1},b^{n}_{1},p^{n}_{1})} \dyad{a^{n}_{1}} \otimes \dyad{b^{n}_{1}} \otimes \dyad{p^{n}_{1}} \otimes \rho_{E}^{(a^{n}_{1},b^{n}_{1},p^{n}_{1})} \ ,
\end{align*}
where $a_{i}$, $b_{i}$ are outcomes of Alice's and Bob's private registers and outcomes of all public registers of the $i^{\text{th}}$ round are grouped as $p_{i}$, and $\rho_{E}^{(a^{n}_{1},b^{n}_{1},p^{n}_{1})}$ is the (subnormalized) conditional state of Eve.

Next Alice and Bob perform parameter estimation.\footnote{In \Cref{prot:DDQKD}, sifting and the key map happen each round along with parameter estimation. In Renner's original framework \cite{Renner2005}, it is not the case as the sifting and key map happen afterwards in the blockwise processing, which is the reason for this discrepancy from the main text.} Canonically the assumption is they are going to pick a random subset of the $n$ signals of size $m$ and announce their respective outcomes on those $m$ rounds.\footnote{This is equivalent to applying a random permutation and then performing parameter estimation on the first $m$ indices, but to match the EAT, we do not reformat in this manner. We note an alternative model one might propose that seems to be intuitive is to perform probabilistic round-by-round testing, count the total number of tests at the end, and then a posteriori treat the protocol as having picked that many tests. This does not fit into a pre-existing universally composable security proof for fixed-length testing. In particular, it does not fit into the security proof method of \cite{Renner2005}.} Finally they decide if they continue the protocol or not based on those outcomes. That is, they have some set of observed sequences $\mathcal{Q} \subset {\cX'}^{\times m}$ which they accept on, where $\cX'$ is as defined in \Cref{sec:crossover-min-tradeoff}. It follows that one could split $\mathcal{E}^{PE}$ as a channel which generates which rounds to test, a map where they make their announcements, and a map where they abort if the observations are not in a set of pre-agreed upon acceptable observations $\mathcal{Q}$:

$$ \mathcal{E}^{PE} = \mathcal{E}^{\mathcal{Q}} \circ \mathcal{E}^{PE,ann} \circ \mathcal{E}^{PE, test gen} \ . $$

Note that while fundamentally $\mathcal{Q}$ is a set of sequences of length $m$, $x^{m}_{1} \in {\cX'}^{\times m}$, it can trivially be extended to a set of sequences of length $n$, $x_{1}^{n} \in (\cX' \cup \{\perp\})^{\times n}$, which better matches the EAT and is denoted $\Omega$. In the case that it is a set of sequences of length $n$, it only makes sense for all sequences in $\Omega$ to have $n-m$ of the entries be the $\perp$ symbol. To match with EAT we mean this event $\Omega$ than aligns properly with extending $\cQ$.

Having completed the explanation of the protocol, we can see the natural definitions for these maps are:
$$ \mathcal{E}^{PE, test gen}(\alpha) = \frac{1}{|T_m|} \sum_{t^{n}_{1} \in T_{m}} \dyad{t^{n}_{1}} $$
where $T_{m} = \{x \in \{0,1\}^{n} : \text{Hamming Weight}(x) = m \}$.
\begin{align*}
	\mathcal{E}^{PE,ann}(\rho_{A^{n}_{1}B^{n}_{1}T^{n}_{1}}) =& \sum_{(a^{n}_{1},b^{n}_{1},t^{n}_{1})} \dyad{\mathrm{x}^{\times n}(a^{n}_{1},b^{n}_{1},t^{n}_{1},p^{n}_{1})}_{X^{n}_{1}} \otimes \dyad{a^{n}_{1}} \\ 
	& \hspace{6cm} \otimes \dyad{b^{n}_{1}} \otimes \dyad{p^{n}_{1}} \otimes \rho_{E}^{(a^{n}_{1},b^{n}_{1},p^{n}_{1})}
\end{align*}
where $\mathrm{x}(a_{i},b_{i},t_{i},p_{i}) = (a_{i} \times b_{i}) \in \cA \times \cB$ when $t_{i} = 1$ and otherwise equals $\perp$, i.e.\ $\mathrm{x}$ is the \textit{deterministic} announcement function due to testing that constructs $x_{i} \in \cX$.

$$ \mathcal{E}^{\mathcal{Q}}(\rho_{X^{n}_{1}}) = \sum_{x^{n}_{1} \in \Omega} \dyad{x^{n}_{1}}\rho_{X^{n}_{1}}\dyad{x^{n}_{1}} \ ,$$
where we the event $\Omega$ as the set of sequences of length $n$ that satisfy $\mathcal{Q}$ as stated earlier. 

Since $p_{i}$ is a function of $(a_{i},b_{i},t_{i})$ by construction of \Cref{prot:physicalDDQKD} and $x_{i}$ is a function of $(a_{i},b_{i},t_{i},p_{i})$, we are interested in the set $\{(a^{n}_{1},b^{n}_{1},t^{n}_{1}) \text{ such that } x^{n}_{1} \in \Omega \text{ with non-zero probability}\}$. Throughout, when we write a summation over something of the form  ``$(a^{n}_{1},b^{n}_{1},t^{n}_{1}) : x^{n}_{1} \in \Omega$," we are indicating the set of sequences that result in $x^{n}_{1} \in \Omega$. Therefore, combining all of this, we find the output state after parameter estimation is 
\begin{align*}
	\widetilde{\rho}^{out}
	=& \frac{1}{|T_{m}|} \sum_{\substack{(a^{n}_{1},b^{n}_{1},t^{n}_{1}) \, : x^{n}_{1} \in \Omega}} \dyad{\mathrm{x}^{\times n}(a^{n}_{1},b^{n}_{1},t^{n}_{1},p^{n}_{1})}_{X^{n}_{1}} \otimes \dyad{a^{n}_{1}} \\
	& \hspace{5cm} \otimes \dyad{b^{n}_{1}} \otimes \dyad{p^{n}_{1}} \otimes \dyad{t^{n}_{1}} \otimes  \rho_{E}^{(a^{n}_{1},b^{n}_{1},p^{n}_{1})} \ ,
\end{align*}
where $ \widetilde{\rho}^{out} = (\mathcal{E}^{PE,ann} \circ \mathcal{E}^{PE,test gen} \circ \mathcal{E}^{M} \otimes \id_{E})(\rho_{Q^{A,n}_{1}{Q}^{B,n}_{1}E})$. We stress by construction of $\Omega$, that any element of $(a^{n}_{1},b^{n}_{1},t^{n}_{1}) : x^{n}_{1} \in \Omega$ has $t^{n}_{1} \in T_{m}$, which is why we don't restrict the space $T^{n}$. We will define the final output after sifting and the key map (but prior to error correction) as
\begin{align*}
	\rho^{out}
	=& \frac{1}{|T_{m}|} \sum_{(\hat{a}^{n}_{1},\hat{b}^{n}_{1},t^{n}_{1}) : x^{n}_{1} \in \Omega} \dyad{\mathrm{x}^{\times n}(\hat{a}^{n}_{1}, \hat{b}^{n}_{1},t^{n}_{1},p^{n}_{1})}_{X^{n}_{1}} \otimes \dyad{\hat{a}^{n}_{1}} \\
	& \hspace{5cm}\otimes \dyad{\hat{b}^{n}_{1}} \otimes \dyad{p^{n}_{1}} \otimes \dyad{t^{n}_{1}} \otimes  \rho_{E}^{(\hat{a}^{n}_{1},\hat{b}^{n}_{1},p^{n}_{1})} \ ,
\end{align*}
where hats denoted updated registers, and the reason that we can define $\mathrm{x}$ on the updated sequence is that by assumption testing commutes with sifting and the key map, given the registers $T^{n}_{1}$.

This completes our account of the output state of parameter estimation in the fixed-length device-dependent QKD protocol. However, we will need these extra definitions when comparing to the output of the EAT protocol:
\begin{equation}
    \begin{aligned}
	\rho^{out}_{|\Omega_{sub}} = (\mathcal{E}^{\mathcal{Q}} \otimes \id_{E})(\rho^{out}) \qquad
	\rho^{out}_{|\Omega} = \frac{1}{\Tr(\rho^{out}_{\Omega_{sub}})} \rho^{out}_{\Omega_{sub}} \ .
	\end{aligned}
\end{equation}

Using the same notation as we did for the fixed-length device-dependent QKD protocol, it is straightforward from the effective tensor product structure of device-dependent QKD (\Cref{subsec:ApplyingEATtoDDQKD}) that the output of the EAT protocol can be written as:
\begin{align*}
	\rho_{EAT}^{out} =& \sum_{(\hat{a}^{n}_{1},\hat{b}^{n}_{1},t^{n}_{1}): x^{n}_{1} \in \Omega} \dyad{\mathrm{x}^{\times n}(\hat{a}^{n}_{1}, \hat{b}^{n}_{1},t^{n}_{1},p^{n}_{1})}_{X^{n}_{1}} \otimes \dyad{\hat{a}^{n}_{1}} \\
	& \hspace{5cm} \otimes \dyad{\hat{b}^{n}_{1}} \otimes \dyad{p^{n}_{1}} \otimes \dyad{t^{n}_{1}} \otimes  \rho_{E}^{(\hat{a}^{n}_{1},\hat{b}^{n}_{1},p^{n}_{1})} \ . 
\end{align*}
Our interest is then proving the relation of $\rho^{out}_{EAT}$ to $\rho^{out}$ to prove \Cref{thm:EATOutvsTrueOut}.

\begin{proof}[Proof of \Cref{thm:EATOutvsTrueOut}]
	Ultimately we need to argue about the probabilities of given sequences conditioned on passing, so we start by simplifying the joint probability of the sequences. For notational simplicity in this proof,  we let 
	$$\hat{r}_{i} := (\hat{a}_{i},\hat{b}_{i},p_{i}) \quad i \in [n] \hspace{5mm}\text{and}\hspace{5mm} \widehat{\cS}_{|\Omega} := \{(x^{n}_{1},r^{n}_{1},t^{n}_{1}) : x^{n}_{1} \in \Omega \} \ , $$
	and $r_{i}$, $\cS_{|\Omega}$ are defined the same way but with the un-updated $a_{i},b_{i}$.
	First consider $\rho^{in} \equiv \rho_{Q^{n}_{A,1} Q^{n}_{B,1}E}$ to the fixed-length protocol that satisfies data partitioning as in \Cref{prot:physicalDDQKD} (so that we consider the same announcement structure between the two settings for fair comparison). Then, for any \textit{non-zero} probability sequence, we have:
	\begin{align*}
		\Pr(x^{n}_{1},r^{n}_{1},t^{n}_{1}) 
		=  \Pr(r^{n}_{1},t^{n}_{1}) 
		=  \Pr(t^{n}_{1})\Pr(r^{n}_{1}) 
		=  \frac{1}{|T_{m}|}\Pr(r^{n}_{1}) \ ,
	\end{align*}
	where the first equality is because $x^{n}_{1}$ is produced by a deterministic function of the other sequences, the second equality is because $t^{n}_{1}$ is determined independent of the other sequences, and the third equality is because $t^{n}_{1}$ was uniformly chosen from the set $T_{m}$. Note that sifting and the key map, while deterministic, are not bijective in general. However, since accepting or not accepting a sequence is a deterministic function given $x^{n}_{1}$ which is a deterministic function of $(r^{n}_{1},t^{n}_{1})$ and doesn't change under updating, we have
	\begin{align*}
	\sum_{ (x^{n}_{1},\hat{r}^{n}_{1},t^{n}_{1}) \in \cS_{|\Omega}} \Pr(x^{n}_{1},\hat{r}^{n}_{1},t^{n}_{1})
	=\sum_{ (x^{n}_{1},r^{n}_{1},t^{n}_{1}) \in \cS_{|\Omega}} \Pr(x^{n}_{1},r^{n}_{1},t^{n}_{1})
	= \frac{1}{|T_{m}|} \sum_{(x^{n}_{1},r^{n}_{1},t^{n}_{1}) \in \cS_{|\Omega}} \Pr(x^{n}_{1}|r^{n}_{1},t^{n}_{1})\Pr(r^{n}_{1}) \ .
	\end{align*}
	
	Now we consider the EAT protocol. Note that the sequence $t^{n}_{1}$ is an i.i.d. sequence generated from sampling the distribution $Q =(1-\gamma , \gamma)^{T},$ whose probability is determined by the method of types. Following the notation used in \cite{Cover2005}, we have $\Pr(t^{n}) = Q^{n}(t^{n}_{1})$ for any binary sequence $t^{n}_{1}$ sampled from $Q$ $n$ times. It follows,
	\begin{align*}
		\Pr(r^{n}_{1},t^{n}_{1}) 
		=  \Pr(t^{n}_{1})\Pr(r^{n}_{1}) 
		= Q^{n}(t^{n}_{1})\Pr(r^{n}_{1})  \ ,
	\end{align*}
	which, again using that acceptance is a deterministic function of  $(r^{n}_{1},t^{n}_{1})$,
	\begin{align*}
		\sum_{(x^{n}_{1},\hat{r}^{n}_{1},t^{n}_{1}) : x^{n}_{1} \in \Omega} \Pr(x^{n}_{1},\hat{r}^{n}_{1},t^{n}_{1}) 
		=& \sum_{(x^{n}_{1},\hat{r}^{n}_{1},t^{n}_{1}) : x^{n}_{1} \in \Omega} \Pr(x^{n}_{1}|r^{n}_{1},t^{n}_{1})\Pr(r^{n}_{1},t^{n}_{1}) \\
		=& \sum_{(x^{n}_{1},r^{n}_{1},t^{n}_{1}): x^{n}_{1} \in \Omega} Q^{n}(t^{n}_{1})\Pr(x^{n}_{1}|r^{n}_{1},t^{n}_{1})\Pr(r^{n}_{1}) \ ,
	\end{align*}
	where in the first line we stress we have kept the sum indexing over the non-updated registers, so the sum is independent of how they get updated conditioned on $x^{n}_{1} \in \Omega$. By \cite[Theorem 11.1.2]{Cover2005}, we know $Q^{n}(t^{n}_{1}) = 2^{-n\left(H(f_{t}) - D(f_{t}||Q)\right)}$ where $f_{t}$ is the frequency distribution over $T$ induced by $t^{n}_{1}$ (i.e. the type). Therefore, if ${t}^{n}_{1}$ has Hamming weight $m$, since $m$ is defined as $m = \gamma n$, we have $Q^{n}({t}^{n}_{1}) = 2^{-nh(\gamma)}$ as the relative entropy term is zero, where $h(\cdot)$ is the binary entropy function. As a sequence will only be accepted if $t^{n}_{1}$ has Hamming weight $m$ and $x^{n}_{1}$ is acceptable, we have that
	\begingroup
	\allowdisplaybreaks
	\begin{align*}
		\rho^{out}_{EAT}[\Omega] & = \Tr(\sum_{(x^{n}_{1},\hat{r}^{n}_{1},t^{n}_{1}) \in \widehat{\cS}_{|\Omega}} \Pr(x^{n}_{1},\hat{r}^{n}_{1},t^{n}_{1}) \dyad{x^{n}_{1},\hat{r}^{n}_{1},t^{n}_{1}} \otimes \hat{\rho}_{E}^{(\hat{r}^{n}_{1})}) \\
		& = \sum_{(x^{n}_{1},\hat{r}^{n}_{1},t^{n}_{1}) \in \widehat{\cS}_{|\Omega}} \Pr(x^{n}_{1},\hat{r}^{n}_{1},t^{n}_{1}) \\
		& = 2^{-nh(\gamma)} \sum_{(x^{n}_{1},\hat{r}^{n}_{1},t^{n}_{1}) \in \cS_{|\Omega}} \Pr(x^{n}_{1}|r^{n}_{1},t^{n}_{1})\Pr(r^{n}_{1}) \\
		& = 2^{-nh(\gamma)} |T_{m}| \sum_{(x^{n}_{1},\hat{r}^{n}_{1},t^{n}_{1}) \in \cS_{|\Omega}} \frac{1}{|T_{m}|} \Pr(x^{n}_{1}|r^{n}_{1},t^{n}_{1})\Pr(r^{n}_{1}) \\
		& = 2^{-nh(\gamma)} |T_{m}| \Tr(\sum_{(x^{n}_{1},\hat{r}^{n}_{1},t^{n}_{1}) \in \widehat{\cS}_{|\Omega}} \Pr(x^{n}_{1},\hat{r}^{n}_{1},t^{n}_{1}) \dyad{x^{n}_{1}\hat{r}^{n}_{1},t^{n}_{1}} \otimes \hat{\rho}_{E}^{(\hat{r}^{n}_{1})}) \\
		&= 2^{-nh(\gamma)} \binom{n}{m} \rho^{out}[\Omega] \ ,
	\end{align*}
	\endgroup
	where the last line is because there are $\binom{n}{m}$ $n$-length bit strings with Hamming weight $m$.\\
	
	To see that $\rho^{out}_{EAT|\Omega} = \rho^{out}_{|\Omega}$, we make the following observations:
	\begin{enumerate}
		\item The set of sequences $(a^{n}_{1},b^{n}_{1},p^{n}_{1},t^{n}_{1})$ allowed by $\Omega$ is exactly the same in both cases.
		\item For any tuple of sequences $(a^{n}_{1},b^{n}_{1},p^{n}_{1})$, $p(a^{n}_{1},b^{n}_{1},p^{n}_{1})$ is the same for both protocols given that $\rho^{in}$ is the same and the measurement and announcement procedures are the same--- it is just the distribution over $\Pr(t^{n}_{1})$ that differs.
		\item For every \textit{allowed} tuple of sequences, the probability $\Pr(t^{n}_{1})$ is a fixed (uniform) value. Specifically, for allowed tuples of sequences, $\Pr(t^{n}_{1}) = 2^{-nh(\gamma)}$ for EAT, but $\Pr(t^{n}_{1}) =\frac{1}{|T_{m}|}$ for fixed test length. While the fixed value differs in the two cases, it cancels under the renormalization of conditioning the state on the allowed tuples of sequences.\footnote{For a simple example, imagine we have the two states: $\rho_{1} = y p_{00} \dyad{00} + x p_{01} \dyad{01} + x p_{10} \dyad{10} + z p_{11} \dyad{11}$ and $\rho_{2} = y' p_{00} \dyad{00} + x' p_{01} \dyad{01} + x' p_{10} \dyad{10} + z' p_{11} \dyad{11}$. Say we condition only on the strings of Hamming weight one. Then, we have $\rho_{1|\Omega} = \frac{1}{xp_{01} + xp_{10}}(x p_{01} \dyad{01} + x p_{10} \dyad{10})$ and $\rho_{2|\Omega} = \frac{1}{x'p_{01} + x'p_{10}}(x' p_{01} \dyad{01} + x' p_{10} \dyad{10})$. These conditional states are the same because the renormalization cancels the fixed (but different) pre-factor $x$ (resp. $x'$) for all the strings we are interested in.}
	\end{enumerate}
	It therefore follows under the renormalization that the distribution over the sequences in both cases is identical (and $\rho_{E}^{(a^{n}_{1},b^{n}_{1},p^{n}_{1})}$ for each sequence is the same too since both are based on $(a^{n}_{1},b^{n}_{1},p^{n}_{1})$ followed by the introduction of an independent sequence and then deterministic functions).
\end{proof}
Note this equivalence was able to be derived in effect because the testing in both cases is independent and the probability is uniform over sequences contained in the event $\Omega$ in question beyond the conditional probabilities that arise from the measurement, and since the measurements and announcements are the same process in both cases, this is not altered.

Finally, we explain the derivation of the bound presented below \Cref{thm:EATOutvsTrueOut}:
\begin{align*}
	\log( \vBar \rho^{out}_{EAT}[\Omega]) &= \log(\vBar) + \log(2^{-nh(\gamma)} {n \choose m} \rho^{out}[\Omega]) \\
	&= \log(\vBar \rho^{out}[\Omega]) - nh(\gamma) + \log{n \choose m} \\
	&\geq \log(\vBar \rho^{out}[\Omega]) + \log( \frac{e^2}{\sqrt{2\pi}} \sqrt{\gamma(1-\gamma)n}) \\
	& = \log(\vBar \vPE) + \log( \frac{e^2}{\sqrt{2\pi}} \sqrt{\gamma(1-\gamma)n})     \ ,
\end{align*}
where the inequality follows from using  $\sqrt{2\pi} n^{n+1/2}e^{-n}\leq n! \leq en^{n+1/2}e^{-n}$  \cite{Romik2000} to bound ${n \choose m}$.

\section{Security of QKD using EAT with Smoothing}\label{app:EAT-Sec-with-Smoothing}
In this section we present the security proof for smooth min-entropy rather than sandwiched R{\'e}nyi entropies. 
\subsection{Background smooth entropy results}
For completeness, in this section we present the max-tradeoff function and original version of the EAT \cite{Dupuis2016} for max-tradeoff functions which we make use of in the security proof in this Section. Likewise, we give an account of the Leftover Hashing Lemma \cite{Renner2005} and how it is applied in \Cref{thm:keyLengthWithSmoothing}.

We first state the min-tradeoff function in terms of smooth min-entropy.
\begin{prop}[Theorem V.2 of \cite{Dupuis2019} (Special Case)]
	\label{prop:EAT2}
	Consider EAT channels $\cM_{1},...,\cM_{n}$ and $\rho_{S^{n}_{1}P^{n}_{1}X^{n}_{1}E}$ as defined in \Cref{eq:outputStateForm} such that it satisfies the Markov conditions [\Cref{eq:Markov_cond}]. Let $h \in \mathbb{R}$, $\alpha \in (1,2)$, $f$ be a min-tradeoff function for $\cM_1,\dots,\cM_n$,  and let $\bar{\varepsilon} \in (0,1)$. Let $S_{i}$ always be classical. Then, for any event $\Omega \subseteq \cX^n$ that implies $f(\freq{X_1^n}) \geqslant h$,
	\begin{align}
		\label{eqn:eat-min_v2} H_{\min}^{\bar{\varepsilon}}(S_1^n | P_1^n E)_{\rho_{|\Omega}} & \geq nh - n \frac{(\alpha-1)\ln 2}{2}V^{2} - \frac{1}{\alpha-1} \log \frac{2}{\vBar^{2} \rho[\Omega]^{2}} - n(\alpha-1)^{2} K_{\alpha}
	\end{align}
	holds for 
	\begin{align}
		V &= \sqrt{\Var{f}+2} + \log(2d_{S}^{2}+1) \; \label{eq: EAT constant V_EATv2} \\
		K_\alpha &= \frac{1}{6(2-\alpha)^{3} \ln 2} 2^{(\alpha-1)(\log d_{S} + \Max{f} - \MinSigma{{f})}} \ln^{3} \left( 2^{\log d_{S} + \Max{f} - \MinSigma{f}} + e^{2} \right) \label{eq: EAT constant Kalpha_EATv2}
	\end{align}
	where $\Max{f})$, $\Min{f}$, $\mathrm{Min}_{\Sigma}(f)$ and $\mathrm{Var}(f)$ are defined in \Cref{eq:definition_minmaxvar} and $d_S = \max_{i\in [n]} |S_i|$ is the maximum dimension of the systems~$S_i$.

	Moreover, for a specific choice of $\alpha \in (1,2)$ \cite{Dupuis2019}, 
	\begin{align*}
		H_{\min}^{\bar{\varepsilon}}(S_1^n | P_1^n E)_{\rho_{|\Omega}} & > n h - c \sqrt{n} - c'
	\end{align*}
	where $c$ and $c'$ are functions of $\vBar$, $d_{S}$, $\rho[\Omega]$, and properties of the tradeoff function as follows:
	\begin{align}
	    c &= \sqrt{2 \ln 2} \left( \log (2d_S^2+1) + \sqrt{2 + \Var{f}} \right) \sqrt{ 1-2\log (\bar{\varepsilon} \rho[\Omega])}\; \label{eq:EAT_constant1} \\
	    c' &= \frac{35(1-2\log (\bar{\varepsilon} \rho[\Omega]))}{\left( \log(2 d_S^2 + 1) + \sqrt{2+\Var{f}} \right)^2} 2^{\log d_S + \Max{f} -  \MinSigma{f} } \ln^3\left( 2^{ \log d_S + \Max{f} -  \MinSigma{f}} + e^2 \right) \label{eq:EAT_constant2} \ .
	\end{align}
\end{prop}

Next we recall the notion of max-tradeoff function.
\begin{definition}
	Following the notation of \Cref{sec:EATBackground}. An affine function $f$ on $\mathbb{P}(\cX)$ is called max-tradeoff function $f$ for an EAT Channel $\cM_{i}$ if it satisfies
	\begin{align*}
		f(\vect{q}) \geq  \underset{\nu \in \Sigma_{i}(\vect{q})}{\sup} H(S_{i}|P_{i}R)_{\nu}, \hspace{.5cm} \forall \vect{q} \in \mathbb{P}(\cX) \ .
	\end{align*}
\end{definition}

We now state the EAT for max-tradeoff functions from \cite{Dupuis2016}.
\begin{theorem}[Part of Theorem 4.4 of \cite{Dupuis2016}]\label{thm:EATv1} 
	Let $\mathcal{M}_1,\dots,\mathcal{M}_n$ be EAT channels and $\rho_{S_1^n P_1^n X_1^n E}$ be of the form in \Cref{eq:outputStateForm} which satisfies the Markov conditions from \Cref{eq:Markov_cond}. Let  $h \in \mathbb{R}$, let $f$ be an affine max-tradeoff function for $\mathcal{M}_1,\dots,\mathcal{M}_n$,  and let  $\bar{\varepsilon} \in (0,1)$. Then, for any event $\Omega \subseteq \mathcal{X}^n$ that implies $f(\freq{X_1^n}) \leqslant h$,
	\begin{align}
		H_{\max}^{\bar{\varepsilon}}(S_1^n | P_1^n E)_{\rho_{|\Omega}} & < n h +  c \sqrt{n}
	\end{align}
	holds for $c = 2 \bigl(\log (1+2 d_S) + \left\lceil \| \nabla f \|_\infty \right\rceil  \bigr) \sqrt{1- 2 \log (\bar{\varepsilon} \rho[\Omega])}$, where $d_S:= \max_{i\in [n]}\{|S_i|\}$.
\end{theorem}

This completes the statement of the EAT for max-tradeoff functions. We now turn our attention to the Leftover Hashing Lemma (LHL). The LHL tells us that applying a two-universal hash function to a somewhat secret string constructs a shorter $\varepsilon$-secret key. The notion of secrecy is measured by distance from uniformity given Eve's side-information.

\begin{theorem}[Corollary 5.6.1 of \cite{Renner2005}]\label{thm:LHL}
	Let $\rho_{XE} \in \mathrm{D}(XE)$ be a classical-quantum state where $X \cong \mathbb{C}^{|\Sigma|}$ where $\Sigma$ is some finite set. Let $\mathcal{F}$ be a family of two-universal hash functions from $\Sigma \to \{0,1\}^{\ell}$ and $\varepsilon > 0$. Then
	\begin{align*}
		\frac{1}{2} \left \| \rho_{F(X)EF} - \frac{1}{|\{0,1\}^{\ell}|} \mathbb{1}_{F(X)} \otimes \rho_{EF} \right \|_{1} \leq \varepsilon + 2^{- \frac{1}{2} (H^{\varepsilon'}_{\min}(X|E) - \ell - 2)} \ .
	\end{align*}
\end{theorem}
\noindent From this we obtain the following corollary immediately.
\begin{corollary}\label{corr:LHL}
	Let $\rho_{XE} \in \mathrm{D}(XE)$ be a classical-quantum state. Let $\varepsilon_{sec} := \varepsilon' + 2^{- \frac{1}{2} (H^{\varepsilon'}_{\min}(X|E) - \ell - 2)}$ and $\varepsilon' > 0$. Let $\vPA := \varepsilon_{sec} - \varepsilon' > 0$. Then for an $\varepsilon_{sec}$-secret key, the length of the key, $\ell$, must satisfy the following bound:
	\begin{align*}
		\ell \leq H^{\varepsilon'}_{\min}(X|E) - 2\log(\frac{2}{\vPA}) \ .
	\end{align*}
\end{corollary}

\subsection{Technical lemmas for smooth entropy key rate}
The following lemma shows that generating a classical register from the classical side information does not change the value of the smooth min-entropy.
\begin{lemma}\label{lem: min entropy extra side information}
	For any CQ state $\rho_{AB} = \sum_i p_i \rho_A^i \ox \dyad{i}_B$ and $\rho_{ABC} = \sum_i p_i \rho_A^i \ox \dyad{i}_B \ox \dyad{f(i)}_C$ with a given function $f$, it holds that
	\begin{align}
		H_{\min}^{\varepsilon}(A|B) = H_{\min}^{\varepsilon}(A|BC).
	\end{align}
\end{lemma}
\begin{proof}
	This can be seen by defining isometry $V:\ket{i}_B \to \ket{i}_B\ket{f(i)}_C$. Then $\rho_{ABC} = V\rho_{AB}V^\dagger$. The result  then holds by the fact that the smooth min-entropy is invariant under local isometries~\cite[Corollary 6.11]{Tomamichel2016}. Note that this holds for any function $f$ which may not necessarily be deterministic one (c.f.~\cite[Eqs. 6.77]{Tomamichel2016}).
\end{proof}

\begin{lemma}\label{lem:reconstructingPrivateRegister}
	Consider a CCQ state $\rho_{XPE} = \sum_{(x,p)} \dyad{x}{x} \otimes \dyad{p}{p} \otimes \rho_{E}^{(x,p)}$. Consider another classical register $Y$ which be constructed using a deterministic function on the registers $X$ and $P$, i.e. $f: X \times P \to Y$. That is, there exists an isometry, $V \equiv \sum_{(x,p)} \dyad{x} \otimes \dyad{p}{p} \otimes \ket{f(x,p)}_{Y}$ reconstructing $Y$ from $X$ and $P$. Then
	\begin{align}
		H^{\ve}_{\min}(YX|PE)_{V \rho V^{\dagger}} = H^{\ve}_{\min}(X|PE)_{\rho} \ . 
	\end{align}
\end{lemma}
\begin{proof}
	To prove this proposition, we simply prove \begin{align}
		& H^{\ve}_{\min}(YX|PE)_{V \rho V^{\dagger}} \geq H^{\ve}_{\min}(X|PE)_{\rho} \label{eqn:propProofGoal1} \ ,\\ 
		& H^{\ve}_{\min}(YX|PE)_{V \rho V^{\dagger}} \leq H^{\ve}_{\min}(X|PE)_{\rho} \ . \label{eqn:propProofGoal2}
	\end{align} 
	We start with \Cref{eqn:propProofGoal1} as it is straightforward using previous results. Namely, we see:
	\begin{align}
		H^{\ve}_{\min}(YX|PE)_{V \rho V^{\dagger}} \geq H^{\ve}_{\min}(X|PE)_{V\rho V^{\dagger}} = H^{\ve}_{\min}(X|PE)_{\rho} \ .
	\end{align}
	The first inequality comes from the fact that the entropy of a classical register is non-negative~\cite[Lemma 6.17]{Tomamichel2016}. The equality follows from noting that $\Tr_{Y} (V \rho V^{\dagger}) = \rho.$ Therefore we can focus on the other direction of the inequality.
	
	By combining~\cite[Definitions 6.2 and 6.9]{Tomamichel2016} we note the definition of the smooth min-entropy terms we are concerned with:
	\begin{align}
		H^{\ve}_{\min}(YX|PE)_{V\rho V^{\dagger}} & := \underset{\overline{\rho} \in \mathcal{B}^{\ve}(V \rho V^{\dagger})}{\max} \,
		\underset{\overline{M}_{PE} \in \Density_{\leq}(PE) }{\max} \max \left\{\overline{\lambda} \in \mathbb{R} : \mathbb{1}_{Y} \otimes \mathbb{1}_{X} \otimes 2^{-\overline{\lambda}} \overline{M}_{PE} \geq  \overline{\rho}_{YXPE} \right\} \\
		H^{\ve}_{\min}(X|PE)_{\rho} & := \underset{\widetilde{\rho} \in \mathcal{B}^{\ve}(\rho)}{\max} \, \underset{\widetilde{M}_{PE} \in \Density_{\leq}(PE) }{\max} \max \left\{ \widetilde{\lambda} \in \mathbb{R} : \mathbb{1}_{X} \otimes 2^{- \widetilde{\lambda}} \widetilde{M}_{PE} \geq  \widetilde{\rho}_{XPE} \right\} \ ,
	\end{align}
	where $\Density_{\leq}$ is the set of sub-normalized states.
	One can see from these definitions it would be sufficient to show the optimal solution for 
	$H^{\ve}_{\min}(YX|PE)_{V\rho V^{\dagger}}$ is a feasible solution for  $H^{\ve}_{\min}(X|PE)_{\rho}$ to prove \Cref{eqn:propProofGoal2}. Suppose the optimal solution of $ H^{\ve}_{\min}(YX|PE)_{V\rho V^{\dagger}}$ is taken at  $(\overline{\lambda}, \overline{M}_{PE}, \overline{\rho}_{YXPE})$ where $\overline{\rho}_{YXPE}$ is chosen as a CCCQ state by~\cite[Lemma 6.13]{Tomamichel2016}. This gives 
	\begin{align}
		\mathbb{1}_{Y} \otimes \mathbb{1}_{X} \otimes 2^{- \overline{\lambda}}  \overline{M}_{PE} \geq  \overline{\rho}_{YXPE}:= \sum_{(y,x,p)}
		\dyad{y} \otimes \dyad{x} \otimes \dyad{p} \otimes \overline{\rho}^{(y,x,p)}_{E} \ .
	\end{align}
	Applying the pinching map $\Delta_P(\cdot):= \sum_p \ket{p}\bra{p} \cdot \ket{p}\bra{p}$ on both sides, we get
	\begin{align}
		\sum_{(y,x,p)}  \dyad{y} \otimes \dyad{x} \otimes  \dyad{p} \ox 2^{- \overline{\lambda}} \overline{M}_{E}^{(p)} \geq  \sum_{(y,x,p)}
		\dyad{y} \otimes \dyad{x} \otimes \dyad{p} \otimes \overline{\rho}^{(y,x,p)}_{E},
	\end{align}
	with $\overline{M}_{E}^{(p)}:= \<p|\overline{M}_{PE}|p\>$. We can further `un-compute' $Y$ by applying $V^\dagger \cdot V$ on both sides,
	\begin{align}
		\sum_{(x,p)} \dyad{x} \otimes  \dyad{p} \ox 2^{- \overline{\lambda}} \overline{M}_{E}^{(p)} \geq  \sum_{(x,p)} \dyad{x} \otimes \dyad{p} \otimes \overline{\rho}^{(f(x,p),x,p)}_{E},
	\end{align}
	which is equivalent to
	\begin{align}
		\mathbb{1}_{X} \ox 2^{- \overline{\lambda}} \sum_p \dyad{p} \ox  \overline{M}_{E}^{(p)} \geq  V^\dagger \overline{\rho}_{YXPE} V \ .
	\end{align}
	Since $V^\dagger \cdot V$ is a CPTNI map, we know that the purified distance cannot increase under this map. Therefore,
	\begin{align}
		P(V^\dagger \overline{\rho}_{YXPE} V, \rho_{XPE}) \leq P(\overline{\rho}_{YXPE}, V\rho_{XPE}V^\dagger) \leq \ve,
	\end{align}where the second inequality follows by the assumption of $\overline{\rho}_{YXPE}$. Thus, it is the case that $(\overline{\lambda},\sum_p \dyad{p} \ox  \overline{M}_{E}^{(p)}, V^\dagger \overline{\rho}_{YXPE} V)$ is a feasible solution for $H^{\ve}_{\min}(X|PE)_{\rho}$, which implies \Cref{eqn:propProofGoal2}.
\end{proof}

\subsection{Entropy Accumulation subprotocol with smoothing}
Here we present the equivalent of \Cref{thm:entRate}. The proof is the same as in the main text except that we need to use the smooth min-entropy EAT (\Cref{prop:EAT2}).
\begin{shaded}
	\begin{theorem} 
		\label{thm:entRatewithSmoothing} 
		Consider the entropy accumulation protocol defined in \Cref{prot:DDEA} where the announcements $\widetilde{A}_i$'s and $\widetilde{B}_i$'s satisfy \Cref{assumption:block_diagonal}. Let $\Omega = \{x^{n}_{1} \in \cX^{n}: \freq{x^{n}_{1}} \in \mathcal{Q}\}$ and $\rho$ be the output of the protocol. Let $h$ such that $f(\q) \geq h$ for all $\q \in \mathcal{Q}$ where $f$ is the min-tradeoff function generated by \Cref{alg:Algorithm1_min_tradeoff_function} or \Cref{alg:Algorithm2_min_tradeoff_function}. Then for any $\alpha \in (1,2)$, $\vEA,\vBar \in (0,1)$, either the protocol aborts with probability greater than $1-\vEA$, or
		\begin{align}\label{eq:entropyRate_with_smoothing}
			H^{\vBar}_{\min}(\overline{A}^{n}_{1} \overline{B}^{n}_{1}| \widetilde{A}^{n}_{1} \widetilde{B}^{n}_{1} T^{n}_{1} E)_{\rho_{|\Omega}} > n h -n \frac{(\alpha-1)\ln 2}{2}V^{2} - \frac{1}{\alpha-1} \log \frac{2}{\vBar^{2} \vEA^{2}} - n(\alpha-1)^{2} K_{\alpha}
		\end{align}
		where $V$ and $K_{\alpha}$ are given by \Cref{eq: EAT constant V_EATv2,eq: EAT constant Kalpha_EATv2}  and we replace $\rho[\Omega]$ by $\vEA$.\\
		
		\noindent Moreover, the bound can be further optimized by using \Cref{eqn:eat-min_v2} and optimizing over $\alpha$.
	\end{theorem}
\end{shaded}

\subsection{QKD security proof with smooth entropies}
Here we present the equivalent of \Cref{thm:keylengthwithoutsmoothing}. The proof is similar, but we include it for completeness.
\begin{shaded}
	\begin{theorem}\label{thm:keyLengthWithSmoothing} 
		Consider any QKD protocol which follows \Cref{prot:physicalDDQKD} and satisfies \Cref{assumption:block_diagonal}. Let $\vBar,\vPA,\vEC,\vAcc \in (0,1)$ such that $\vAcc \geq \vBar + \vPA + \vEC$. Let $h$ such that $f(\vect{q}) \geq h$ for all $\vect{q} \in \mathcal{Q}$ where $f$ is the min-tradeoff function generated by \Cref{alg:Algorithm1_min_tradeoff_function} or \Cref{alg:Algorithm2_min_tradeoff_function}. For $\alpha \in (1,2)$, the QKD protocol is $\vAcc$-secure for key length
		\begin{align}
			\label{eq:KeyLength_with_smoothing} 
			\ell &\leq nh - {\leak}_{\vEC} - n\frac{(\alpha-1)\ln 2}{2}V^{2} - n(\alpha-1)^{2}K_{\alpha} - n\gamma \log(|\cA||\cB|) \notag \\
			& \hspace{.25cm} - 2\sqrt{n}\log(1+2d_{S})\sqrt{1-2\log(\vBar/4 \cdot \vAcc)} - \frac{1}{\alpha-1}\log \frac{32}{(\vBar \cdot \vAcc)^{2}}  \\
			& \hspace{.5cm} -2\log(1-\sqrt{1-(\vBar/4)^{2}}) - 2\log(2/\vPA)  \notag
		\end{align}
		where $V$ and $K_{\alpha}$ are given by \Cref{eq: EAT constant V_EATv2,eq: EAT constant Kalpha_EATv2},
		$\mathcal{A}$ and $\mathcal{B}$ are the alphabets of private outcomes Alice and Bob will announce excluding the symbol $\perp$, $d_{S} = (|\cA|+1)(|\cB|+1)$, $\leak_{\vEC}$ is the amount of information leaked during error correction.
	\end{theorem}
\end{shaded}

\begin{proof} 
	By the equivalence of \Cref{prot:physicalDDQKD} and \Cref{prot:DDQKD}, which is by construction, proving the security of \Cref{prot:physicalDDQKD} is equivalent to proving the security of \Cref{prot:DDQKD}. We refer to \Cref{prot:DDQKD} in the proof, but this is just for the ease of the reader to link to the main text and \Cref{prot:DDQKD}'s sequential presentation. Let $\Omega'$ be the event that \Cref{prot:DDEA} does not abort (i.e. $\Omega$ is satisfied) and error correction succeeds (i.e. $\hash(\mathbf{Z}) = \hash(\mathbf{\widehat{Z}})$). Let $\Omega''$ be the event that \Cref{prot:DDQKD} does not abort and error correction verification succeeds. Let $\tilde{\rho}_{|\Omega''}$ be the output of the \Cref{prot:DDQKD} conditioned on not aborting and error correction succeeding and $\rho_{|\Omega'}$ be the output of \Cref{prot:DDEA} conditioned on not aborting \textit{and} error correction being applied and succeeding.
	
	Fundamentally, we are interested in the entropy of the raw key which excludes the sifted registers or the registers used for parameter estimation. As defined in \Cref{prot:DDQKD}, $\mathbf{Z} \equiv \ol{A}^{n}_{1} \setminus (\ol{A}_{\mathcal{S}} \cup \ol{A}_{\mathcal{T}})$. However, we note that Alice announces $\ol{A}_{\mathcal{T}}$, making it part of what is conditioned on, and $\ol{A}_{\mathcal{S}}$ is determined entirely by the registers $\wt{A}^{n}_{1}$, $\wt{B}^{n}_{1}$, and $T^{n}_{1}$. Therefore, by \Cref{lem:reconstructingPrivateRegister}, we can conclude:
	$$ H_{\min}^{\vBar}(\mathbf{Z}|\wt{A}^{n}_{1}\wt{B}^{n}_{1} T^{n}_{1} \ol{A}_{\mathcal{T}} C E)_{\tilde{\rho}_{|\Omega''}} = H_{\min}^{\vBar}(\ol{A}^{n}_{1}|\wt{A}^{n}_{1}\wt{B}^{n}_{1} T^{n}_{1} \ol{A}_{\mathcal{T}} C E)_{\tilde{\rho}_{|\Omega''}} \ . $$
	
	We can handle the classical communication cost of error correction by \cite[Lemma 6.17]{Tomamichel2016}:
	\begin{align}
		H^{\vBar}_{\min}(\ol{A}^{n}_{1}|\wt{A}^{n}_{1}\wt{B}^{n}_{1} T^{n}_{1} \ol{A}_{\mathcal{T}} C E)_{\tilde{\rho}_{|\Omega''}} &\geq H^{\vBar}_{\min}(\ol{A}^{n}_{1}|\wt{A}^{n}_{1}\wt{B}^{n}_{1} T^{n}_{1} \ol{A}_{\mathcal{T}} E)_{\tilde{\rho}_{|\Omega''}} - {\leak}_{\vEC} \\
		& = H^{\vBar}_{\min}(\ol{A}^{n}_{1}|\wt{A}^{n}\wt{B}^{n} T^{n}_{1} \ol{A}_{\mathcal{T}} E)_{\rho_{|\Omega'}} - {\leak}_{\vEC}
	\end{align}
	where $C$ is the register storing the classical communication due to error correction, and the equality comes from noting that $\rho_{\ol{A}^{n}_{1}\wt{A}^{n}_{1}\wt{B}^{n}_{1}T^{n}_{1}X^{n}_{1}E_{|\Omega'}} = \tilde{\rho}_{\ol{A}^{n}_{1}\wt{A}^{n}_{1}\wt{B}^{n}_{1}T^{n}_{1}X^{n}_{1}E_{|\Omega''}}$. We stress that Bob's private register, $\ol{B}^{n}_{i}$, for both the EAT and the QKD protocol, is excluded in the equality as they differ. We note that $\ol{A}_{\mathcal{T}}$ is conditioned upon as Alice publicly announced it, but $X^{n}_{1}$ is not conditioned upon as Bob computes it locally and does not need to announce it unless aborting.
	
	With all of this addressed, we can conclude that we would like to do privacy amplification using the smooth min-entropy term:
	\begin{align}
		\label{eq:goal}
		H^{\vBar}_{\min}(\ol{A}^{n}_{1}|\wt{A}^{n}_{1}\wt{B}^{n}_{1} T^{n}_{1} \ol{A}_{\mathcal{T}} E)_{\rho_{|\Omega'}} \ .
	\end{align}
	\sloppy This imposes one major difficulty: we need to find a way to relate this to the term $H^{\vBar}_{\min}(\ol{A}^{n}_{1} \ol{B}^{n}_{1}|\wt{A}^{n}\wt{B}^{n} T^{n}_{1} E)_{\rho_{|\Omega'}}$ on the left-hand side of \Cref{eq:entropyRate_with_smoothing}. 
	 
	Following derivations in the DI-QKD setting \cite{Arnon2018}, but also handling the $\ol{A}_{\mathcal{T}}$ register, we use the following series of claims:
	\begin{align}
		& H^{\vBar}_{\min}(\ol{A}^{n}_{1}|\wt{A}^{n}_{1} \wt{B}^{n}_{1} T^{n}_{1} \ol{A}_{\mathcal{T}} E)_{\rho_{|\Omega'}} \notag \\
		\geq & H^{\vBar}_{\min}(\ol{A}^{n}_{1}|\wt{A}^{n}_{1} \wt{B}^{n}_{1} \ol{A}_{\cT} \ol{B}^{n}_{1} T^{n}_{1} \ol{A}_{\mathcal{T}} E)_{\rho_{|\Omega'}} \notag \\
		\geq & H^{\vBar/4}_{\min}(\ol{A}^{n}_{1} \ol{B}^{n}_{1} \ol{A}_\cT |\wt{A}^{n}_{1}\wt{B}^{n}_{1} T^{n}_{1} E)_{\rho_{|\Omega'}} - H^{\vBar/4}_{\max}(\ol{B}^{n}_{1}\ol{A}_{\cT}|\wt{A}^{n}_{1}\wt{B}^{n}_{1} T^{n}_{1} E)_{\rho_{|\Omega'}}  - 2 \log \left( 1 - \sqrt{1-(\vBar/4)^2} \right) \notag \\
		\geq & H^{\vBar/4}_{\min}(\ol{A}^{n}_{1} \ol{B}^{n}_{1} |\wt{A}^{n}_{1}\wt{B}^{n}_{1} T^{n}_{1} E)_{\rho_{|\Omega'}} - H^{\vBar/4}_{\max}(\ol{B}^{n}_{1}\ol{A}_{\cT}|T^{n}_{1} E)_{\rho_{|\Omega'}} - 2 \log \left( 1 - \sqrt{1-(\vBar/4)^2} \right) \ .
	\end{align}
	The first inequality holds by data-processing. The second inequality follows from \cite[Eq. (6.60)]{Tomamichel2016}. The third inequality comes from data-processing for the smooth max-entropy term and applying \Cref{lem:reconstructingPrivateRegister} to remove $\ol{A}_{\cT}$ in the smooth min-entropy term. More specifically, the strong subadditivity follows from the data-processing inequality of smooth max-entropy of \cite[Theorem 6.19]{Tomamichel2016} where we have applied the partial trace map on the side information.
	
	Now all that is left is to bound the max-entropy term using the original version of EAT (\Cref{thm:EATv1}) \cite{Dupuis2016}. To do this we use the EAT in the form of \Cref{thm:EATv1} with the replacements $S_{i} \to \ol{A}_{\cT_{i}}\ol{B}_{i}$, $P_{i} \to T_{i}$, and $E \to E$. We note it is okay to use $\ol{A}_{\mathcal{T}}$ as we may treat it as if it were actually $n$ registers where it is $\perp$ whenever $T_{i} \neq 0$. The Markov conditions trivially hold since $T_{i}$ uses seeded randomness. Then using the fact that $\ol{A}_{\cT_{i}}\ol{B}_{i} = \perp \times \perp$ except when $T_{i}=1$ which happens with probability $\gamma$, we construct the max-tradeoff function analytically by the following observation:
	\begin{align}\label{eq:maxtradeoffconstruction_with_smoothing}
		H(\ol{A}_{\cT_{i}} \ol{B}_{i}|T_{i}R')_{\cM_{i}(\omega)} & \leq H(\ol{A}_{\cT_{i}}\ol{B}_{i}|T_{i})_{\cM_{i}(\omega)}\notag\\
		& = (1-\gamma)H(\ol{A}_{\cT_{i}}\ol{B}_{i}|T_{i}=0)_{\cM_{i}(\omega)} + \gamma H(\ol{A}_{\cT_{i}}\ol{B}_{i}|T_{i}=1)_{\cM_{i}(\omega)} \notag \\
		& \leq \gamma \log|\mathcal{A} \times \mathcal{B}|
	\end{align}
	where $\mathcal{A}$ and $\mathcal{B}$ are the alphabets of private outcomes Bob utilizes during parameter estimation. This tells us we can let the max-tradeoff function $f_{\max}(\vect{q}) = \gamma \log |\mathcal{A} \times \mathcal{B}|$ for all $\vect{q} \in \mathbb{P}(\cX)$. This also implies $\|\nabla f_{\max}\|_{\infty} = 0$. Therefore, using this with \Cref{thm:EATv1}
	we get 
	\begin{align}
		& H^{\vBar/4}_{\max}(\ol{B}^{n}_{1}\ol{A}_{\cT}|T^{n}_{1} E)_{\rho_{|\Omega'}} \leq n\gamma\log|\mathcal{A} \times \mathcal{B}| + \sqrt{n} 2\log(1+2 d_{S})\sqrt{1-2\log(\vBar/4 \cdot \vAcc)}
	\end{align}
	where we have assumed $\rho[\Omega'] \geq \vAcc$, and $d_{S} := (|\cA|+1)(|\cB|+1)$ using that $|\ol{B}_{i}|=|\cB|+1, |\ol{A}_{\cT,i}|=|\cA|+1$.
	
	We now put all of these together to get the following:
	\begin{align}
		H^{\vBar}_{\min}(\ol{A}^{n}_{1}| \wt{A}^{n}_{1} \wt{B}^{n}_{1} T^{n}_{1} \ol{A}_{\mathcal{T}} E)_{\rho_{|\Omega'}} 
		&\geq H^{\vBar/4}_{\min}(\ol{A}^{n}_{1} \ol{B}^{n}_{1}|\wt{A}^{n}_{1} \wt{B}^{n}_{1} T^{n}_{1} E)_{\rho_{|\Omega'}} -\leak_{\vEC}  -n\gamma\log(|\mathcal{A} \times \mathcal{B}|) \notag\\ & \hspace{0.75cm}
		- 2\sqrt{n}\log(1+2d_{S})\sqrt{1-2\log(\vBar/4 \cdot \vAcc)} \notag\\  &\hspace{1cm}  - 2 \log \left( 1 - \sqrt{1-(\vBar/4)^2} \right) \ .
	\end{align}
	We now apply \Cref{thm:entRatewithSmoothing} to the right-hand-side min-entropy term (with appropriate choice of $\vBar/4$ and $\vEA \to \vAcc$) so that we have that either the input state aborts with probability greater than $1-\vAcc$ or 
	\begin{equation}
		\begin{aligned}
			& H^{\vBar}_{\min}(\ol{A}^{n}_{1}| \wt{A}^{n}_{1} \wt{B}^{n}_{1} T^{n}_{1} \ol{A}_{\mathcal{T}} E)_{\rho_{|\Omega'}} \\  
			& \geq  nh - \leak_{\vEC} -n \frac{(\alpha-1)\ln 2}{2}V^{2} - n(\alpha-1)^{2} K_{\alpha} -n\gamma\log(|\mathcal{A} \times \mathcal{B}|) \\
			& \hspace{0.5cm} -2\sqrt{n}\log(1+2 d_{S})\sqrt{1-2\log(\vBar/4 \cdot \vAcc)} - 2 \log \left( 1 - \sqrt{1-(\vBar/4)^2} \right) \\
			& \hspace{1cm}  - \frac{1}{\alpha-1} \log \frac{32}{(\vBar \cdot \vAcc)^{2}} \ .
		\end{aligned}
	\end{equation}
	Applying the leftover hash lemma \cite{Renner2005} (\Cref{corr:LHL}) completes the proof. 
	
	The proof of the security claim is nearly identical to the non-smooth case. We consider two mutually exclusive and exhaustive cases. When $\rho[\Omega] < \vAcc$, then the output is trivially $\vAcc$-secure. Now consider $\rho[\Omega]$. Starting from the subnormalized state\footnote{Subnormalization is due to aborting the protocol on the input} $\rho_{S_{A}EF}$ where $S_{A}:= h(\mathbf{Z})$ is the hashed key, $h$ is the hash function, and $F$ is the announcement of the hash function,
	\begin{align*}
	     \frac{1}{2} \| \rho_{S_{A}EF} - |S_{A}|^{-1}\idm_{S_{A}} \otimes \rho_{EF} \|_{1} 
	    =& \frac{1}{2} \rho[\Omega'] \| \rho_{S_{A}F_{|\Omega'}} - |S_{A}|^{-1}\idm_{S_{A}} \otimes \rho_{EF_{|\Omega'}} \|_{1} \\
	    \leq& \frac{1}{2}\| \rho_{S_{A}F_{|\Omega'}} - |S_{A}|^{-1}\idm_{S_{A}} \otimes \rho_{EF_{|\Omega'}} \|_{1} \ , \\
	    \leq& \vSec = \vBar + \vPA
	\end{align*}
	where the second inequality is by \cref{thm:LHL} along with the definition of $\vSec$ in \Cref{corr:LHL} by choosing $\ve' \to \vBar$. The final inequality is by the definition $\vPA := \vSec - \ve'$ in \Cref{corr:LHL}. The correctness argument is the same as in the unsmoothed proof. Therefore we the key is $\max\{\vAcc,\vBar+\vPA+\vEC\}$-secure. As in the previous proofs, note that if $\vAcc < \vBar+\vPA+\vEC$, you could improve the key rate without changing the security claim by setting $\vAcc = \vBar + \vPA + \vEC$. This completes the proof.
\end{proof}

\section{Key Rate Optimization Based on Quantum Conditional Entropy}\label{app:key_rate_formula}

We find a way to represent Eve's conditional states in terms of Alice's and Bob's joint state and express the quantum conditional entropy $H(S_i|P_iR)$ with those conditional states. (We use register $R$ to refer to Eve's register in a single round as depicted in \Cref{fig:EATProcess}.) We then show this function is convex. We write the objective function based on the quantum conditional entropy as
\begin{aeq}\label{eq:objective_function_conditional_entropy}
	W(\rho)&:= H(S_i|P_iR)_{\rho}, 
	\end{aeq}for $\rho \in \DM(Q_i)$. \Cref{prop:conditional_entropy_form} gives us a more explicit expression of $W(\rho)$ in terms of $\rho \in \DM(Q_i)$ and Alice's and Bob's joint POVM. The convexity of $W(\rho)$ is given in \Cref{prop:convexity_conditional_entropy}.
We now prove \Cref{prop:conditional_entropy_form}.
\begin{proof}[Proof of \Cref{prop:conditional_entropy_form}]
	Recall from the definition of conditional entropy that for a classical-quantum state $\rho_{S_iP_iR} = \sum_{s,p} \dyad{s}{s}_{S_i} \otimes \dyad{p}{p}_{P_i} \otimes \rho_R^{(s,p)}$, we have $H(S_i|P_iR) = \sum_{s,p} H(\rho_R^{(s,p)}) - \sum_p H(\rho_R^{(p)})$, where $\rho_R^{(p)} = \sum_{s} \rho_R^{(s,p)}$. We need to find the marginal states $\rho_R^{(s,p)} = \Tr_{AB}[ (M_{sp} \otimes \idm_R) \rho_{Q_iR}]$. Without loss of generality, we can write the purifying state as $\rho_{Q_iR} = \dyad{\psi}{\psi}$, where $\ket{\psi} = \sqrt{\rho}\otimes \idm_R \ket{\Phi^+}$ and $\ket{\Phi^+} = \sum_j \ket{j}_{Q_i}\ket{j}_R$, from which we find
	\begin{aeq}
		\rho_R^{(s,p)} &= \Tr_{Q_i}[ (M_{sp} \otimes \idm_R) \rho_{Q_iR}]\\
		&= \Tr_{Q_i}[(\sqrt{\rho}M_{sp}\sqrt{\rho} \otimes \idm_R) \dyad{\Phi^+}{\Phi^+}]\\
		&= \sum_{s,p} \left(\sqrt{\rho}M_{sp}\sqrt{\rho}\right)^T \ .
	\end{aeq}By the fact that for any operator $A$, $(A^\dagger A)$ and $(A^\dagger A)^T$ have the same spectrum, it is the case that $H(\rho_R^{(s,p)}) = H\Big([(\sqrt{\rho}K_{sp}) (K_{sp}\sqrt{\rho})]^T\Big) = H(K_{sp} \rho K_{sp})= H(K_{sp} \rho K_{sp}^{\dagger})$ and a similar expression holds for $H(\rho_R^{(p)})$.
\end{proof}

\begin{prop}\label{prop:convexity_conditional_entropy}
	The function $W(\rho)$ is convex.
\end{prop} 
\begin{proof}
	Let $\bar \rho_Q = \sum_\lambda p_\lambda \rho^\lambda_Q$ be a convex decomposition of $\bar \rho_Q$. We want to show that
	\begin{IEEEeqnarray}{rL}
		\sum_\lambda p_{\lambda} W(\rho_Q^\lambda) \geq W(\bar \rho_Q)\,.
	\end{IEEEeqnarray}	
	Let $\rho_{QR}^\lambda$ be the purification of $\rho^\lambda_Q$. From these, we define the state $\rho_{QR\Lambda} = \sum_\lambda p_{\lambda} \rho^\lambda_{QR} \otimes \ketbra \lambda_\Lambda$, which satisfies $\Tr_{R\Lambda}[\rho_{QR\Lambda}] = \bar \rho_Q$. Then
	\begin{IEEEeqnarray}{rL}
		\sum_\lambda p_{\lambda} W(\rho_Q^\lambda) &:= \sum_\lambda p(\lambda) H(S|PR)_{\widetilde\cM\otimes \id_R(\rho^\lambda_{QR})}\\
		&= H(S|PR\Lambda)_{\widetilde\cM\otimes \id_{R\Lambda}(\rho_{QR\Lambda})}\\
		&\geq H(S|PR)_{\widetilde\cM\otimes \id_{R}(\bar \rho_{QR})}\\
		&:= W(\bar \rho_Q)
	\end{IEEEeqnarray}
	where $\bar \rho_{QR}$ is the purification of $\bar \rho_Q$ and where the inequality comes from subadditivity because the entropy will be minimized when Eve holds a purification of $\bar\rho_Q$.
\end{proof}
The physical intuition of the proof is the fact that if Eve has a family of states $\rho_{Q}^\lambda$ that she sends out with probability $p(\lambda)$, while keeping their purification, then that attack will be suboptimal. This is because it is better to send out the average state $\bar \rho_Q = \sum_\lambda p(\lambda)\rho_Q^\lambda$ and keep the corresponding purification, as this produces the same measurement statistics and Eve holds more information about the state.

As mentioned in the main text, we need to introduce a small perturbation to guarantee that the gradient $\nabla W(\rho)$ is well-defined for any $\rho \geq 0$. To do so, we introduced the perturbed objective function $W_\epsilon(\rho)$, which is
\begin{align}
	W_\epsilon(\rho)= \sum_{s,p} H\left(\cK^\epsilon_{sp}(\rho)\right) - \sum_p H\left( \cK^\epsilon_p(\rho)\right).
\end{align}Then the gradient $\nabla W_\epsilon(\rho)$ always exists for all $\rho \geq  0$:
\begin{align}
	\nabla W_\epsilon(\rho) = - \sum_{s,p} (\cK^\epsilon_{sp})^\dagger \left[\log \cK^\epsilon_{sp}(\rho) \right] + \sum_p (\cK_{p}^\epsilon)^\dagger \left[\log \cK^\epsilon_p(\rho)\right].
\end{align}
By~\cite[Lemma 7]{Winick2018}, we know that 
\begin{align}
	\norm{\cK_{sp}(\rho) - \cK_{sp}^\epsilon(\rho)}_1  \leq \epsilon(d-1) \quad \text{and} \quad \norm{\cK_{p}(\rho) - \cK_{p}^\epsilon(\rho)}_1 \leq \epsilon(d-1).
\end{align}
By the continuity of von Neumann entropy (known as Fannes' inequality), it holds for any $\epsilon \in (0,1/(e(d-1))]$~\footnote{Here $e$ is the base of the natural logarithm.} and any density operator $\rho$,
\begin{align}
	\left|H\left(\cK_{sp}(\rho)\right) - H\left(\cK_{sp}^\epsilon(\rho)\right)\right| \leq \zeta_\epsilon \quad \text{and}\quad
	\left|H\left(\cK_p(\rho)\right) - H\left(\cK_p^\epsilon(\rho)\right)\right| \leq \zeta_\epsilon.
\end{align}
with $\zeta_\epsilon = \epsilon(d-1)\log \frac{d}{\epsilon(d-1)}$. Thus we have \Cref{eq:continuity_bound} after taking the worst-case scenario among all contributing terms.

\section{Proof for the Second Algorithm Construction}\label{app:alg_proof}
\subsection{Basic definitions}
For completeness, we review some basic definitions. In this appendix, $f, g, h$ denote some general functions. Readers are cautioned not to confuse them with min-tradeoff functions used in the main text.
\begin{definition}[Affine hull]\label{def:affine_hull}
	For a set $C \subseteq \mbR^n$, the \emph{affine hull} of $C$, denoted by $\aff(C)$, is defined as
	\begin{aeq}
		\aff(C) = \Big\{\sum_{i=1}^k \theta_i x_i:  x_1, \dots, x_k \in C, \sum_{i=1}^k \theta_i = 1 \text{ for some }k \in \mbN\Big\}.
		\end{aeq} 
	\end{definition}
We note that the affine hull is the smallest affine set that contains $C$.
\begin{definition}[Relative interior]\label{def:relative_interior}
	The \emph{relative interior} of a set $C \subseteq \mbR^n$, denoted by $\relint(C)$ is defined as
	\begin{aeq}
		\relint(C) =\{x \in C: B_r(x) \cap \aff(C) \subseteq C \text{ for some } r > 0 \},
		\end{aeq}where $B_r(x)$ is a ball centered at $x$ with a radius $r$.
	\end{definition}

It is often convenient to let a function $f$ be defined over the entire $\mbR^n$ and let it take the value $+\infty$ outside its domain which is denoted by $\dom(f)$. It is thus interesting to talk about the extended-value extension.
\begin{definition}[Extended-value extension]\label{def:extended-value-extension}
	\index{extended-value extension}
	If $f:\dom(f)\subseteq\mbR^n \rightarrow \mbR$ is convex, we define its \emph{extended-value extension} $\tilde{f}: \mbR^n \rightarrow \mbR \cup \{+\infty\}$ by
	\begin{aeq}\label{eq:extended-value-extension}
		\tilde{f}(x)  = \begin{cases} 
			f(x) & x\in \dom(f) \\
			+ \infty & x\not\in \dom(f).
		\end{cases}
	\end{aeq} 
	\end{definition}
It is sometimes useful to define the domain of an extension as
\begin{aeq}\label{eq:domain_of_extension}
	\dom(\tilde{f})=\{x \in \mbR^n: \tilde{f}(x) < \infty\}.
	\end{aeq}

\subsection{Fenchel duality}\label{sec:fenchel_duality}
We review one form of duality called \emph{Fenchel duality} that is relevant for our algorithm. See \cite{Borwein2006} for details about Fenchel duality.
\begin{definition}[Fenchel conjugate]\label{def:fenchel_conjugate}
	\index{Fenchel conjugate}
	The \emph{Fenchel conjugate} of a function $h: \mbR^n \rightarrow [-\infty, +\infty]$ is the function $h^*: \mbR^n  \rightarrow [-\infty, +\infty]$ defined by 
	\begin{aeq}\label{eq:fenchel_conjugate}
		h^*(v) = \sup_{x \in \mbR^n} \{\langle v, x \rangle  - h(x)\}.
		\end{aeq}
	\end{definition}
 We now state Fenchel's duality theorem \cite[Theorem 3.3.5]{Borwein2006} with the condition for strong duality replaced by \Cref{eq:fenchel_duality_condition} according to \cite[page 74, Exercise 20(e)]{Borwein2006}.
 \begin{shaded}
\begin{theorem}[{Fenchel duality, \cite[Theorem 3.3.5]{Borwein2006}}]\label{thm:fenchel_duality}
	\index{Fenchel duality}
	For given functions $f: \mbR^n \rightarrow (-\infty, +\infty]$ and $g: \mbR^m \rightarrow (-\infty, +\infty]$ and a (bounded) linear map $A: \mbR^n \rightarrow \mbR^m$, let $p, d \in [-\infty, +\infty]$ be primal and dual values in the following Fenchel problems:
	\begin{align}
		p &= \inf_{x \in \mbR^n} \{f(x)+g(Ax)\} \label{eq:fenchel_primal}, \\
		d  &= \sup_{v \in \mbR^m} \{-f^*(A^{\dagger}v) - g^*(-v)\} \label{eq:fenchel_dual}.
		\end{align}These values satisfy the weak duality $p \geq d$. 
	
	If $f$ and $g$ are convex and satisfy the condition
	\begin{aeq}\label{eq:fenchel_duality_condition}
	\relint(\dom(g)) \cap A \relint(\dom(f)) \neq \emptyset,
		\end{aeq}where $A \relint(\dom(f))  = \{Ax: x \in \relint(\dom(f)) \}$, then the strong duality holds; i.e., $p=d$, and the supremum in the dual problem [\Cref{eq:fenchel_dual}] is attained if $d < + \infty$. \index{$\relint$: relative interior}
	\end{theorem}
\end{shaded}
We note that although the formulation of this theorem is for unconstrained optimization, one can easily handle constrained optimization by the extended-value extension. Similarly, one can define an indicator function for a set $C$ by 
\begin{aeq}\label{eq:indicator_function}
	\delta_{C}(x) = \begin{cases}
		0 & \text{ if }  x \in C \\
		\infty & \text{ if }  x \not\in C, 
		\end{cases}
	\end{aeq}and use this indicator function to convert a constraint set $C$ into a part of the objective function.

	\subsection{Deriving dual problem of  \Cref{eq:algorithm2_primal}}\label{app_sec:algorithm2details}
	In this section, let $\cH$ be a Hilbert space and define a map that stores the measurement outcomes of the POVM $\{M_{\lambda}\}$ with associated alphabet $\cX$ for outcomes as
\begin{aeq}\label{eq:PE_measurement_channel_vector}
	\Phi_{\vect{\cM}}(\sigma) = \sum_{\lambda \in \cX} \Tr( \sigma M_{\lambda}) \ket{\lambda}.
\end{aeq}Let $W$ be the quantum conditional entropy function in \Cref{eq:objective_conditional_entropy}. Let $m= \abs{\cX}$ be the number of outcomes. We define the set
	\begin{aeq}
	\cF_\dual:=\{\vect{f} \in \mbR^{m}: \vect{f} \cdot \Phi_{\vect{\cM}}(\rho) \leq W(\rho) \ \forall \rho \in \DM(\cH)\}.
	\end{aeq}It is also easy to see that $\cF_\dual$ is a convex set. Moreover, we can write the indicator function [see \Cref{eq:indicator_function}] for the set $\cF_\dual$ in a useful way by the following lemma.
\begin{lemma}\label{lemma:rewrite_dual_set_indicator_function}
Let $\delta_{\cF_\dual}$ be the indicator function for the set $\cF_\dual$. Then,
	\begin{aeq}\label{eq:algorithm2_dual_set_indicator_relation}
		\sup_{\rho}\{\langle \vect{f},\Phi_{\vect{\cM}}(\rho) \rangle - W(\rho): \rho \in \pos(\cH)\}=\delta_{\cF_\dual}(\vect{f}) \ \, \forall \vect{f} \in \mbR^m.
	\end{aeq}
	\end{lemma}
\begin{proof}
	Let $\alpha(\vect{f}) :=	\sup_{\rho}\{\langle \vect{f},\Phi_{\vect{\cM}}(\rho) \rangle - W(\rho): \rho \in \pos(\cH)\}$. 
	
	If $\vect{f} \in \cF_\dual$, then $\langle \vect{f},\Phi_{\vect{\cM}}(\rho) \rangle - W(\rho) \leq 0$ for any $\rho \in \DM(\cH)$. As any positive operator can be scaled from $\rho \in \DM(\cH)$ by a non-negative coefficient, it is the case that $\alpha \leq 0$. Thus, $\alpha(\vect{f}) = 0$ with $\rho = 0$. 
	
	If  $\vect{f} \not\in \cF_\dual$, then there exists $\rho_0 \in \DM(\cH)$ such that $\langle \vect{f},\Phi_{\vect{\cM}}(\rho_0) \rangle - W(\rho_0) > 0$. Let $\gamma = \langle \vect{f},\Phi_{\vect{\cM}}(\rho_0) \rangle - W(\rho_0)$ and $\rho_{\beta} = \beta \rho_0$ with $\beta \geq 0$. We notice that $W(\beta \rho) = \beta W(\rho)$ for any $\beta \geq 0$. Because $\Phi_{\vect{\cM}}$ is linear, with $\rho_\beta$, it is the case that $\alpha(\vect{f}) > \beta \gamma$. With $\beta \rightarrow \infty$, $\alpha(\vect{f}) \rightarrow \infty$. Thus, \Cref{eq:algorithm2_dual_set_indicator_relation} holds.
	\end{proof}

To derive the dual problem of \Cref{eq:algorithm2_primal}, we first compute the Fenchel conjugate of two functions that will become useful.
\begin{lemma}\label{lemma:conjugate_of_s}
	For $c_0\geq 0$ and $c_1 >0$, let $s: \mbR \rightarrow \mbR\  \cup \{+\infty\}$ be defined as
	\begin{aeq}
	s(x) = \begin{cases} 
		-\sqrt{c_0^2 - \abs{x}^2/c_1^2}  & \text{ if } \abs{x} \leq c_0 c_1,\\
		\infty & \text{ otherwise.}
		\end{cases}
		\end{aeq}The Fenchel conjugate function of $s$ is
	\begin{aeq}\label{eq:fenchel_conjugate_of_s}
		s^*(y) = c_0 \sqrt{1+c_1^2y^2}.
		\end{aeq}
	\end{lemma}
\begin{proof}
	By \Cref{def:fenchel_conjugate}, the conjugate function of $s$ is
	\begin{aeq}
		s^*(y) &= \sup_{x \in \mbR} \{xy - s(x)\}\\
		&=\sup_{x \in \mbR} \{xy + \sqrt{c_0^2 - \abs{x}^2/c_1^2}: \abs{x} \leq c_0c_1\}
		\end{aeq}By a simple calculation to optimize over $x$, it is easy to verify that \Cref{eq:fenchel_conjugate_of_s} is indeed the conjugate function of $s$.\footnote{To find the supremum, one can solve case by case for two cases: $x\geq 0$ and $x < 0$. For each case, one looks for a solution of $x$ where the derivative is zero by simple computation. Then, one verifies that both solutions in these two cases lead to the same optimal value which is indeed the supremum.} 
	\end{proof}
\begin{lemma}\label{lemma:eat_second_order_correction_conjugate}
	For $c_0\geq 0$ and $c_1 >0$, let $\cE(\vect{\lambda}): \mbR^m \rightarrow \mbR  \ \cup \{+\infty\}$ be defined as 
	\begin{aeq}	
		\cE(\vect{\lambda}) = 
	\begin{cases}
		s(\norm{\vect{\lambda}}_1/2) & \text{ if } \sum_{x} \vect{\lambda}(x)=0 \\
	\infty & \text{ otherwise,}
	\end{cases}
\end{aeq}where $s$ is defined in \Cref{lemma:conjugate_of_s}. Then $\cE$ is a convex function and its conjugate function $\cE^*(\vect{f})$ is $\cE^*(\vect{f}) = s^*(\max(\vect{f})-\min(\vect{f}))$, that is, 
	\begin{aeq}
		\cE^*(\vect{f}) =  c_0 \sqrt{1+c_1^2[\max(\vect{f})-\min(\vect{f})]^2}\ .
		\end{aeq}
	\end{lemma}
\begin{proof}
	The convexity of $\cE$ follows from the convexity of $s$ and the convexity of the set $\{\vect{\lambda} \in \mbR^m: \sum_{x} \vect{\lambda}(x) = 0\}$. From \Cref{def:fenchel_conjugate},
	\begin{aeq}
		\cE^*(\vect{f}) &= \sup_{\vect{\lambda}\in \mbR^m} \{\langle \vect{\lambda}, \vect{f}\rangle - \cE(\vect{\lambda})\}\\
		& = \sup_{\vect{\lambda}\in \mbR^m} \{\langle \vect{\lambda}, \vect{f}\rangle - s(\norm{\vect{\lambda}}_1/2):  \sum_{x} \vect{\lambda}(x)=0 \}\\
			& = \sup_{\vect{\lambda}\in \mbR^m} \{\langle  \vect{\lambda}, \vect{f}\rangle	+ \sqrt{c_0^2 - \norm{\vect{\lambda}}_1^2/(4c_1^2)}: \norm{\vect{\lambda}}_1 \leq 2c_0c_1,    \sum_{x} \vect{\lambda}(x)=0 \}
		\end{aeq}Let $\vect{\lambda_+}$ be defined as $\vect{\lambda_+}(x) = \max(\vect{\lambda}(x),0)$ and  $\vect{\lambda_-}$ be defined as $\vect{\lambda_-}(x) = \max(-\vect{\lambda}(x),0)$. Then, it is clear that $\vect{\lambda} = \vect{\lambda_+} - \vect{\lambda_1}$ and $\norm{\vect{\lambda}}_1 =\norm{\vect{\lambda_+}}_1  +\norm{\vect{\lambda_-}}_1$. The condition $\sum_{x}\vect{\lambda}(x)=0$ implies $\norm{\vect{\lambda_+}}_1  = \norm{\vect{\lambda_-}}_1  = \frac{1}{2}\norm{\vect{\lambda}}_1.$ Thus,
	\begin{aeq}
			\cE^*(\vect{f}) &= \sup_{\substack{\vect{\lambda_+}\in \mbR^m\\ \vect{\lambda_-} \in \mbR^m}} \{\langle  \vect{f}, \vect{\lambda_+}\rangle - \langle \vect{f}, \vect{\lambda_-}\rangle 	+ \sqrt{c_0^2 - \norm{\vect{\lambda_+}}_1^2/c_1^2}: \norm{\vect{\lambda_+}}_1 = \norm{\vect{\lambda_-}}_1 \leq c_0c_1, \vect{\lambda_+} \geq 0, \vect{\lambda_-} \geq 0 \}\\
			&=\sup_{\vect{\lambda_+}\in \mbR^m, v \in \mbR} \{  \langle  \vect{f}, \vect{\lambda_+}\rangle - v \min(\vect{f})	+ \sqrt{c_0^2 - v^2/c_1^2}:
			0 \leq v = \norm{\vect{\lambda_+}}_1 \leq c_0c_1, \vect{\lambda_+} \geq 0 \}\\
			& = \sup_{v \in \mbR}  \{  v[\max(\vect{f}) - \min(\vect{f})]	+ \sqrt{c_0^2 - v^2/c_1^2}:
			0 \leq v \leq c_0c_1 \}\\
				& = \sup_{v \in \mbR}  \{  v[\max(\vect{f}) - \min(\vect{f})]	+ \sqrt{c_0^2 - v^2/c_1^2}:
		\abs{v} \leq c_0c_1 \}\\
			&= \sup_{v \in \mbR}  \{  v[\max(\vect{f}) - \min(\vect{f})] - s(v)\}\\
			&= s^*(\max(\vect{f}) - \min(\vect{f})). 
		\end{aeq}In the second line above, we optimize $\vect{\lambda_-}$ and the optimal value is achieved when $\vect{\lambda_-}$ contains a single nonzero entry in the position corresponding to $\min(\vect{f})$. In the third line, we optimize $\vect{\lambda_+}$ and the optimal value is achieved when $\vect{\lambda_+}$ contains a single nonzero entry in the position corresponding to $\max(\vect{f})$. In the fourth line, we drop the constraint $v \geq 0$ since the optimal value is always achieved when $v \geq 0$ due to the fact $\max(\vect{f}) - \min(\vect{f}) \geq 0$. In the fifth line, we use the definition of $s$. The last line follows from the definition of $s^*$.
	\end{proof}
\begin{prop}\label{prop:algorithm2_primal_dual_relation}
	 Let $c_0\geq 0$ and $c_1 >0$ be constants. For the problem
	\begin{aeq}\label{eq:eat_fenchel_primal_problem}
h_{\primal}(\vect{q}_0):=	&\minimize_{\rho} & \; \; W(\rho) + \sqrt{c_0^2 - \norm{\vect{q}_0- \Phi_{\vect{\cM}}(\rho)}_1^2/(4c_1^2)}\\
		&\st & \ \Tr(\rho) = 1 \\
	&	&\	\norm{\vect{q}_0- \Phi_{\vect{\cM}}(\rho)}_1 \leq 2c_0c_1\\
	&	& \ \rho \geq 0, 
		\end{aeq}its Fenchel dual problem is
	\begin{aeq}\label{eq:eat_fenchel_dual_problem}
	h_{\dual}(\vect{q}_0):=&	\maximize_{\vect{f}} & \  \vect{q}_0 \cdot \vect{f} - c_0\sqrt{1+c_1^2[\max(\vect{f})-\min(\vect{f})]^2} \\ 
&\st	&\ \vect{f} \cdot \Phi_{\vect{\cM}}(\rho) \leq W(\rho) \  \forall \rho \in \DM(\cH).
		\end{aeq}Moreover, if there exists $\sigma \in \DM(\cH)$ such that $\Phi_{\vect{\cM}}(\sigma) = \vect{q}_0$, then $h_{\primal}(\vect{q}_0)= 	h_{\dual}(\vect{q}_0)$.
	\end{prop} 
\begin{proof}
	We introduce the indicator function $\delta_{\pos(\cH)}$ for the set $\pos(\cH)$ (see \Cref{eq:indicator_function}). We identify any $\rho \in \pos(\cH)$ as an element in $\mbR^n$ for $n=4\dim(\cH)^2$ by representing $\rho$ with its real and imaginary parts, and then applying an appropriate vectorization of $\rho$. Similarly, we abuse the notation $\Phi_{\vect{\cM}}$ and let it denote the corresponding linear map from $\mbR^n$ to $\mbR^m$ where $m$ is the number of POVM elements for this measurement map. To apply \Cref{thm:fenchel_duality}, we define $f(\rho)= W(\rho) + \delta_{\pos(\cH)}(\rho)$, $g(\vect{\lambda}) = \cE(\vect{\lambda}-\vect{q}_0)$ and let $\Phi_{\vect{\cM}}$ play the role of $A$ in the theorem. We first calculate $f^*$ and $g^*$:
	\begin{align}
		f^*(Z) &= \sup_{\rho} \{\langle Z, \rho \rangle - f(\rho)\} \notag \\
		&= \sup_{\rho} \{\langle Z, \rho \rangle - W(\rho): \rho \in \pos(\cH)\} \\
			g^*(\vect{f}) &=  \sup_{\vect{\lambda}} \{\langle \vect{f}, \vect{\lambda} \rangle -\cE(\vect{\lambda}-\vect{q_0})\}\notag \\
			&=  \sup_{\vect{\lambda'}} \{\langle \vect{f}, \vect{\lambda'}  \rangle + \langle \vect{f}, \vect{q_0} \rangle -\cE(\vect{\lambda'})\}\notag \\
			&= \cE^*(\vect{f}) + \langle  \vect{f},\vect{q_0} \rangle.
		\end{align}

Because $\sum_x \vect{q_0}(x) = 1$ and $\sum_x \Phi_{\vect{\cM}}(\rho) (x)= 1$ for $\Tr(\rho)=1$, one can use the definition of $f$ and $\cE$ to rewrite the problem in \Cref{eq:eat_fenchel_primal_problem} as
	\begin{aeq}
	h_{\primal}(\vect{q_0})&= \inf_{\rho \in \mbR^n} \{ f(\rho) + \cE(\Phi_{\vect{\cM}}(\rho)-\vect{q_0}) \}	\\
	&= \inf_{\rho \in \mbR^n}\{f(\rho)+ g(\Phi_{\vect{\cM}}(\rho))\}.
	\end{aeq}
We then use the indicator function for the set $\cF_\dual$ and the conjugate function $\cE^*$ of $\cE$ to rewrite \Cref{eq:eat_fenchel_dual_problem} as
\begin{aeq}
	h_{\dual}(\vect{q}_0)&=\sup_{\vect{f}}	\{- \langle \vect{q}_0, -\vect{f} \rangle  - \cE^*(\vect{f}) - \delta_{\cF_\dual}(\vect{f})\}\\
	&= \sup_{\vect{f}}	\{-g^*(-\vect{f}) - \delta_{\cF_\dual}(\vect{f})\}.
	\end{aeq}By \Cref{lemma:rewrite_dual_set_indicator_function}, $\delta_{\cF_\dual}(\vect{f})= \sup_{\rho}\{ \langle \vect{f},\Phi_{\vect{\cM}}(\rho) \rangle  - W(\rho): \rho \in \pos(\cH)\}.$ Using this substitution, we find
\begin{aeq}
		h_{\dual}(\vect{q}_0)	&= \sup_{\vect{f} \in \mbR^m}\{-\sup_{\rho}\{\langle \vect{f},\Phi_{\vect{\cM}}(\rho) \rangle -W(\rho): \rho \in \pos(\cH) \}- g^*(-\vect{f}) \}\\
		&= \sup_{\vect{f} \in \mbR^m}\{-\sup_{\rho}\{\langle \Phi_{\vect{\cM}}^{\dagger}(\vect{f}), \rho \rangle -W(\rho): \rho \in \pos(\cH) \}- g^*(-\vect{f}) \}\\
	&= \sup_{\vect{f} \in \mbR^m}\{-f^*(\Phi_{\vect{\cM}}^{\dagger}(\vect{f}))- g^*(-\vect{f})\}.
\end{aeq}Therefore, the optimization problem in \Cref{eq:eat_fenchel_dual_problem} is the dual problem of the optimization problem in \Cref{eq:eat_fenchel_primal_problem}. 

To show $h_{\primal}(\vect{q}_0)= 	h_{\dual}(\vect{q}_0)$, we just need to verify the condition 
\begin{aeq}\label{eq:Fenchel_strong_duality_condition}
\relint(\dom(g)) \cap \Phi_{\vect{\cM}} \relint(\dom(f)) \neq \emptyset.
\end{aeq}From the definition of domain [see \Cref{eq:domain_of_extension}], \begin{aeq}
\dom(g) &= \{\vect{\lambda} \in \mbR^m: \sum_{x}\vect{\lambda}(x)=1, \norm{\vect{\lambda}-\vect{q_0}}_1 \leq 2c_0c_1\}\\
\dom(f) & =  \{\rho \in \mbR^n: \rho \geq 0\}
\end{aeq}It is clear that $ \vect{q}_0 \in \relint(\dom(g))$. Moreover, there exists some $\epsilon >0$ such that $(1-\epsilon)\vect{q}_0 + \epsilon \vect{1}/m \in \relint(\dom(g))$. We now show that $ (1-\epsilon)\vect{q}_0 + \epsilon \vect{1}/m \in \Phi_{\vect{\cM}} \relint(\dom(f))$. By the assumption, there exists $\sigma \in \DM(\cH)$ such that $\Phi_{\vect{\cM}}(\sigma) = \vect{q}_0$.  Then for any $\delta > 0$, it is the case that $\cD_\delta(\sigma) \in \relint(\dom(f))$. In particular, $\cD_\epsilon(\sigma) \in \relint(\dom(f))$. Since $\Phi_{\vect{\cM}}(\cD_\epsilon(\sigma) ) = (1-\epsilon)\vect{q}_0 + \epsilon \vect{1}/m$, it is the case that  $(1-\epsilon)\vect{q}_0 + \epsilon \vect{1}/m \in \Phi_{\vect{\cM}} \relint(\dom(f))$. Therefore, \Cref{eq:Fenchel_strong_duality_condition} holds. We conclude that $h_{\primal}(\vect{q}_0)= h_{\dual}(\vect{q}_0)$.
	\end{proof}
Since we establish that $h_{\primal}(\vect{q}_0)= 	h_{\dual}(\vect{q}_0)$, we can now solve the primal problem in \Cref{eq:eat_fenchel_primal_problem}. The difficulty is that the objective function is not differentiable at points where$\Phi_{\vect{\cM}}(\rho) = \vect{q}_0$. Since many standard solvers handle differentiable functions more effectively, we introduce a slack variable and rewrite the problem to an equivalent one according to the following lemma.
\begin{lemma}
	The following optimization problem
	\begin{aeq}
\tilde{h}_\primal(\vect{q}_0):=&	\minimize_{\rho, \vect{\xi}} \  & \ W(\rho) - \sqrt{c_0^2 - \Big(\sum_x\vect{\xi}(x)\Big)^2/(4c_1^2)} \\
	&\st \ & \ -\vect{\xi} \leq \Phi_{\vect{\cM}}(\rho) - \vect{q}_0 \leq \vect{\xi}\\
	&\ & \ \sum_x \vect{\xi}(x) \leq 2c_0c_1\\
	&\ & \Tr(\rho) = 1\\
	&\ & \ \rho \geq 0, \ \vect{\xi} \in \mbR^{\abs{\cX}}
	\end{aeq}has the same optimal value as the problem in \Cref{eq:eat_fenchel_primal_problem}, that is, $\tilde{h}_\primal(\vect{q}_0)=h_\primal(\vect{q}_0).$
	\end{lemma}
\begin{proof}
		We first notice that the function $- \sqrt{c_0^2 - t^2/(4c_1^2)}$ is a monotonically non-decreasing function of $0 \leq t \leq 2c_0c_1$. For $\vect{\xi} \geq 0, \norm{\vect{\xi}}_1 = \sum_j\vect{\xi}(j).$ The condition  $-\vect{\xi} \leq \Phi_{\vect{\cM}}(\rho) - \vect{q_0} \leq \vect{\xi}$ implies that $\vect{\xi} \geq 0$ and $\norm{\vect{\xi}}_1 \geq \norm{\vect{q}_0- \Phi_{\vect{\cM}}(\rho)}_1$. Thus, it is clear that $\tilde{h}_\primal(\vect{q}_0) \geq h_\primal(\vect{q}_0)$. Let $\rho^\star$ be an optimal solution for $h_\primal(\vect{q}_0)$ and let $\vect{\xi}^\star(j) := \abs{\vect{q}_0(j)- \Phi_{\vect{\cM}}(\rho^\star)(j)}$. Then $\norm{\vect{\xi}^\star}_1 = \norm{\vect{q}_0- \Phi_{\vect{\cM}}(\rho^\star)}_1 \leq 2c_0c_1.$ This implies, $(\rho^\star, \vect{\xi}^\star)$ is a feasible solution for $\tilde{h}_\primal(\vect{q}_0)$ that achieves the same objective function value as $h_\primal(\vect{q}_0)$. Therefore, $\tilde{h}_\primal(\vect{q}_0) \leq h_\primal(\vect{q}_0)$. We conclude that $\tilde{h}_\primal(\vect{q}_0)=h_\primal(\vect{q}_0).$
	\end{proof}

\end{document}